\documentclass[runningheads]{llncs}
\usepackage[T1]{fontenc}

\usepackage{amsmath}
\usepackage{amssymb}
\usepackage[english]{babel}
\usepackage[utf8]{inputenc}
\usepackage{float}
\usepackage{tabularx}
\usepackage{hyperref}
\usepackage{nameref,cleveref}
\usepackage{multicol}
\usepackage{wrapfig}

\usepackage[inline]{enumitem}

\usepackage[backend=biber, maxcitenames=2, maxbibnames=99, style=numeric]{biblatex} 
\AtEveryBibitem{
    \clearfield{urlyear}
    \clearfield{urlmonth}
}

\usepackage{aligned-overset}

\allowdisplaybreaks
\usepackage{multirow}
\usepackage{todonotes}
\usepackage{makecell}
\usepackage{booktabs}
\usepackage{subcaption}
\usepackage{thmtools,thm-restate}
\usepackage{tikz}
\usepackage{tcolorbox}
\usepackage{stmaryrd}
\usepackage{subcaption}

\usepackage{aligned-overset}
\usepackage{apptools}
\AtAppendix{\counterwithin{lemma}{section}}

\usepackage{aligned-overset}

\usepackage{tikz}
\usetikzlibrary{positioning,arrows,automata,calc,shapes}
\definecolor{darkgreen}{rgb}{0.0, 0.5, 0.0}
\definecolor{darkred}{rgb}{0.9, 0.0, 0.0}
\definecolor{darkblue}{rgb}{0.0, 0.0, 0.9}

\tikzset{
state/.style={rounded rectangle,draw=black,inner sep=1.1mm,minimum width=5mm,minimum height=3mm},
bullet/.style={circle,draw=black,fill=black,inner sep=0cm, minimum size=0.5mm},
ptran/.style={rounded corners, ->,>=stealth',auto},
}

\usetikzlibrary{shapes.geometric, arrows}
\usetikzlibrary{decorations.pathreplacing}
\usetikzlibrary{calc}
\usetikzlibrary{patterns}

\tikzstyle{product} = [rectangle, rounded corners, 
minimum width=3cm, 
minimum height=1cm,
text centered, 
text width=3cm, 
draw=black, 
fill=gray!70]

\tikzstyle{certificate} = [ellipse, 
minimum width=3cm, 
minimum height=1cm,
text centered, 
text width=3cm, 
draw=black, 
fill=gray!70]

\tikzstyle{io} = [trapezium, 
trapezium stretches=true, %
trapezium left angle=70, 
trapezium right angle=110, 
minimum width=3cm, 
minimum height=1cm, text centered, 
draw=black, fill=blue!30]

\tikzstyle{process} = [rectangle, 
minimum width=3cm, 
minimum height=1cm, 
text centered, 
text width=3cm, 
draw=black, 
fill=gray!20]

\tikzstyle{decision} = [diamond, 
minimum width=3cm, 
minimum height=1cm, 
text centered, 
draw=black, 
fill=green!30]
\tikzstyle{arrow} = [thick,->,>=stealth]

\newcommand{\union}{\cup}

\newcommand{\prob}{\mathsf{Pr}}
\newcommand{\vect}[1]{\mathbf{#1}}
\newcommand{\reals}{\mathbb{R}}
\newcommand{\realsnn}{\reals_{\geq 0}}
\DeclareMathOperator{\suppOp}{supp}
\newcommand{\supp}[2][]{\suppOp_{#1}(#2)}
\DeclareMathOperator{\stateSuppOp}{state-supp}
\newcommand{\stateSupp}[2][]{\stateSuppOp_{#1}(#2)}

\newcommand{\rationals}{\mathbb{Q}}
\newcommand{\card}[1]{{\lvert {#1} \rvert}}
\DeclareMathOperator{\distrOp}{Distr}
\newcommand{\distr}[1]{\distrOp(#1)}
\newcommand{\compBowtie}{\mathbin{\overline{\bowtie}}}

\newcommand{\plh}{\mathcal{F}}
\newcommand{\dqPoly}[2]{\plh_{#1, #2}}
\newcommand{\cqPoly}[2]{\mathcal{H}_{#1, #2}}
\newcommand{\boldLambda}{\pmb{\lambda}}
\newcommand{\boldMu}{\pmb{\mu}}
\newcommand{\boldDelta}{\pmb{\delta}}
\newcommand{\boldGamma}{\pmb{\gamma}}

\newcommand{\expectation}[3][]{\mathbb{E}_{#1}^{#2}[#3]}

\newcommand{\Conj}{\bigwedge}
\newcommand{\Disj}{\bigvee}

\newcommand{\eventually}{\lozenge}
\newcommand{\globally}{\square}

\newcommand{\targetSet}{G}
\newcommand{\mdp}{\mathcal{M}}
\newcommand{\dtmc}{\mathcal{D}}
\newcommand{\states}{S}
\newcommand{\state}{s}
\newcommand{\targets}{F}
\newcommand{\target}{f}
\newcommand{\actions}{Act}
\newcommand{\action}{a}
\newcommand{\init}{s_{in}}
\newcommand{\transMat}{\vect{P}}
\newcommand{\exit}{\bot}

\newcommand{\mdpRF}{(\states \union \targets, \actions, \init, \transMat)}
\newcommand{\SM}{\vect{A}}
\newcommand{\SA}{\mathcal{E}}
\newcommand{\TM}{\vect{T}}
\newcommand{\scheduler}{\sigma}
\newcommand{\schedulers}{\Sigma}
\newcommand{\mSchedulers}{\Sigma_{\textsf{M}}}

\DeclareMathOperator{\pathsOp}{Paths}
\DeclareMathOperator{\lastOp}{last}
\newcommand{\last}[1]{\lastOp({#1})}
\newcommand{\paths}{\pathsOp}
\newcommand{\pathsFin}{\pathsOp_{\mathsf{fin}}}
\newcommand{\Path}{\pi}
\newcommand{\initDistr}{\pmb{\delta}_{in}}

\newcommand{\freq}{\mathsf{freq}}

\newcommand{\mec}{C}
\newcommand{\allMecStates}{\states_{\mathsf{MEC}}}
\newcommand{\MECS}{\mathsf{MEC}}
\newcommand{\quotientState}{\state}

\newcommand{\mems}{\mathsf{M}}
\newcommand{\mem}{m}

\newcommand{\universalDQ}{\ensuremath{(\forall, \lor)}}
\newcommand{\universalCQ}{\ensuremath{(\forall, \land)}}
\newcommand{\existsDQ}{\ensuremath{(\exists, \lor)}}
\newcommand{\existsCQ}{\ensuremath{(\exists, \land)}}

\newcommand{\query}{\Psi}
\newcommand{\prop}[2]{\phi_{#1}^{#2}}
\newcommand{\propAlt}[2]{\psi_{#1}^{#2}}
\newcommand{\automaton}{\mathcal{A}}

\newcommand{\Reach}{\textsf{Reach}}
\newcommand{\ReachInv}{\textsf{ReachInv}}
\newcommand{\MP}{\textsf{MP}}

\DeclareMathOperator{\lrOp}{\mathsf{MP}}
\newcommand{\lrInf}[1]{\underline{\lrOp}(#1)}
\newcommand{\lrSup}[1]{\overline{\lrOp}(#1)}

\usepackage{graphicx}
\begin{document}
\title{Certificates and Witnesses for Multi-Objective Queries in Markov Decision Processes
\thanks{The authors were supported by the German Federal Ministry of Education and Research (BMBF) within the project SEMECO Q1 (03ZU1210AG) and by the German Research Foundation (DFG) through the Cluster of Excellence EXC 2050/1 (CeTI, project ID 390696704, as part of Germany’s Excellence Strategy) and the DFG Grant 389792660 as part of TRR 248 (Foundations of Perspicuous Software System).}}

\titlerunning{Certificates and Witnesses for Multi-Objective Queries} 

\author{Christel Baier\orcidID{0000-0002-5321-9343} \and Calvin Chau\orcidID{0000-0002-3437-0240} \and Sascha Klüppelholz\orcidID{0000-0003-1724-2586}}

\authorrunning{Baier et al.}

\institute{
Technische Universität Dresden, Dresden, Germany\\
\email{\{christel.baier, calvin.chau,\\ sascha.klueppelholz\}@tu-dresden.de}
}

\maketitle

\begin{abstract}
Certifying verification algorithms not only return whether a given property holds or not, but also provide an accompanying independently checkable certificate and a corresponding witness. The certificate can be used to easily validate the correctness of the result and the witness provides useful diagnostic information, e.g.\ for debugging purposes. Thus, certificates and witnesses substantially increase the trustworthiness and understandability of the verification process. In this work, we consider certificates and witnesses for \emph{multi-objective reachability-invariant} and \emph{mean-payoff} queries in Markov decision processes, that is conjunctions or disjunctions either of reachability and invariant or mean-payoff predicates, both universally and existentially quantified. Thereby, we generalize previous works on certificates and witnesses for single reachability and invariant constraints. To this end, we turn known linear programming techniques into certifying algorithms and show that witnesses in the form of schedulers and subsystems can be obtained. As a proof-of-concept, we report on implementations of certifying verification algorithms and experimental results.
\end{abstract}

\section{Introduction}
Probabilistic model checking (PMC) is a technique for analysing and formally verifying probabilistic models, inter alia, aiming to enable higher trustworthiness of correctness of systems. However, PMC tools have been observed to contain bugs themselves \cite{wimmer_demand_2008}, thereby diminishing trust in the verification results. The paradigm of \emph{certifying algorithms} \cite{kratsch_certifying_2006, mcconnell_certifying_2011} is a well-accepted way of addressing this issue. Instead of solely returning a result, a certifying algorithm is required to also provide an accompanying \emph{certificate}, which can be used to \emph{easily} check the correctness of the result in a mathematically rigorous manner. There is a plethora of certifying algorithms \cite{mehlhorn_checking_1999, mcconnell_certifying_2011, kratsch_certifying_2006} and certifying \emph{verification algorithms} \cite{debbi_counterexamples_2018, namjoshi_certifying_2001, peled_falsification_2001, cousot_induction_1982, kupferman_certifying_2021, kupferman_certifying_2021-1, kupferman_complementation_2005}.

Most relevant for this paper are the existing certification and explication techniques for probability or expectation constraints in Markovian models.
Early work towards the explication of PMC results introduced probabilistic counterexamples as sets of paths (see e.g. \cite{aljazzar_generation_2009, han_diagnosis_2009, han_counterexample_2009}) which however tend to be huge. This motivated the generation of more concise explications, including the generation of fault-trees from probabilistic counterexamples \cite{kuntz_probabilistic_2011}, causality-based explanations \cite{leitner-fischer_synergy_2013} and the concept of \emph{witnessing subsystems} \cite{jansen_hierarchical_2011, wimmer_minimal_2012, jansen_counterexamples_2015, quatmann_counterexamples_2015, funke_farkas_2020, jantsch_certificates_2022}. Witnessing subsystems are parts of a system that demonstrate the satisfaction of a property and provide useful insights into why a property is violated or satisfied.

\emph{Multi-objective queries} are existentially or universally quantified disjunctions or conjunctions of either multiple invariant and reachability or mean-payoff predicates, e.g.\ ``Is it possible to reach the goal with probability at least 0.9 and encounter an error with probability at most 0.2?''. Thus, in many settings they are useful for reasoning about multiple conflicting goals \cite{etessami_multi-objective_2008}. Reachability probabilities and expected mean-payoffs in Markov decision processes (MDPs) can be characterized as linear programs (LP), extensively studied in \cite{kallenberg_linear_1983, kallenberg_survey_1994}. The techniques for verifying existentially quantified multi-objective reachability \cite{etessami_multi-objective_2008, forejt_quantitative_2011} and multi-objective mean-payoff queries \cite{brzdil_two_2011, brazdil_markov_2014} are also based on LP characterizations. However, the authors have not considered the solutions of the LP in the context of certificates nor have witnesses in the form of subsystems been addressed. In \cite{forejt_pareto_2012} and its extension to mean-payoff \cite{quatmann_multi-objective_2021, quatmann_verification_2023}, the certificates for universally quantified queries are only implicitly considered and the connection to subsystems has not been studied. The verification of multi-objective queries is supported by \textsc{Prism} \cite{kwiatkowska_prism_2011}, \textsc{Multigain} \cite{brazdil_multigain_2015} and \textsc{Storm} \cite{hensel_probabilistic_2022}.

The work of \cite{funke_farkas_2020, jantsch_certificates_2022} considers certificates and witnesses based on Farkas lemma' and the LP characterizations \cite{kallenberg_linear_1983, kallenberg_survey_1994}, referred to as \emph{Farkas certificates}. The techniques for finding certificates and minimal witnessing subsystems are implemented in the tool \textsc{Switss} \cite{jantsch_switss_2020}. However, certificates and witnesses have only been considered for \emph{single reachability and invariant probabilities}. Further, the computation of subsystems for invariants is not supported by \textsc{Switss}.

The purpose of this paper is to study \emph{certificates and witnesses} in the context of \emph{multi-objective queries} in MDPs. Building on the characterization considered in \cite{etessami_multi-objective_2008, kallenberg_linear_1983, brzdil_two_2011, brazdil_markov_2014}, we derive certificates using Farkas' lemma and show that they can be used to identify \emph{witnesses}, both in the form of schedulers and subsystems, generalizing the results from \cite{funke_farkas_2020, jantsch_certificates_2022}. In particular, we show how to devise witnesses in the form of schedulers with \emph{stochastic memory updates} as introduced in \cite{brazdil_markov_2014}. Lastly, we present an implementation of our techniques and experimentally evaluate it on several benchmarks. All omitted proofs are in the Appendix.

\medskip

\noindent\textbf{Contributions.}
\begin{itemize}
\item We present the foundations of \emph{Farkas certificates} for \emph{existentially} and \emph{universally} quantified multi-objective \emph{reachability-invariant} (\Cref{section:farkas-and-witnesses}) and \emph{mean-payoff} queries (\Cref{section:mean-payoff}) in MDPs.
\item Farkas certificates for multi-objective queries are shown to have a direct correspondence to witnessing subsystems and enable the computation of minimal witnessing subsystems. We hereby generalize prior work \cite{funke_farkas_2020, jantsch_certificates_2022} on single-objective reachability and invariant constraints.
\item We show that witnesses in the form of schedulers can first be computed for the MEC quotient \cite{de_alfaro_formal_1997} (see \Cref{section:preliminaries}) and then transferred to the underlying MDP, using schedulers with \emph{stochastic memory updates} to traverse the end components of the MDP.
\item An implementation of our techniques with experimental results on several case studies is presented.
\end{itemize}

\section{Preliminaries}
\label{section:preliminaries}
\noindent\textbf{Notation and Farkas' lemma.}
We write $[k]$ to denote the set $\{1, \dots, k\}$. Let $\states = \{\state_0, \dots, \state_n \}$ be a finite set. In this work, vectors and matrices are written in boldface, e.g. $\vect{x}$ and $\SM$. Instead of writing $\vect{x} \in \reals^\card{\states}$, we write $\vect{x} \in \reals^\states$ and $\vect{x}(\state_i)$ to denote the $i$th entry of $\vect{x}$. Matrices are treated similarly. The support of a vector $\vect{x}$ is defined as $\supp{\vect{x}} = \{\state \in \states \mid \vect{x}(\state) > 0\}$. Throughout this work we consider $\bowtie \ \in \{<, \leq, >, \geq\}$, $\gtrsim \ \in \{>, \geq\}$ and $\lesssim \ \in \{<, \leq \}$. \emph{Farkas' lemma} is a fundamental result of linear algebra and linear programming. Essentially, it relates the solvability of a linear system with the unsolvability of another one.
\begin{lemma}[Farkas' lemma (e.g. Proposition 6.4.3 in \cite{matousek_understanding_2007})]
\label{lemma:farkas}
\newline Let $\vect{A} \in \reals^{m \times n}$ and $\vect{b} \in \reals^{m}$, then the following holds:
\begin{enumerate}[label={(\roman*)}, itemsep=0mm, topsep=0.8mm, parsep=0mm]
\item $\exists \vect{x} \in \realsnn^n \centerdot \vect{A} \vect{x} \leq \vect{b} \iff \neg \exists \vect{y} \in \realsnn^m \centerdot \vect{A}^\top \vect{y} \geq 0 \land \vect{b}^\top \vect{y} < 0$ \label{lemma:farkas-1}
\item $\exists \vect{x} \in \reals^n \centerdot \vect{A} \vect{x} = \vect{b} \iff \neg \exists \vect{y}\in \reals^m \centerdot \vect{A}^\top \vect{y} = 0 \land \vect{b}^\top \vect{y} \neq 0$ \label{lemma:farkas-2}
\end{enumerate}
\end{lemma}
\noindent\textbf{Markov decision processes.}
A \emph{Markov decision process} (MDP) \cite{puterman_markov_1994} $\mdp$ is a tuple $(\states, \actions, \boldDelta, \transMat)$ where $\states$ is a finite set of \emph{states}, $\actions$ a finite set of \emph{actions}, $\boldDelta \in [0, 1]^\states$ an \emph{initial distribution} and $\transMat \colon \states \times \actions \to \distr{S}$ a \emph{partial transition function}, where $\distr{S}$ denotes the set of all probability distributions over $S$. We often write $\transMat(\state, \action, \state')$ instead of $\transMat(\state, \action)(\state')$.
A state-action pair $(\state, \action) \in \states \times \actions$ is said to be \emph{enabled} if $\transMat(\state, \action)$ is defined and we write $\SA_\mdp$ to denote the set of all enabled pairs. The set of enabled actions in $\state$ is defined by $\actions(\state) = \{ \action \in \actions \mid (\state, \action) \in \SA_\mdp\}$. A state $\state$ is \emph{absorbing} if $\transMat(\state, \action, \state) = 1$ for all $\action \in \actions(\state)$. We write $(\states, \actions, \init, \transMat)$ for the MDP $(\states, \actions, \initDistr, \transMat)$ where $\initDistr$ is \emph{Dirac} in $\init$. A \emph{path} $\Path$ in $\mdp$ is a sequence $\Path = \state_0 \action_0 \state_1 \action_1 \ldots$ where $\transMat(\state_i, \action_i, \state_{i+1}) > 0$ for all $i$. $\last{\Path}$ refers to the last state of a finite path $\Path$ and $\paths(\mdp)$ ($\pathsFin(\mdp)$) is the set of all infinite (finite) paths starting in $\init$.

An \emph{end component} (EC) of $\mdp$ is a set $\emptyset \subset \mec \subseteq \SA_\mdp$ such that the induced sub-MDP is strongly connected. The states of $\mec$ are denoted with $\states(\mec)$ and we refer to $(\state, \action) \in \mec$ as \emph{internal}. An EC $\mec$ is a \emph{maximal} EC (MEC) if there is no another EC $\mec'$ such that $\mec \subset \mec'$. $\MECS(\mdp)$ denotes the set of MECs of $\mdp$ and $\allMecStates \subseteq \states$ the set of states contained in a MEC. We consider the \emph{MEC quotient} from \cite{baier_foundations_2022}, akin to the quotient from \cite{de_alfaro_formal_1997}. W.l.o.g.\ we assume that the actions of the states are pairwise disjoint. The MEC quotient of $\mdp$ is then given by $\hat{\mdp} = (\hat{\states} \union \{ \exit_\mec \mid \mec \in \MECS(\mdp)\}, \actions \union \{\tau\}, \hat{\state}_{in}, \hat{\transMat})$ where $\hat{\states} = (\states \union \{\quotientState_\mec \mid \mec \in \MECS(\mdp) \}) \setminus \allMecStates$. We define $\iota \colon \states \to \hat{\states}$ as $\iota(\state) = \state$ if $\state \in \states \setminus \allMecStates$ and $\iota(\state) = \quotientState_\mec$ if $\state \in \states(\mec)$. The initial state is given by $\hat{\state}_{in} = \iota(\init)$. For states $\state \in \states \setminus \allMecStates$ we define $\hat{\transMat}(\state, \action, \state') = \transMat(\state, \action, \iota^{-1}(\state'))$ for all $\state' \in \hat{\states}$. For all MECs $\mec$, $\state \in \states(\mec)$ and $\action \in \actions(\state)$, we set $\hat{\transMat}(\quotientState_\mec, \action, \state') = \transMat(\state, \action, \iota^{-1}(\state'))$ for all $\state' \in \hat{\states}$ if $(\state, \action) \not \in \mec$ and set $\hat{\transMat}(\quotientState_\mec, \tau, \exit_\mec) = 1$, i.e. taking $\tau$ corresponds to staying in $\mec$ forever.

We consider a \emph{discrete-time Markov chain} (DTMC) $\dtmc$ to be an MDP with a single action that is enabled in all states. Thus, we omit the actions when speaking of paths in DTMCs and write $(\states, \boldDelta, \transMat)$ instead of $(\states, \actions, \boldDelta, \transMat)$.

\medskip

\noindent\textbf{Schedulers and probability measure.} A \emph{scheduler} $\scheduler$ maps a finite path in an MDP $\mdp$ to a distribution over the available actions, i.e. $\scheduler \colon \pathsFin(\mdp) \to \distr{\actions}$ with $\supp{\scheduler(\Path)} \subseteq \actions(\last{\Path})$, and is \emph{memoryless} if it can be seen as a function of the form $\scheduler \colon \states \to \distr{\actions}$. Let $\schedulers^\mdp$ and $\mSchedulers^\mdp$ denote the set of unrestricted and memoryless schedulers of $\mdp$. A scheduler $\scheduler$ can also be represented as a tuple $(\alpha_\mathsf{update}, \alpha_\mathsf{next}, \mems, \boldDelta_\mems)$ where $\mems$ is a set of memory locations\footnote{Infinitely many memory locations might be needed.}, $\boldDelta_\mems \in \distr{\mems}$ an initial memory distribution, $\alpha_{\mathsf{update}} \colon \actions \times \states \times \mems \to \distr{\mems}$ a \emph{stochastic memory update} and $\alpha_{\mathsf{next}} \colon \states \times \mems \to \distr{\actions}$ the \emph{next move} function \cite{brazdil_markov_2014, randour_percentile_2015}. The update $\alpha_\mathsf{update}$ takes an action $\action$ that has \emph{lead to} state $\state$ and current memory location $\mem$ to update the memory location. Depending on the current location $\mem$, $\alpha_{\mathsf{next}}$ schedules the available actions in $\state$.

We consider the standard probability measure $\prob_{\mdp}^\scheduler$ \cite{baier_principles_2008}. For $\targetSet \subseteq \states$, we write $\prob_{\mdp, \state}^\scheduler(\eventually \targetSet)$ and $\prob_{\mdp, \state}^\scheduler(\globally \targetSet)$ to denote the probability of eventually reaching $\targetSet$ and only visiting $\targetSet$ under $\scheduler$ when starting in $\state$, respectively. We define $\freq_\mdp^\scheduler(\state, \action) = \sum_{t=0}^{\infty} \prob_{\mdp, \state}^\scheduler\{\state_0 \action_0 \state_1 \action_1 \ldots \mid (\state_t, \action_t)= (\state, \action)\}$ for all $\state \in \states$ and $\action \in \actions(\state)$ if the value exists \cite{baier_foundations_2022}. $\freq_\mdp^\scheduler(\state, \action)$ describes the \emph{expected frequency} of playing state action pair $(\state, \action)$ under $\scheduler$.
For a given \emph{reward vector} $\vect{r}  \in \rationals^{\SA_\mdp}$ and a path $\Path = \state_0 \action_0 \state_1 \action_1 \ldots \in \paths(\mdp)$, the \emph{mean-payoff} is defined as $\lrInf{\vect{r}}(\Path) = \liminf_{n \to \infty} \frac{1}{n} \sum_{t=0}^{n-1} \vect{r}(\state_t, \action_t)$ and $\lrSup{\vect{r}} \coloneqq -\lrInf{-\vect{r}}$. The \emph{expected mean-payoff} is then defined as $\expectation[\mdp, \state]{\scheduler}{\lrInf{\vect{r}}}$ for $\state \in \states$ and $\scheduler \in \schedulers_\mdp$. Whenever we omit the subscript $\state$, we refer to the probability and expectation in $\init$. %

\medskip

\noindent\textbf{Subsystems.} A \emph{subsystem} of an MDP $\mdp = (\states, \actions, \init, \transMat)$ is an MDP $\mdp' = (\states' \union \{\exit\}, \actions, \init, \transMat')$ if $\init \in \states' \subseteq \states$, $\exit$ is absorbing and for all $\state, \state' \in \states'$ and $\action \in \actions$ we have $\transMat'(\state, \action, \state') \in \{0, \transMat(\state, \action, \state')\}$ \cite{jantsch_certificates_2022}. Further, an action $\action$ is enabled in $\state$ in $\mdp'$ if and only if $\action$ is enabled in $\state$ in $\mdp$. Additionally, for reward vectors $\vect{r} \in \rationals^{\SA_\mdp}$ for $\mdp$, we consider the corresponding reward vector $\vect{r}' \in \rationals^{\SA_{\mdp'}}$ where $\vect{r}'(\state, \action) = \vect{r}(\state, \action)$ for all $\state \in \states'$ and $\action \in \actions(\state)$ and $\vect{r}'(\exit, \action) = \min_{(\state', \action') \in \SA} \vect{r}(\state', \action')$ for all $\action \in \actions$. Intuitively, once $\exit$ is reached the smallest possible reward is collected. A subsystem $\mdp_{\states'}$ is said to be \emph{induced} by a set $\states' \subseteq \states$ if it consists of the states $\states' \union \{\exit\}$ and all transitions leading to $\states \setminus \states'$ are redirected to $\exit$ \cite{jantsch_certificates_2022}. More precisely, for all $\state, \state' \in \states'$ and $\action \in \actions$ we have $\transMat'(\state, \action, \state') = \transMat(\state, \action, \state')$ and $\transMat'(\state, \action, \exit) = \sum_{\state' \in \states \setminus \states'} \transMat(\state, \action, \state')$. 

\medskip

\noindent\textbf{Multi-objective queries.}
A \emph{reachability}, \emph{invariant} or \emph{mean-payoff predicate} is an expression of the form 
$\prob_{\mdp}^\scheduler(\eventually \targetSet) {\bowtie} \lambda$, $\prob_{\mdp}^\scheduler(\globally \targetSet) {\bowtie} \lambda$ or $\expectation[\mdp]{\scheduler}{\lrInf{\vect{r}}} \allowbreak {\bowtie} \lambda$
where $\lambda \in \reals$. A \emph{multi-objective reachability, reachability-invariant or mean-payoff property} $\prop{}{\scheduler}(\boldLambda)$ is then a conjunction or disjunction of reachability, reachability and invariant or mean-payoff predicates where $\boldLambda = (\lambda_1, \dots, \lambda_k)^\top$ contains the bounds. We refer to the former as \emph{conjunctive} and the latter as \emph{disjunctive} property and write $\prop{\bowtie}{\scheduler}$ to refer to a property where all predicates have $\bowtie$ as comparison operator. A \emph{multi-objective query} $\query$ is then an existentially or universally quantified property, i.e. $\exists \scheduler {\in} \schedulers \centerdot \prop{}{\scheduler}(\boldLambda)$ or $\forall \scheduler {\in} \schedulers \centerdot \prop{}{\scheduler}(\boldLambda)$. We distinguish between reachability (\Reach), reachability-invariant (\ReachInv) and mean-payoff (\MP) queries and use the quantifier and logical connective to refer to the query type, e.g. $\existsCQ$ to refer to existentially-quantified conjunctive queries.

We omit the super- and subscript $\mdp$ and term ``multi-objective'' whenever it is clear from the context.

\section{Certificates and Witnesses for \ReachInv-Queries}
\label{section:farkas-and-witnesses}
Now we consider certificates and witnesses for \ReachInv-queries. An overview of our approach is shown in \Cref{fig:overview-approach}. 
\begin{figure}[!t]
\centering
\scalebox{0.6}{
\begin{tikzpicture}[node distance=1.8cm]
\node (mdp1) [process] {MDP $\mathcal{N}$\\ \ReachInv-query $\query_{\mathcal{N}}$};
\node (mdp2) [process, right of=mdp1, xshift=2cm] {Product $\mdp = \mathcal{N} {\times} \mathcal{A}$\\ \ReachInv-query $\query_\mdp$};
\node (mdp3) [process, right of=mdp2, xshift=2cm] {MEC Quotient $\hat{\mdp}$\\ \Reach-query $\query_{\hat{\mdp}}$};
\node (certificate) [product, below right of=mdp3, xshift=2.5cm, yshift=0.4cm] {Farkas certificates for $\hat{\mdp} \models \query_{\hat{\mdp}}$};
\node (witness3) [product, below left of=certificate, xshift=-2.5cm, yshift=0.4cm] {Witnesses for\\$\hat{\mdp} \models \query_{\hat{\mdp}}$};
\node (witness2) [product, left of=witness3, xshift=-2cm] {Witnesses for\\$\mdp \models \query_\mdp$};
\node (witness1) [product, left of=witness2, xshift=-2cm] {Witnesses for\\$\mathcal{N} \models \query_\mathcal{N}$};

\node[below of=certificate, yshift=0.65cm] (labelFarkas) {\Cref{subsection:farkas-certificates}};
\draw[arrow, thin] (labelFarkas) -- +(0, 1.66em);
\node[below of=witness2, yshift=0.65cm] (labelTransfer) {\Cref{subsection:transfer-witnesses}};
\node[below of=witness3, yshift=0.65cm] (labelWitness) {\Cref{subsection:farkas-certificates}};
\draw[arrow, thin] (labelWitness) -- +(0, 1.66em);

\draw [arrow] (mdp1) -- (mdp2);
\draw [arrow] (mdp2) -- (mdp3);
\draw [arrow, dotted] (mdp3) to[bend left=15] (certificate);
\draw [arrow, dotted] (certificate) to[bend left=15] node (CertificateToWitness) [midway, below] {} (witness3);
\draw [arrow, dashed] (witness3) -- (witness2) node (QuotientToProduct) [midway, below] {};
\draw [arrow, dashed] (witness2) -- (witness1) node (ProductToMdp) [midway, below] {};

\draw [arrow, thin] (labelTransfer.east) -| (QuotientToProduct.south);
\draw [arrow, thin] (labelTransfer.west) -| (ProductToMdp.south);
\draw[arrow, thin] (labelTransfer) -- +(0, 1.66em);
\draw[arrow, thin] (labelWitness) -| (CertificateToWitness);
\end{tikzpicture}
}
\caption{Overview of our approach for \ReachInv-queries.}
\label{fig:overview-approach}
\end{figure}
We start from an arbitrary MDP $\mathcal{N}$ and a \ReachInv-query $\query_{\mathcal{N}}$ containing only lower bounds. Note that every \ReachInv-query can be rephrased to a \ReachInv-query containing only lower bounds\footnote{$\prob^\scheduler(\eventually \targetSet) {\lesssim} \lambda \Leftrightarrow \prob^\scheduler(\globally (\states\setminus\targetSet)) {\gtrsim} 1{-}\lambda$ and $\prob^\scheduler(\globally \targetSet) {\lesssim} \lambda \Leftrightarrow \prob^\scheduler(\eventually (\states\setminus\targetSet)) {\gtrsim} 1{-}\lambda$}. Then, we construct the \emph{product MDP} $\mdp {=} \mathcal{N} {\times} \automaton$ and corresponding \ReachInv-query $\query_\mdp$. The automaton $\automaton$ keeps track of the state sets that have already been visited (see e.g. \cite[Proposition~2]{forejt_pareto_2012}). This is necessary because schedulers generally require exponential memory for such queries \cite{randour_percentile_2015}. 
Motivated by the fact that the computation of the MECs can be made \emph{certifying} \cite{jantsch_certificates_2022}, we then consider the MEC quotient $\hat{\mdp}$. Crucially, this allows us to rephrase $\query_{\mathcal{N}}$ to a \Reach-query $\query_{\hat{\mathcal{M}}}$ containing \emph{only lower bounds}. More precisely, invariant predicates $\prob_{\mathcal{N}}^\scheduler(\globally \targetSet) {\gtrsim} \lambda$ can be rephrased to reachability predicates of the form $\prob_{\hat{\mdp}}^\scheduler(\eventually T) {\gtrsim} \lambda$ where $T \subseteq \{\exit_1, \dots, \exit_\ell\}$. Note that the quotient $\hat{\mdp}$ is in \emph{reachability form} \cite{jantsch_certificates_2022}, i.e.\ its target states are absorbing. This allows us to restrict our attention to ``simple'' certificates for \Reach-queries and MDPs in reachability form, instead of tackling certificates for $\mathcal{N}$ and $\query_{\mathcal{N}}$ directly. The reduction from $\mathcal{N}$ to $\hat{\mdp}$ (upper part in \Cref{fig:overview-approach}) uses well-known methods from the literature. Since the reduction is simple and can be made certifying, we use the certificates for $\hat{\mdp}$ to act as certificates for $\mathcal{N}$ and $\mdp$. We detail the reduction in \Cref{appendix:transfer}.

In \Cref{subsection:farkas-certificates} we then derive \emph{Farkas certificates} for \Reach-queries and MDPs in reachability form from known techniques for multi-objective model checking \cite{etessami_multi-objective_2008, forejt_pareto_2012}, which have not been considered through the lens of certifying algorithms yet. We make the certificates explicit and show that they yield \emph{witnesses} for $\hat{\mdp}$, both in the form of \emph{witnessing schedulers} and \emph{witnessing subsystems}. Conceptually, a scheduler describes how to control the MDP, whereas a subsystem highlights relevant parts of the MDP. Depending on the use case, one or the other may be more appropriate, and we enable a more flexible perspective.

Lastly, we present new techniques for transferring witnesses from $\hat{\mdp}$ to $\mathcal{N}$ in \Cref{subsection:transfer-witnesses} (lower part in \Cref{fig:overview-approach}). We discuss how schedulers and subsystems for $\mathcal{N}$ can be constructed from their respective counterpart in $\hat{\mdp}$. If $\mdp$ contains many large MECs, this allows us to tackle each MEC individually, resulting in smaller and potentially more tractable subproblems.

\subsection{Farkas Certificates and Witnesses for \Reach-Queries}
\label{subsection:farkas-certificates}
For the remainder of this subsection, we fix an MDP $\mdp = \mdpRF$ with absorbing states $\targets$ and consider reachability properties $\prop{\bowtie}{\scheduler}(\boldLambda)$ with predicates $\allowbreak\prob^\scheduler(\eventually \targetSet_1)\allowbreak{\bowtie} \lambda_1,\allowbreak \dots,\allowbreak \prob^\scheduler(\eventually \targetSet_k) {\bowtie} \lambda_k$ where $\targetSet_1, \dots, \targetSet_k \subseteq \targets$. W.l.o.g.\ we assume that for every $\state \in \states$ there exists $\scheduler$ such that $\prob^\scheduler_\state(\eventually \targets) > 0$ \cite{etessami_multi-objective_2008}. We say that $\mdp$ is in \emph{reachability form} \cite{jantsch_certificates_2022} and exclude state-action pairs of states in $F$ from $\SA$, i.e. $\SA = \{(\state, \action) \mid \state \in \states, \action \in \actions(\state) \}$. For a concise presentation, we define $\SM \in \reals^{\SA \times \states}$ where $\SM((\state, \action), \state') {=} 1 \allowbreak {-} \transMat(\state, \action, \state')$ $\text{if } \state {=} \state'$ and $\SM((\state, \action), \state') {=} {-}\transMat(\state, \action, \state')$ otherwise for all $(\state, \action) {\in} \SA$ and $\state' {\in} \states$ \cite{funke_farkas_2020, jantsch_certificates_2022}. Let $\TM \in \realsnn^{\SA \times [k]}$ be defined as $\TM((s,a), i) \allowbreak = \sum_{\state' \in \targetSet_i} \transMat(\state, \action, \state')$ for all $(s,a) \in \SA$ and $i \in [k]$. $\mdp$ is said to be \emph{EC-free} if its only ECs are formed by states in $\targets$.

Farkas certificates are vectors that satisfy linear inequalities derived from LP-characterizations for MDPs \cite{kallenberg_linear_1983, kallenberg_survey_1994} and Farkas' lemma. Given a certificate, we can \emph{easily validate} whether a query is indeed satisfied, by checking whether the certificate satisfies the inequalities. In contrast, if a user is given a scheduler, they need to compute the probability in the induced Markov chain to validate the result, which is not as straightforward. For \existsCQ-queries, certificates have been considered in \cite{etessami_multi-objective_2008} and we summarize their results in our notation and setting.
\begin{restatable}[Certificates for \existsCQ-queries]{lemma}{certificatesExistsCq}
\label{lemma:exist-CP}
For a conjunctive reachability property $\prop{\gtrsim}{\scheduler}(\boldLambda)$ we have:
\begin{enumerate}[label={(\roman*)}, align=left, leftmargin=*, itemsep=0mm, topsep=0.8mm, parsep=0mm]
\item $\exists \scheduler \in \schedulers \centerdot \prop{\gtrsim}{\scheduler}(\boldLambda) \iff
\exists \vect{y} \in \realsnn^\SA \centerdot \SM^\top \vect{y} \leq \initDistr \land \TM^\top \vect{y} \gtrsim \boldLambda$ \label{lemma:exist-CP-lb}
\end{enumerate}
and if $\mdp$ is EC-free we also have:
\begin{enumerate}[label={(\roman*)}, align=left, leftmargin=*, itemsep=0mm, topsep=0.8mm, parsep=0mm]
\setcounter{enumi}{1}
\item $\exists \scheduler \in \schedulers \centerdot \prop{\lesssim}{\scheduler}(\boldLambda) \iff
\exists \vect{y} \in \realsnn^\SA \centerdot \SM^\top \vect{y} \geq \initDistr \land \TM^\top \vect{y} \lesssim \boldLambda$ \label{lemma:exist-CP-ub}
\end{enumerate}
\end{restatable}
\begin{proof}[Sketch]
Follows from \cite[Theorem~3.2]{etessami_multi-objective_2008} and \cite[Lemma~3.17]{jantsch_certificates_2022}.
\end{proof}
The value $\vect{y}(\state, \action)$ can be interpreted as the frequency of playing $(\state, \action)$ under a scheduler that reaches $\targets$ almost surely \cite{jantsch_certificates_2022}. To satisfy queries with upper bounds in MDPs with ECs, it might be required to reach $F$ with probability smaller than $1$. Hence, the restriction to EC-free MDPs. The certificates for \existsDQ-queries can be derived by using the distributivity of the existential quantifier and applying the results from the single-objective setting to each disjunct \cite{funke_farkas_2020, jantsch_certificates_2022}.
Likewise, the certificates for \universalCQ-queries can be derived.

To the best of our knowledge, no explicit characterization of the certificates for \universalDQ-queries that also enable finding witnessing subsystems has been discussed yet. The works \cite{etessami_multi-objective_2008, forejt_quantitative_2011, forejt_pareto_2012} are mainly interested in checking the query and do so by considering the dual \existsCQ-query. The following lemma provides an explicit presentation of the certificates. An overview of certificates for all query types can be found in the Appendix, including \existsDQ- and \universalCQ-queries.
\begin{restatable}[Farkas certificates for \universalDQ-queries]{lemma}{certificatesUniversalDq}
\label{lemma:farkas-universal-DQs}
For a disjunctive reachability property $\prop{\bowtie}{\scheduler}(\boldLambda)$ we have:
\begin{enumerate}[label={(\roman*)}, align=left, leftmargin=*, itemsep=0mm, topsep=0.8mm, parsep=0mm]
\item $\forall \scheduler \in \schedulers \centerdot \prop{\leq}{\scheduler}(\boldLambda) \iff \exists \vect{x} \in \reals^\states \centerdot \exists \vect{z} \in \realsnn^{[k]} \setminus \{\vect{0}\} \centerdot \SM \vect{x} \geq \TM \vect{z} \land \vect{x}(\init) \leq \pmb{\lambda}^\top \vect{z}$ \label{lemma:universal-DQ-non-strict-ub}
\item $\forall \scheduler \in \schedulers \centerdot \prop{<}{\scheduler}(\boldLambda) \iff \exists \vect{x} \in \reals^\states \centerdot \exists \vect{z} \in \realsnn^{[k]} \centerdot \SM \vect{x} \geq \TM \vect{z} \land \vect{x}(\init) < \pmb{\lambda}^\top \vect{z}$ \label{lemma:universal-DQ-strict-ub}
\end{enumerate}
and if $\mdp$ is EC-free we also have:
\begin{enumerate}[label={(\roman*)}, align=left, leftmargin=*, itemsep=0mm, topsep=0.8mm, parsep=0mm]
\setcounter{enumi}{2}
\item $\forall \scheduler \in \schedulers \centerdot \prop{\geq}{\scheduler}(\boldLambda) \iff \exists \vect{x} \in \reals^\states \centerdot \exists \vect{z} \in \realsnn^{[k]} \setminus \{\vect{0}\} \centerdot \SM \vect{x} \leq \TM \vect{z} \land \vect{x}(\init) \geq \pmb{\lambda}^\top \vect{z}$ \label{lemma:universal-DQ-non-strict-lb}
\item $\forall \scheduler \in \schedulers \centerdot \prop{>}{\scheduler}(\boldLambda) \iff \exists \vect{x} \in \reals^\states \centerdot \exists \vect{z} \in \realsnn^{[k]} \centerdot \SM \vect{x} \leq \TM \vect{z} \land \vect{x}(\init) > \pmb{\lambda}^\top \vect{z}$ \label{lemma:universal-DQ-strict-lb}
\end{enumerate}
\end{restatable}
\begin{proof}[Sketch]
Application of Farkas' lemma (\Cref{lemma:farkas}) to \Cref{lemma:exist-CP}.
\end{proof}
We can now devise a simple \emph{certifying verification algorithm} based on \Cref{lemma:exist-CP} and \Cref{lemma:farkas-universal-DQs}. Given a query $\query$, the algorithm tries to find a certificate for $\query$ and if it cannot find such certificate, it computes a certificate for $\neg \query$. Note that certificates can be computed via linear programming in polynomial time \cite{matousek_understanding_2007}.
\begin{remark}
\label{remark:pareto-connection}
The decision algorithm for \existsCQ-queries in \cite[Algorithm~1]{forejt_pareto_2012} checks satisfaction by computing a sequence of vectors $\vect{w}$ that are akin to the vector $\vect{z}$ in \Cref{lemma:farkas-universal-DQs}. However, $\vect{x}$ is not computed nor characterized. We note that it is not obvious how to turn it into a certifying algorithm, because no \emph{easily checkable certificate} arises from the computations when the query holds.
\end{remark}
The certificates from Lemma \ref{lemma:exist-CP} and \ref{lemma:farkas-universal-DQs} are related to witnessing schedulers and subsystems.
The relation to schedulers is well-known and we summarize existing results. For subsystems this is less obvious and we now generalize \cite{funke_farkas_2020, jantsch_certificates_2022}.

\medskip

\noindent \textbf{Witnessing schedulers and Farkas certificates.} For $\existsCQ$-queries, the correspondence between the certificates $\vect{y}$ and memoryless schedulers in $\mdp$ is well-known \cite{kallenberg_linear_1983, etessami_multi-objective_2008, jantsch_certificates_2022}. The memoryless scheduler $\scheduler$, defined by $\scheduler(\state)(\action) = \vect{y}(\state, \action) \mathbin{/} \sum_{\action' \in \actions(\state)} \vect{y}(\state, \action')$ if $\sum_{\action' \in \actions(\state)} \vect{y}(\state, \action') > 0$ and any distribution over the available actions otherwise for all $(\state, \action) \in \SA$, is known to satisfy the query.

A $\universalDQ$-query asks a property to hold under all schedulers and it is less clear how to obtain a single scheduler demonstrating the satisfaction. Let $\prop{\bowtie}{\scheduler}(\boldLambda)$ be a \emph{disjunctive} and $\propAlt{\bowtie}{\scheduler}(\boldLambda)$ be a \emph{conjunctive} property with the same predicates and let $\mathsf{Ach} {=} \{\boldLambda' {\in} [0, 1]^k \mid \exists \scheduler \centerdot \propAlt{\not\bowtie}{\scheduler}(\boldLambda')\}$. Observe that $\forall \scheduler \centerdot \prop{\bowtie}{\scheduler}(\boldLambda)$ if and only if $\boldLambda \notin \mathsf{Ach}$. In \cite{forejt_pareto_2012}, this relation is used to determine a vector $\vect{z}$ (as described in \Cref{lemma:farkas-universal-DQs}) that separates $\boldLambda$ from $\mathsf{Ach}$, i.e.\ $\vect{z}^\top \boldLambda > \vect{z}^\top \boldLambda'$ for all $\boldLambda' \in \mathsf{Ach}$. This amounts to finding a scheduler $\scheduler^*$ such that $\sum_{i=1}^k \vect{z}(i) {\cdot} \prob^{\scheduler^*}(\eventually \targetSet_i) \eqqcolon \gamma^*$ is maximal. If $\vect{z}^\top \boldLambda \gtrsim \gamma^*$ holds, we can then conclude that $\boldLambda \notin \mathsf{Ach}$ and consequently $\scheduler^*$ can then serve as witness for the satisfaction of $\forall \scheduler \centerdot \prop{\bowtie}{\scheduler}(\boldLambda)$.

\medskip

\noindent \textbf{Witnessing subsystems and Farkas certificates.}
To consider witnesses in the form of \emph{subsystems}, we first show that the satisfaction of a lower-bounded query (not necessarily \Reach-query) in a subsystem implies that the query is also satisfied in the original MDP. Crucially, this allows us to use a subsystem as a witness for the satisfaction in the original MDP.
\begin{restatable}[Monotonicity]{theorem}{subsystemsLowerBounds}
\label{theorem:subsystems-and-lower-bounds}
Let $\mathcal{N}$ be an arbitrary MDP and $\mathcal{N}'$ be a subsystem of $\mathcal{N}$. Further, let $\prop{\gtrsim}{\scheduler}(\boldLambda)$ be a multi-objective property. Then we have:
\begin{enumerate}[label={(\roman*)}, align=left, leftmargin=*, itemsep=0mm, topsep=0.8mm, parsep=0mm]
\item  $\exists \scheduler' \in \schedulers^{\mathcal{N}'} \centerdot \prop{\gtrsim}{\scheduler'}(\boldLambda) \implies \exists \scheduler \in \schedulers^\mathcal{N} \centerdot \prop{\gtrsim}{\scheduler}(\boldLambda)$ \label{theorem:subsystems-and-lower-bounds-exists}
\item  $\forall \scheduler' \in \schedulers^{\mathcal{N}'} \centerdot \prop{\gtrsim}{\scheduler'}(\boldLambda) \implies \forall \scheduler \in \schedulers^\mathcal{N} \centerdot \prop{\gtrsim}{\scheduler}(\boldLambda)$ \label{theorem:subsystems-and-lower-bounds-universal}
\end{enumerate}
\end{restatable}
\Cref{theorem:subsystems-and-lower-bounds} is precisely the reason for considering lower-bounded \ReachInv-queries instead of \Reach-queries with mixed bounds. For the latter, monotonicity does not hold in general, as adding states might result in surpassing a threshold (also see \cite[Section~4.4]{jantsch_certificates_2022}).
It has been shown that there is a correspondence between witnessing subsystems and Farkas certificates for \emph{single-objective} reachability \cite{funke_farkas_2020, jantsch_certificates_2022}. Now we generalize the previous results to \emph{multi-objective} reachability.
Let $\cqPoly{\mdp}{\gtrsim}(\boldLambda)$ be the polyhedron formed by the conditions in \Cref{lemma:exist-CP} for \existsCQ-queries and $\dqPoly{\mdp}{\gtrsim}(\pmb{\lambda})$ the polyhedron formed by the conditions in \Cref{lemma:farkas-universal-DQs}. 
Let $\stateSupp{\vect{y}} = \{\state \in \states \mid \exists \action \in \actions(\state) \centerdot \vect{y}(\state, \action) > 0\}$ \cite{jantsch_certificates_2022}.
\begin{restatable}{theorem}{farkasSupportSubsystem}
\label{theorem:farkas-support-witnessing-subsystems}
Let $\prop{\gtrsim}{\scheduler}(\boldLambda)$ be a disjunctive and $\propAlt{\gtrsim}{\scheduler}(\boldLambda)$ be a conjunctive reachability property and $\states' \subseteq \states$. Then we have:
\begin{enumerate}[label={(\roman*)}, align=left, leftmargin=*, itemsep=0mm, topsep=0.8mm, parsep=0mm]
\item $\exists \vect{y} \in \cqPoly{\mdp}{\gtrsim}(\pmb{\lambda}) \centerdot \stateSupp{\vect{y}} \subseteq \states' \iff \exists \scheduler' \in \schedulers^{\mdp_{\states'}} \centerdot \propAlt{\gtrsim}{\scheduler'}(\boldLambda)$ \label{theorem:farkas-support-witnessing-subsystems-cq}
\end{enumerate}
and if $\mdp$ is EC-free we also have:
\begin{enumerate}[label={(\roman*)}, align=left, leftmargin=*, itemsep=0mm, topsep=0.8mm, parsep=0mm]
\setcounter{enumi}{1}
\item $\exists (\vect{x}, \vect{z}) \in \dqPoly{\mdp}{\gtrsim}(\pmb{\lambda}) \centerdot \supp{\vect{x}} \subseteq \states' \land \vect{x} \geq 0 \iff \forall \scheduler' \in \schedulers^{\mdp_{\states'}} \centerdot \prop{\gtrsim}{\scheduler'}(\boldLambda)$ \label{theorem:farkas-support-witnessing-subsystems-dq}
\end{enumerate}
\end{restatable}
In essence, the subsystems that are induced by the support of the certificates satisfy the query. Consequently, finding witnessing subsystems with a small number of states corresponds to finding certificates with a small support. For the single-objective setting, this observation has been made in \cite{funke_farkas_2020, jantsch_certificates_2022}, where mixed-integer LPs are used to find minimal certificates. Note that for downstream tasks, e.g. for manual inspection \cite{korn_pmc-vis_2023}, it can be desirable to obtain \emph{minimal witnessing subsystems}. 
Based on \Cref{theorem:farkas-support-witnessing-subsystems}, we also use MILPs to compute certificates with a minimal support, thereby yielding minimal witnessing subsystems. The MILPs in \Cref{fig:reachability-milps} use a \emph{Big}-$M$ encoding (see e.g.\ \cite{griva_linear_2008}) where $M$ is a sufficiently large constant. We refer to the \Cref{appendix:big-m} for a discussion on the choice of $M$.
\subsection{Transferring Witnesses}
\label{subsection:transfer-witnesses}
Recall that we reduce a \ReachInv-query $\query_\mathcal{N}$ in an MDP $\mathcal{N}$ to a corresponding query $\query_\mdp$ in the product MDP $\mdp$ and then to a \Reach-query $\query_{\hat{\mdp}}$ in the MEC quotient $\hat{\mdp}$, allowing us to restrict our attention to certificates and witnesses for \Reach-queries for MDPs in reachability form. Now we describe the transfer of witnesses for $\query_{\hat{\mdp}}$ in $\hat{\mdp}$ to witnesses for $\query_\mathcal{N}$ in $\mathcal{N}$ (lower part in \Cref{fig:overview-approach}).

\medskip

\noindent\textbf{Transferring witnessing subsystems.}
\label{subsection:transfer-witnessing-subsystem}
We can easily obtain a witnessing subsystem for $\mathcal{N}$ from a witnessing subsystem of $\hat{\mdp}$. Essentially, every state $\hat{\state}$ of $\hat{\mdp}$ corresponds to a set of states of $\mathcal{N}$. For an arbitrary subsystem $\hat{\mdp}'$ induced by $\hat{\states}' \subseteq \hat{\states}$, we consider the corresponding set of states $\states_{\mathcal{N}}' \subseteq \states_{\mathcal{N}}$. Let $\mathcal{N}'$ be the subsystem induced by $\states_{\mathcal{N}}'$. Then the following holds:
\begin{restatable}{lemma}{transferSubsystem}
If $\hat{\mdp}'$ satisfies $\query_{\hat{\mdp}}$, then $\mathcal{N}'$ satisfies $\query_{\mathcal{N}}$.
\label{lemma:transferring-subsystems}
\end{restatable}
Recall that if $\mathcal{N}'$ satisfies $\query_{\mathcal{N}}$, then so does $\mathcal{N}$ (\Cref{theorem:subsystems-and-lower-bounds}). While the minimality of a subsystem for $\hat{\mdp}$ is generally not preserved when transferring the subsystem to $\mathcal{N}$, we can weight states of the MEC quotient by the number of states of $\mathcal{N}$ they represent in the MILPs, resulting in small subsystems for $\mathcal{N}$.
\begin{figure}[t]
\centering
\scriptsize{
\begin{subfigure}[t]{0.475\textwidth}
\centering
$\begin{aligned}
&\text{min } \sum\nolimits_{\state \in \states} \boldGamma(\state) \text{  subject to:} \\
&\boldGamma \in \{0, 1\}^\states \text{ and } \vect{y} \in \cqPoly{\mdp}{\geq}(\boldLambda)\\
&\forall (\state, \action) \in \SA \centerdot \vect{y}(\state, \action) \leq \boldGamma(\state) \cdot M
\end{aligned}$
\caption{\scriptsize{MILP for $\existsCQ$-queries}}
\label{subfig:reachability-milps-exists}
\end{subfigure}%
\hfill
\begin{subfigure}[t]{0.475\textwidth}
\centering
$\begin{aligned}
&\text{min } \sum\nolimits_{\state \in \state} \boldGamma(\state)\text{  subject to:}\\
&\boldGamma \in \{0, 1\}^\states \text{ and } (\vect{x}, \vect{z}) \in \dqPoly{\mdp}{\geq}(\boldLambda)\\
&\forall \state \in \states \centerdot \vect{x}(\state) \leq \boldGamma(\state) \cdot M \land \vect{x}(s) \geq 0
\end{aligned}$
\caption{\scriptsize{MILP for $\universalDQ$-queries}}
\label{subfig:reachability-milps-universal}
\end{subfigure}%
}
\caption{MILPs for finding minimal witnessing subsystems for \Reach-queries.}
\label{fig:reachability-milps}
\end{figure}
\begin{figure}[t]
\centering
\begin{subfigure}{0.19\textwidth}
\centering
\begin{tikzpicture}[x=18mm,y=15mm,font=\scriptsize]
          \node[state] (s0) {$s_0$};
          \node[state] (s2) [opacity=0.5,below = 0.8cm of s0]  {$s_3$};
          \node[state] (s3) [right = 0.2cm  of s0] {$s_1$};
          \node[state] (s1) [right = 0.3cm of s3]  {$s_2$};
          \node[state] (s4) [opacity=0.5,below = 0.8cm of s1]  {$s_4$};

          \node[bullet] (s0s2s3) [below = .2cm of s0] {};
          
                    \node (init) [above = 0.3cm of s0] {};
          \draw (init) edge[ptran] (s0);

			{\color{black}
			
				\draw (s1) edge[ptran] (s3);
			}
			
			{\color{orange!90}
			
				\draw (s3) edge[ptran, bend left=18] node[above,pos=.5]{$c$} (s1);
			}
			
			{\color{green!50!black}
				\draw (s0) -- node[left,pos=.8]{$b$} (s0s2s3);
				\draw (s0s2s3) edge[ptran] coordinate[pos=.3] (bs0s2) node[right,pos=.5]{\tiny{$0.5$}} (s2);
				\draw (s0s2s3) edge[ptran, bend right, out=0.2] coordinate[pos=.08] (bs0s3) node[above,pos=.5, yshift=-0.3mm]{\tiny{$0.5$}} (s3);
				\draw (bs0s2) to[bend right] (bs0s3);
			}
			
			{\color{black}
			
				\draw (s1) edge[ptran, loop above, looseness=6] (s1);
			}
			
			{\color{black}
			
				\draw (s3) edge[ptran, loop above, looseness=6] (s3);
			}
			
			{\color{black}
			
				\draw [opacity = 0.5](s2) edge[ptran, bend right=9] (s4);
			}
			
			{\color{black}
			
				\draw [opacity = 0.5](s4) edge[ptran, bend right=9] (s2);
			}
			
			{\color{blue!50!black}
			
				\draw[opacity = 0.5] (s4) edge[ptran] node[left,pos=.2]{$a$} (s1);
			}
			
			{\color{purple}
			
				\draw[opacity = 0.5] (s2) edge[ptran] node[above,pos=.5]{$d$} (s1);
			}
			
\end{tikzpicture}
\caption{\scriptsize{MDP $\mathcal{N}$}}
\label{subfig:mdp}
\end{subfigure}
\hfill
\begin{subfigure}{0.32\textwidth}
\centering
\begin{tikzpicture}[x=18mm,y=15mm,font=\scriptsize]

          \node[state] (s0) {$s_0, q_0$};
          \node[state] (s2) [opacity = 0.5, below = 0.8cm of s0]  {$s_3, q_1$};
          \node[state] (s3) [right = 0.4cm  of s0] {$s_1, q_0$};
          \node[state] (s1) [right = 0.4cm of s3]  {$s_2, q_2$};
          \node[state] (s32) [below = 0.8cm  of s1] {$s_1, q_2$};
          \node[state] (s41) [opacity = 0.5, below = 0.8cm  of s3] {$s_4, q_1$};

          	\draw[black,densely dotted, opacity=0.5, line width=0.12mm] ($(s2.north west)+(-0.12,0)$)  rectangle ($(s41.south east)+(0.12,0)$) node[pos=.5, xshift=-10mm, yshift=3.4mm]{\tiny{$\mec_1$}};
			
			\draw[black,densely dotted, line width=0.12mm] ($(s3.north west)+(-0.12,0.14)$)  rectangle ($(s3.south east)+(0.12,0)$) node[left, yshift=0, xshift=-8.5mm, yshift=5mm]{\tiny{$\mec_2$}};
			
			\draw[black,densely dotted, line width=0.12mm] ($(s1.north west)+(-0.12,0.14)$)  rectangle ($(s32.south east)+(0.12,0)$) node[left, xshift=-8.7mm, yshift=1.2mm]{\tiny{$\mec_3$}};

          \node[bullet] (s0s2s3) [below = .2cm of s0] {};
          
                    \node (init) [above = 0.3cm of s0] {};
          \draw (init) edge[ptran] (s0);
          
          	{\color{blue!50!black}
			
				\draw[opacity = 0.5] (s41) edge[ptran, bend right=15] node[above,pos=.5]{$a$} (s1);
			}

			{\color{black}
			
				\draw (s1) edge[ptran, bend left] (s32);
			}
			
			{\color{orange!90}
			
				\draw (s32) edge[ptran] node[left,pos=.5]{$c$} (s1);
			}
			
			{\color{orange!90}
			
				\draw (s3) edge[ptran] node[below,pos=.5, bend right]{$c$} (s1);
			}
			
			{\color{green!50!black}
				\draw (s0) -- node[left,pos=.8]{$b$} (s0s2s3);
				\draw (s0s2s3) edge[ptran] coordinate[pos=.3] (bs0s2) node[right,pos=.5]{\tiny{$0.5$}} (s2);
				\draw (s0s2s3) edge[ptran, bend right, out=0.2] coordinate[pos=.08] (bs0s3) node[above,pos=.5,yshift=-0.5mm]{\tiny{$0.5$}}  (s3);
				\draw (bs0s2) to[bend right] (bs0s3);
			}
			
			{\color{black}
			
				\draw (s1) edge[ptran, loop above, looseness=6] (s1);
			}
			
			{\color{black}
				\draw (s3) edge[ptran, loop above, looseness=6] (s3);
			}
			
			{\color{black}
				\draw (s32) edge[ptran, out=100,in=130, looseness=5] (s32);
			}
			
			{\color{red!10!black}
			
				\draw[opacity = 0.5](s2) edge[ptran, bend right=12] (s41);
			}
			
			{\color{black}
			
				\draw[opacity = 0.5] (s41) edge[ptran, bend right=12](s2);
			}
			
			{\color{purple}
			
				\draw[opacity = 0.5] (s2) edge[ptran] node[below,pos=.5]{$d$} (s1);
			}

\end{tikzpicture}
\caption{\scriptsize{Product $\mdp$}}
\label{subfig:product}
\end{subfigure}
\hfill
\begin{subfigure}{0.25\textwidth}
\centering
\begin{tikzpicture}[x=18mm,y=15mm,font=\scriptsize]
          \node[state] (s0) {$\state_0, q_0$};
          \node[state] (mec1) [opacity = 0.5, below = 0.8cm of s0]  {$\quotientState_{\mec_1}$};
          \node[state] (mec2) [right = 0.25cm  of s0] {$\quotientState_{\mec_2}$};
          \node[state] (mec3) [below = 0.2cm  of mec2] {$\quotientState_{\mec_3}$};
          \node[state] (exit2) [right = 0.2cm  of mec2] {$\exit_2$};
          \node[state] (exit3) [right = 0.2cm  of mec3] {$\exit_3$};
		  \node[state] (exit1) [below = 0.08cm  of exit3] {$\exit_1$};

          \node[bullet] (s0mec1mec2) [below = .2cm of s0] {};
          
                    \node (init) [above = 0.3cm of s0] {};
          \draw (init) edge[ptran] (s0);
			
			{\color{green!50!black}
				\draw (s0) -- node[left,pos=.8]{$b$} (s0mec1mec2);
				\draw (s0mec1mec2) edge[ptran] coordinate[pos=.3] (bs0mec1) node[right,pos=.5]{\tiny{$0.5$}}  (mec1);
				\draw (s0mec1mec2) edge[ptran, bend right, out=0.2] coordinate[pos=.08] (bs0mec2) node[above,pos=.5,yshift=-0.5mm]{\tiny{$0.5$}}  (mec2);
				\draw (bs0mec1) to[bend right] (bs0mec2);
			}
			
			{\color{orange!90}
			
				\draw (mec2) edge[ptran] node[right,pos=.5]{$c$} (mec3);
			}
			
			{
				\draw [opacity = 0.5] (mec1) edge[ptran, bend right=10] node[above,pos=.8]{$\tau$} (exit1);
			}
			
			{
				\draw (mec2) edge[ptran] node[above,pos=.5]{$\tau$} (exit2);
			}
			
			{
				\draw (mec3) edge[ptran] node[above,pos=.5]{$\tau$} (exit3);
			}
			
			{\color{purple}
			
				\draw[opacity = 0.5] (mec1) edge[ptran, bend left=10] node[below,pos=.6, yshift=1mm]{$d$} (mec3);
			}
			
			{\color{blue!50!black}
			
				\draw [opacity = 0.5] (mec1) edge[ptran, bend right=30] node[right,pos=.4]{$a$} (mec3);
			}
			
\end{tikzpicture}
\caption{\scriptsize{MEC quotient $\hat{\mdp}$}}
\label{subfig:quotient}
\end{subfigure}
\hfill
\begin{subfigure}{0.16\textwidth}
\centering
\begin{tikzpicture}[x=18mm,y=15mm,font=\scriptsize]

          \node[state] (s1)  {$\state_3, q_1$};
          \node[state] (s32) [below = 0.5cm  of s1] {$\state_4, q_1$};
          \node (init) [left = 0.2cm  of s1] {};
          \node (outd) [right =0.3cm of s1] {};
          \node (outa) [right=0.3cm of s32] {};
          \node(prob1)[left=-3mm of outd, inner sep=0.7mm, xshift=-0.2mm]{$\frac{1}{3}$};
          \node(prob2)[left=-3mm of outa, inner sep=0.7mm, xshift=-0.2mm]{$\frac{2}{3}$};
          \node(mu)[inner sep=0.7mm] at ($(prob1)!0.5!(prob2)$) {\tiny{$\boldMu$}};
          \node(sched1)[above=0mm of s1, xshift=1mm, yshift=-0.4mm] {\tiny{${\hat{\scheduler}(\quotientState_{\mec_1})({\color{purple}d})}{=}1{/}4$}};
          \node(sched2)[below=0mm of s32, xshift=1mm, yshift=0.4mm] {\tiny{${\hat{\scheduler}(\quotientState_{\mec_1})({\color{blue!50!black}a})}{=}2{/}4$}};

			\draw (mu) -> (prob1);
			\draw (mu) -> (prob2);

			{\color{black}
			
				\draw (s1) edge[ptran, bend left=20] (s32);
			}
			
			{\color{black}
			
				\draw (s32) edge[ptran, bend left=20] (s1);
			}
			
			{
			\color{green!50!black}
			\draw (init) edge[ptran] (s1);
			}
			
			{
			\color{purple}
			\draw (s1) edge[ptran] node[below,pos=.4]{\tiny{$d$}} (outd);
			}
			
			{
			\color{blue!50!black}
			\draw (s32) edge[ptran] node[above,pos=.4]{\tiny{$a$}} (outa);
			}
\end{tikzpicture}
\caption{\scriptsize{MEC $\mec_1$}}
\label{subfig:mec}
\end{subfigure}
\caption{Grayed-out states are not in the subsystem and transitions leading there are redirected to a fresh $\exit$ state. For readability, some action names are omitted.}
\label{fig:mdps}
\end{figure}
\begin{example}
Consider the MDP $\mathcal{N}$ in \Cref{subfig:mdp} and query $\query_\mathcal{N} {=} \forall \scheduler {\in} \schedulers^\mathcal{N} \centerdot \allowbreak\prob_{\mathcal{N}}^\scheduler(\globally \allowbreak\{\state_0, \allowbreak \state_1\}) \allowbreak {\geq} 0.25 {\lor} \prob_{\mathcal{N}}^\scheduler(\eventually \{\state_2\}) {\geq} 0.25$. We construct the product $\mdp$ (\Cref{subfig:product}), where an automaton state $q_i$ with $i {>} 0$ indicates that a state outside $\{\state_0, \state_1\}$ and $q_2$ the state $\state_2$ has been visited. We then consider the MEC quotient $\hat{\mdp}$ (\Cref{subfig:quotient}) and rephrase $\query_\mathcal{N}$ to $\query_{\hat{\mdp}} {=} \forall \hat{\scheduler} {\in} \schedulers^\mathcal{\hat{\mdp}} \centerdot \prob_{\hat{\mdp}}^{\hat{\scheduler}}(\eventually \exit_2) {\geq} 0.25 {\lor} \prob_{\hat{\mdp}}^{\hat{\scheduler}}(\eventually \exit_3) {\geq} \allowbreak 0.25$. Reaching $\exit_2$ corresponds to staying in a MEC without seeing a state outside $\{\state_0, \state_1\}$ and reaching $\exit_3$ corresponds to having visited $\state_2$. A witnessing subsystem for $\hat{\mdp}$ is given by the non-grayed out states in \Cref{subfig:quotient} and a subsystem for $\mathcal{N}$ is obtained by considering the corresponding states.
\end{example}
\noindent\textbf{Transferring witnessing schedulers.}
Let us now construct a scheduler $\scheduler$ for $\mdp = (\states, \actions, \init, \transMat)$ from a memoryless scheduler $\hat{\scheduler} \in \schedulers^{\hat{\mdp}}$. We then obtain a scheduler for $\mathcal{N}$ from $\scheduler$, by interpreting the automaton component as additional memory locations \cite{baier_principles_2008}.
To construct a scheduler $\scheduler$, special care needs to be taken when $\hat{\scheduler}$ leaves a MEC state $\quotientState_\mec$ with probability $0 < p < 1$. Here, a standard memoryless scheduler for $\mdp$ does not suffice\footnote{A memoryless scheduler either leaves or stays in a MEC almost surely.}. Instead we construct an equivalent scheduler with only 2 memory locations for $\mdp$ and allow \emph{stochastic memory updates} \cite{brazdil_markov_2014}. We proceed as follows:
\begin{enumerate*}[label=(\roman*)]
	\item For every MEC $\mec \in \MECS(\mdp)$ we construct a scheduler $\scheduler_{\mathsf{stay}}$ that stays in $\mec$ almost surely,\label{item:stay-scheduler}
	\item and a scheduler $\scheduler_{\mathsf{leave}}$ that leaves it almost surely with the same probabilities as $\hat{\scheduler}$ (normalized with $p$),\label{item:leave-scheduler}
	\item and finally use these as building blocks for a scheduler $\scheduler$ with 2 memory locations and stochastic memory update for $\mdp$.\label{item:global-scheduler}
Conceptually, upon entering a MEC in $\mdp$, $\scheduler$ either switches to $\scheduler_{\mathsf{stay}}$ or $\scheduler_{\mathsf{leave}}$.
\end{enumerate*}

\medskip

\noindent\emph{\ref{item:stay-scheduler} Construction of $\scheduler_{\mathsf{stay}}$}: A memoryless scheduler $\scheduler_{\mathsf{stay}}$ that stays inside a MEC $\mec$ can be constructed by taking every internal action with a positive probability.

\medskip

\noindent\emph{\ref{item:leave-scheduler} Construction of $\scheduler_{\mathsf{leave}}$}: Let $\quotientState_\mec$ be the state corresponding to $\mec$ in $\hat{\mdp}$. Let $p = 1 {-} \hat{\scheduler}(\quotientState_\mec)(\tau)$ be the probability with which $\hat{\scheduler}$ leaves the MEC $\mec$. The construction of the memoryless scheduler $\scheduler_{\mathsf{leave}}$ is intricate, as we have to ensure that we leave $\mec$ via a state-action pair $(\state, \action)$ with the same probability as $\hat{\scheduler}$ plays $(\quotientState_\mec, \action)$ normalized with $p$. To show that such a scheduler can be constructed, we first establish a result for strongly connected DTMCs. Let $\dtmc = (\states, \boldDelta, \transMat)$ be a strongly connected DTMC. Let $\boldLambda \in [0,1]^\states$ and $\dtmc_{\boldLambda}$ be the DTMC resulting from $\dtmc$ by adding fresh copies $\state'$ for all states $\state \in \states$ and transitions from $\state$ to $\state'$ with probability $\boldLambda(\state)$ (the other transitions are rescaled with $1 - \boldLambda(\state)$).
\begin{restatable}{lemma}{dtmcFrequencies}
For every distribution $\boldMu \in \distr{\states}$ there exists a vector $\boldLambda \in [0,1]^\states$ such that for all states $\state$ we have $\prob_{\dtmc_{\boldLambda}}(\eventually \state') = \boldMu(\state)$.
\label{lemma:dtmc-frequencies}
\end{restatable}
\begin{proof}[Sketch]
We show that $\boldLambda$ can be obtained by solving a system of linear equations that characterize the expected frequencies of each state in $\dtmc_{\boldLambda}$.
\end{proof}
For any distribution $\boldMu$, we can redirect transitions of a strongly connected DTMC such that it is left according to $\boldMu$. 
We consider the scheduler $\scheduler'$ that takes every internal action in $\mec$ uniformly and as such induces a strongly connected DTMC. We instantiate $\boldMu$ such that it captures the probability with which $\hat{\scheduler}$ leaves $\quotientState_\mec$ via a state $\state \in \states(\mec)$ and use the resulting $\boldLambda$ to ``alter'' $\scheduler'$, thereby obtaining $\scheduler_\mathsf{leave}$.
We instantiate $\boldDelta$ and $\boldMu$ as follows. For all $\state \in \states(\mec)$ we define
\[
\Delta(\state) = \sum\nolimits_{(t, \action) \in \SA \setminus \mec} \freq_{\hat{\mdp}}^{\hat{\scheduler}}(\iota(t), \action) \cdot \transMat(t, \action, \state) + \initDistr(\state)
\]
Recall that $\iota$ maps a state from $\mdp$ to the corresponding one in $\hat{\mdp}$ and that $\initDistr$ is Dirac in the initial state $\init$ of $\mdp$. Then $\Delta(\state)$ describes the frequency with which $\mec$ is entered via state $\state$. We then define the initial distribution $\boldDelta(\state) = {\Delta(\state)}\mathbin{/}{\sum_{t \in \states(\mec)} \Delta(t)}$ for $\state \in \states(\mec)$. 
Let $\freq^{\hat{\scheduler}}(\state) = \sum_{(\state, \action) \in \SA \setminus C} \freq_{\hat{\mdp}}^{\hat{\scheduler}}(\quotientState_\mec, \action)$ for $\state \in \states(\mec)$, i.e. frequency of leaving $\mec$ via $\state$. Let $\freq^{\hat{\scheduler}}(\quotientState_\mec) = \sum_{\state \in \states(\mec)} \freq^{\hat{\scheduler}}(\state)$. 
For every state $\state \in \states({\mec})$ and $\action \in \actions(\state)$ with $(\state, \action) \notin \mec$ we define:
\[
\boldMu(\state, \action) = {\freq^{\hat{\scheduler}}(\state, \action)}/{\freq^{\hat{\scheduler}}(\quotientState_\mec)}
\qquad 
\text{and} 
\qquad \boldMu(\state) = \sum\nolimits_{(\state, b) \in \SA {\setminus} \mec} \boldMu(\state, b)
\]
For each state $\state \in \states(\mec)$ and $\action \in \actions(\state)$ we then define $\scheduler_{\mathsf{leave}}(\state)(\action) = (1 - \boldLambda(\state)) \cdot \scheduler'(\state)(\action)$ if $(\state, \action) \in \mec$ and $\scheduler_{\mathsf{leave}}(\state)(\action) = \boldLambda(\state) \cdot (\freq^{\hat{\scheduler}}(\state, \action) / \freq^{\hat{\scheduler}}(\state))$ if $\boldLambda(\state) > 0$ and $\scheduler_{\mathsf{leave}}(\state)(\action) = \scheduler'(\state)(\action)$ otherwise.

\medskip

\noindent\emph{\ref{item:global-scheduler} Construction of witnessing scheduler:} Let $\scheduler_{\mathsf{stay}_\mec}$ and $\scheduler_{\mathsf{leave}_\mec}$ be the schedulers that stay in and leave MEC $\mec$ almost surely, respectively, as previously described. Let $p_\mec$ be the probability with which $\hat{\scheduler}$ leaves MEC $\mec$. If $\init$ is in a MEC, then $p_{\mathsf{init}}$ is the probability with which $\hat{\scheduler}$ leaves the containing MEC, otherwise $p_{\mathsf{init}} = 1$. We define $\scheduler = (\alpha_{\mathsf{update}}, \alpha_{\mathsf{next}}, \{\mem_0, \mem_1\}, \boldDelta_{\mems})$ where $\boldDelta_{\mems}(\mem_0) = p_{\mathsf{init}}$ and $\boldDelta_{\mems}(\mem_1) = 1{-}p_{\mathsf{init}}$. 
Further, the next move function is given as
\[
\alpha_{\mathsf{next}}(\state, \mem) =
\begin{cases}
\hat{\scheduler}(\state), & \text{if } \state \not \in \allMecStates \\
\scheduler_{\mathsf{stay}_\mec}(\state), & \text{if } \mem = \mem_1, \state \in \states({\mec}) \\
\scheduler_{\mathsf{leave}_\mec}(\state), & \text{if } \mem = \mem_0, \state \in \states({\mec})
\end{cases}
\]
The memory update function is defined as $\alpha_{\mathsf{update}}(\action, \state, \mem)(\mem_0) {=} p_\mec$ and $\alpha_{\mathsf{update}}\allowbreak(\action, \state, \mem)(\mem_1) {=} 1{-}p_\mec$ if $\state \in \states({\mec})$ and there does not exist $t \in \states({\mec})$ with $(t, \action) \in \mec$. Otherwise, we set $\alpha_{\mathsf{update}}(\action, \state, \mem)(\mem) {=} 1$.
The scheduler $\scheduler$ ``flips a coin'' \emph{upon entering MECs} to decide whether it stays in or leaves the MEC. Depending on the outcome, it either switches to a scheduler $\scheduler_{\mathsf{stay}}$ or $\scheduler_{\mathsf{leave}}$ that stay in or leave the MEC, respectively. Thereby, we ensure that $\scheduler$ stays in and leaves a MEC with same probability as $\hat{\scheduler}$. Outside MECs, $\scheduler$ behaves like $\hat{\scheduler}$. Note that once $\scheduler$ switches to $\mem_1$ it cannot change back to $\mem_0$ and stays in the MEC almost surely.
\begin{example}
Consider MDP $\mathcal{N}$ from \Cref{fig:mdps} and query $\exists \scheduler \allowbreak{\in} \schedulers^\mathcal{N} \centerdot \allowbreak\prob_{\mathcal{N}}^\scheduler(\globally (\states \setminus \{\state_2\})) \allowbreak {\geq} 0.5 {\land} \prob_{\mathcal{N}}^\scheduler(\eventually \{\state_2\}) {\geq} 0.5$. The query for $\hat{\mdp}$ is given by $\exists \hat{\scheduler} {\in} \schedulers^{\hat{\mdp}} \centerdot \allowbreak\prob_{\hat{\mdp}}^{\hat{\scheduler}}(\eventually \{\exit_1,\allowbreak\exit_2\}){\geq} 0.5 {\land} \prob_{\hat{\mdp}}^{\hat{\scheduler}}(\eventually \{\exit_3\}) {\geq} 0.5$. We consider a witnessing scheduler $\hat{\scheduler}$ that takes {\color{blue!50!black} $a$}, {\color{purple}$d$} and {\color{orange!90}$c$} with probability $1{/}2$, $1{/}4$ and $1{/}4$, respectively. We construct a corresponding scheduler $\scheduler$ for $\mdp$, focusing on MEC $\mec_1$. Note that $\hat{\scheduler}$ stays with probability $1{/}4$ in $\mec_1$, i.e.\ $\hat{\scheduler}(\quotientState_{\mec_1})(\tau) {=} 1{/}4$. A scheduler $\scheduler_{\mathsf{stay}}$ can be constructed by only taking internal actions. For $\scheduler_{\mathsf{leave}}$ we need to ensure that $\mec_1$ is left correctly, that is, \emph{if $\hat{\scheduler}$ leaves $\mec_1$} it does so with probability $1/3$ via {\color{purple}$d$} and $2/3$ via {\color{blue!50!black} $a$}, hence $\boldMu {=} (1/3, 2/3)$ (see \cref{subfig:mec}). Because $\mec_1$ is only entered via $(\state_3, q_1)$, we set $\boldDelta {=} (1, 0)$ and applying \Cref{lemma:dtmc-frequencies} then yields $\boldLambda {=} (1/6, 2/5)$. Consequently, we define $\scheduler_{\mathsf{leave}}(\state_3, q_1)({\color{purple}d}) {=} 1/6$ and $\scheduler_{\mathsf{leave}}(\state_4, q_1)({\color{blue!50!black} a}) {=} 2/5$. We then construct $\scheduler$ by combining the different schedulers. Particularly, $\scheduler$ changes its memory location from $\mem_0$ to $\mem_1$ with probability $0.25$ when taking ${\color{green!50!black}b}$ and arriving in $(\state_3, q_1)$.
\end{example}

\section{Certificates and Witnesses for \MP-Queries}
\label{section:mean-payoff}
Building on the ideas for \Reach-queries, we address certificates and witnesses for multi-objective \MP-queries. We first discuss the certificates for \existsCQ- and \universalDQ-queries. The former characterization is well-studied \cite{brazdil_markov_2014}, while the latter again has only been implicitly considered \cite{quatmann_multi-objective_2021, quatmann_verification_2023}. Analogously, we use the certificates to find witnessing subsystems. For the remainder of this section, we fix an arbitrary MDP $\mdp =  (\states, \actions, \init, \transMat)$ with reward vectors $\vect{r}_1, \dots, \vect{r}_k \in \rationals^\SA$.

\medskip

\noindent\textbf{Farkas Certificates for \MP-queries.} While our certificates closely resemble classical results from \cite{kallenberg_linear_1983, puterman_markov_1994} and \cite{brazdil_markov_2014}, the conditions are slight variations thereof, allowing us to find minimal witnessing subsystems.
We define $\vect{r}_{\min} \in \rationals^{[k]}$ by $\vect{r}_{\min}(i) = \min_{(\state, \action) \in \SA} \vect{r}_i(\state, \action)$ for all $i \in [k]$, i.e.\ the vector containing the minimal reward for each reward vector $\vect{r}_i$. Similarly, we define $\vect{R}_{\min} \in \rationals^{\states \times [k]}$ by $\vect{R}_{\min}(\state, i) = \vect{r}_{\min}(i)$ for all $\state \in \states$ and $i \in [k]$.
\begin{restatable}[Certificates for \existsCQ-mean-payoff queries]{lemma}{certificatesExistsCQ}
\label{lemma:certificates-mp-exists}
There exists a scheduler $\scheduler \in \schedulers^\mdp$ such that $\Conj_{i=1}^k \expectation[\mdp, \init]{\scheduler}{\lrInf{\vect{r}_i}} \geq \lambda_i$ if and only if there exist $\vect{x}, \vect{y} \in \realsnn^\SA$ and $\vect{z} \in \realsnn^\states$ such that:
\begin{itemize}[align=left, leftmargin=*, itemsep=0mm, topsep=0.8mm, parsep=0mm]
\item $\forall \state {\in} \states \centerdot \initDistr(\state) {+} \sum_{(\state', \action') \in \SA} \transMat(\state', \action', \state) {\cdot} \vect{y}(\state', \action') {=}  \vect{z}(\state) {+} \sum_{\action \in \actions(\state)} \vect{y}(\state, \action) {+} \vect{x}(\state, \action)$
\item $\forall \state {\in} \states \centerdot \sum_{(\state', \action') \in \SA} \transMat(\state', \action', \state) \cdot \vect{x}(\state', \action') = \sum_{\action \in \actions(\state)} \vect{x}(\state, \action)$
\item $\forall i {\in} [k] \centerdot \sum_{(\state, \action) \in \SA} \vect{x}(\state, \action) \cdot \vect{r}_i(\state, \action) + \sum_{\state \in \states} \vect{z}(\state) \cdot {\vect{r}}_{\min}(i) \geq \lambda_i$
\end{itemize}
Let $\mathcal{H}_{\mdp}^{\mathsf{MP}}(\boldLambda) \subseteq \realsnn^\SA \times \realsnn^\SA \times \realsnn^\states$ denote the corresponding polyhedron.
\end{restatable}
\noindent In \cite{brazdil_markov_2014}, $\vect{x}$ and $\vect{y}$ correspond to the \emph{recurrent} and \emph{transient} flows, capturing the frequency of the state-action pairs in the limit and transient part, respectively. Compared to \cite{brazdil_markov_2014}, we consider the additional variable $\vect{z}$, allowing flow to be ``redirected'' to an implicit state where the worst possible reward is collected.
\begin{restatable}[Certificates for \universalDQ-mean-payoff queries]{lemma}{certificatesForallDQ}
\label{lemma:certificates-mp-universal}
For all schedulers $\scheduler \in \schedulers^\mdp$ we have $\Disj_{i=1}^k \expectation[\mdp, \init]{\scheduler}{\lrSup{\vect{r}_i}} \geq \lambda_i$ if and only if there exist $\vect{g}, \vect{b} \in \reals^\states$ and $\vect{z} \in \realsnn^{[k]}$ such that:
\begin{itemize}[align=left, leftmargin=*, itemsep=0mm, topsep=0.8mm, parsep=0mm]
\item $\forall (\state, \action) {\in} \SA \centerdot \vect{g}(\state) \leq \sum_{\state' \in \states} \transMat(\state, \action, \state') \cdot \vect{g}(\state')$
\item $\forall (\state, \action) {\in} \SA \centerdot \vect{g}(\state) {+} \vect{b}(\state) \leq \sum_{\state' \in \states} \transMat(\state, \action, \state') {\cdot} \vect{b}(\state') {+} \sum_{i=1}^k  \vect{z}(i) \cdot \vect{r}_i(\state, \action)$
\item $\forall \state {\in} \states \centerdot \vect{g}(\state) \geq \sum_{i=1}^k \vect{z}(i) \cdot {\vect{r}}_{\min}(i)$
\item $\vect{g}(\init) \geq \sum_{i=1}^k \lambda_i \cdot \vect{z}(i)$ and $\sum_{i=1}^k \vect{z}(i) = 1$
\end{itemize}
Let $\mathcal{F}_{\mdp}^{\mathsf{MP}}(\boldLambda) \subseteq \reals^\states \times \reals^\states \times \realsnn^{[k]}$ denote the corresponding polyhedron.
\end{restatable}
We obtain the certificates via application of Farkas' lemma to the characterizations given in \cite{brazdil_markov_2014} and \cite{kallenberg_linear_1983, puterman_markov_1994}. Analogous to the reachability setting \cite{forejt_pareto_2012}, one can interpret $\vect{z}$ as a separating vector. This is used in \cite{quatmann_multi-objective_2021, quatmann_verification_2023} where the vector $\vect{z}$ arises as by-product of verifying the dual \existsCQ-query. However, neither have certificates nor witnessing subsystems been addressed. In \cite{puterman_markov_1994}, $\vect{g}$ and $\vect{b}$ are referred to as gain and bias, capturing the mean-payoff and the expected deviation until the mean-payoff ``stabilizes'' \cite{puterman_markov_1994, kretinsky_efficient_2017}, respectively.

\medskip

\noindent\textbf{Witnessing Subsystems for \MP-queries.}
We focus on obtaining witnessing from the certificates. Schedulers have been extensively studied in \cite{brazdil_markov_2014}. Recall that in subsystems (\Cref{section:preliminaries}), the smallest possible reward is collected in $\exit$.
\begin{restatable}[Certificates and subsystems]{theorem}{certificatesAndSubsystems}
Let $\states' \subseteq \states$. Then we have:
\begin{enumerate}[label={(\roman*)}, align=left, leftmargin=*, itemsep=0mm, topsep=0.8mm, parsep=0mm]
\item $\exists \scheduler' \in \schedulers^{\mdp_{\states'}} \centerdot \Conj_{i=1}^k \expectation[\mdp_{\states'}, \init]{\scheduler'}{\lrInf{\vect{r}_i}} \geq \lambda_i$ if and only if there exists $(\vect{x}, \vect{y}, \vect{z}) \in \mathcal{H}_{\mdp}^{\mathsf{MP}}(\boldLambda)$ such that $\stateSupp{\vect{x}} \union \stateSupp{\vect{y}} \subseteq \states'$. \label{theorem:certificate-subsystem-exists}
\item $\forall \scheduler' \in \schedulers^{\mdp_{\states'}}\Disj_{i=1}^k \expectation[\mdp_{\states'}, \init]{\scheduler'}{\lrSup{\vect{r}_i}} \geq \lambda_i$ if and only if there exists $(\vect{g}, \vect{b}, \vect{z}) \in \dqPoly{\mdp}{\geq}(\boldLambda)$ such that $\supp{\vect{g} - \vect{R}_{\min} \vect{z}} \subseteq \states'$. \label{theorem:certificate-subsystem-forall}
\end{enumerate}
\label{theorem:certificates-and-subsystems}
\end{restatable}
To find minimal witnessing subsystems for \existsCQ-queries, we need to set as many entries of $\vect{x}$ and $\vect{y}$ to zero as possible, effectively redirecting the flow to $\exit$. For \universalDQ-queries we strive to set as many entries $\vect{g}(\state)$ to the minimal possible reward as possible, indicating that transitions to such states can be safely redirected to $\exit$. The corresponding MILPs for minimizing the support are similar to the ones for reachability (\Cref{fig:reachability-milps}) and are described in \Cref{appendix:big-m}.

\section{Experiments}
\label{section:experiments}
\noindent\textbf{Setup.}
We have implemented the computation of certificates and witnesses for multi-objective queries in a prototypical Python tool, using \textsc{Storm}'s Python interface \cite{hensel_probabilistic_2022} for model parsing and MEC quotienting and Gurobi \cite{gurobi_optimization_llc_gurobi_2023} for solving (MI)LPs. Our tool exports certificates as \emph{JSON} files, witnessing subsystems in \textsc{Storm}'s explicit format and schedulers as \emph{DOT} file. For \Reach-queries, witnessing subsystems can also be exported as \textsc{Prism} programs. Our experiments have been run on a machine with an AMD Ryzen 5 3600 CPU (3.6 GHz) and 16 GB RAM. The time limit of Gurobi has been set to $5$ minutes. All time measurements are given in seconds and correspond to wall clock times.

We consider the consensus (\textsf{coin}) and firewire (\textsf{fire}) models from the Prism benchmark \cite{kwiatkowsa_prism_2012}, describing a shared coin and network protocol, respectively. Further, the zeroconf model (\textsf{zero}) from \cite{forejt_pareto_2012, kwiatkowska_compositional_2013} describes the configuration of IP addresses under certain environment assumptions. Additionally, a client-server mutex protocol from \cite{komuravelli_assume-guarantee_2012, brazdil_multigain_2015} (\textsf{csn}), a dining philosophers model from \cite{duflot_randomized_2004} (\textsf{phil}) and a model describing a network of sensors communicating over a lossy channel \cite{komuravelli_assume-guarantee_2012, kretinsky_efficient_2017} (\textsf{sen}). We compute certificates for queries considered for the mentioned models in \cite{quatmann_multi-objective_2021, brazdil_multigain_2015, kretinsky_efficient_2017, forejt_pareto_2012, kwiatkowska_compositional_2013}, compute schedulers as described in \Cref{subsection:transfer-witnesses} for \existsCQ-reachability queries and compute witnessing subsystems via MILPs. We are unaware of other tools for computing \emph{witnessing subsystems for multi-objective queries} and verified all queries with \textsc{Storm} \cite{hensel_probabilistic_2022}.

\begin{table}
\centering
\caption{Summary of the results.}
\setlength{\tabcolsep}{2pt}
\def\arraystretch{1}%
\tiny{
\begin{tabular}{lll|ccc|c|c|c|ccccc|cc}
 &  &  &  &  &  & \textsf{Build} & \textsf{Cert} & \textsf{Sched} & \multicolumn{5}{c|}{\textsf{WS}} & \multicolumn{2}{c}{\textsf{Size} \%}  \\
Model & $\card{\states}$ & $\card{\SA}$ & Type & $k$ & \# & mean & mean & mean & mean & min & max & std & TOs & min & max \\
\hline
\multirow[c]{2}{*}{\textsf{coin3}} & \multirow[c]{2}{*}{400} & \multirow[c]{2}{*}{592} & $(\exists,\land)$ & 2 & 5 & 0.261 & 0.025 & 0.042 &\multicolumn{4}{c}{\textbf{TO}} & 5 & 10.500 & 12.750 \\
 &  &  & $(\forall,\lor)$ & 2 & 5 & 0.202 & 0.026 & - & 2.265 & 1.899 & 2.607 & 0.268 & 0 & 26.500 & 44.250 \\
\multirow[c]{2}{*}{\textsf{coin4}} & \multirow[c]{2}{*}{528} & \multirow[c]{2}{*}{784} & $(\exists,\land)$ & 2 & 5 & 0.342 & 0.025 & 0.055 & \multicolumn{4}{c}{\textbf{TO}} & 5 & 11.364 & 12.879 \\
 &  &  & $(\forall,\lor)$ & 2 & 5 & 0.268 & 0.028 & - & 11.508 & 4.066 & 22.138 & 5.901 & 0 & 28.977 & 44.129 \\
\multirow[c]{2}{*}{\textsf{coin5}} & \multirow[c]{2}{*}{656} & \multirow[c]{2}{*}{976} & $(\exists,\land)$ & 2 & 5 & 0.428 & 0.028 & 0.068 &\multicolumn{4}{c}{\textbf{TO}} & 5 & 11.738 & 12.957 \\
 &  &  & $(\forall,\lor)$ & 2 & 5 & 0.331 & 0.031 & - & 18.594 & 13.366 & 29.227 & 5.564 & 0 & 28.201 & 42.835 \\
 \hline
\multirow[c]{2}{*}{\textsf{csn3}} & \multirow[c]{2}{*}{410} & \multirow[c]{2}{*}{913} & $(\exists,\land)$ & 3 & 1 & 0.249 & 0.021 & 0.101 & 5.077 & 5.077 & 5.077 & 0.000 & 0 & 22.683 & 22.683 \\
 &  &  & $(\forall,\lor)$ & 3 & 1 & 0.148 & 0.023 & - & 0.207 & 0.207 & 0.207 & 0.000 & 0 & 31.463 & 31.463 \\
\multirow[c]{2}{*}{\textsf{csn4}} & \multirow[c]{2}{*}{2115} & \multirow[c]{2}{*}{5749} & $(\exists,\land)$ & 4 & 1 & 1.548 & 0.034 & 0.161 & 70.966 & 70.966 & 70.966 & 0.000 & 0 & 4.397 & 4.397 \\
 &  &  & $(\forall,\lor)$ & 4 & 1 & 0.925 & 0.034 & - & 2.478 & 2.478 & 2.478 & 0.000 & 0 & 28.085 & 28.085 \\
\multirow[c]{2}{*}{\textsf{csn5}} & \multirow[c]{2}{*}{10610} & \multirow[c]{2}{*}{33493} & $(\exists,\land)$ & 5 & 1 & 13.775 & 0.070 & 0.550 &\multicolumn{4}{c}{\textbf{TO}} & 1 & 0.877 & 0.877 \\
 &  &  & $(\forall,\lor)$ & 5 & 1 & 8.883 & 0.051 & - & \multicolumn{4}{c}{\textbf{TO}} & 1 & 26.635 & 26.635 \\
 \hline
\multirow[c]{2}{*}{\textsf{fire3}} & \multirow[c]{2}{*}{4093} & \multirow[c]{2}{*}{5519} & $(\exists,\land)$ & 2 & 5 & 2.493 & 0.035 & 0.174 & \multicolumn{4}{c}{\textbf{TO}} & 5 & 2.174 & 14.097 \\
 &  &  & $(\forall,\lor)$ & 2 & 5 & 2.013 & 0.078 & - & 4.377 & 2.503 & 8.454 & 2.324 & 0 & 5.864 & 100.000 \\
\multirow[c]{2}{*}{\textsf{fire6}} & \multirow[c]{2}{*}{8618} & \multirow[c]{2}{*}{12948} & $(\exists,\land)$ & 2 & 5 & 6.118 & 0.045 & 0.328 & \multicolumn{4}{c}{\textbf{TO}} & 5 & 1.033 & 6.695 \\
 &  &  & $(\forall,\lor)$ & 2 & 5 & 4.541 & 0.220 & - & 29.480 & 5.098 & 96.928 & 34.970 & 0 & 3.272 & 100.000 \\
\multirow[c]{2}{*}{\textsf{fire9}} & \multirow[c]{2}{*}{14727} & \multirow[c]{2}{*}{24229} & $(\exists,\land)$ & 2 & 5 & 12.762 & 0.076 & 0.573 & \multicolumn{4}{c}{\textbf{TO}} & 5 & 0.604 & 3.918 \\
 &  &  & $(\forall,\lor)$ & 2 & 5 & 8.333 & 0.515 & - & 12.655 & 8.936 & 18.463 & 4.160 & 2 & 2.200 & 100.000 \\
 \hline
\textsf{zero2} & 3221 & 9319 & $(\forall,\lor)$ & 2 & 5 & 5.713 & 0.029 & - & 1.884 & 1.871 & 1.902 & 0.013 & 0 & 0.528 & 0.528 \\
\textsf{zero4} & 7259 & 21970 & $(\forall,\lor)$ & 2 & 5 & 25.749 & 0.039 & - & 4.141 & 4.110 & 4.174 & 0.027 & 0 & 0.317 & 0.317 \\
\textsf{zero6} & 12881 & 37891 & $(\forall,\lor)$ & 2 & 5 & 72.284 & 0.053 & - & 7.387 & 7.346 & 7.413 & 0.024 & 0 & 0.225 & 0.225 \\
 \hline\hline
\multirow[c]{2}{*}{\textsf{csn3}} & \multirow[c]{2}{*}{184} & \multirow[c]{2}{*}{439} & $(\exists,\land)$ & 3 & 2 & 0.340 & 0.017 & - & 0.287 & 0.020 & 0.554 & 0.267 & 0 & 5.978 & 5.978 \\
 &  &  & $(\forall,\lor)$ & 3 & 2 & 0.132 & 0.013 & - & 0.039 & 0.021 & 0.056 & 0.017 & 0 & 95.652 & 95.652 \\
\multirow[c]{2}{*}{\textsf{csn4}} & \multirow[c]{2}{*}{960} & \multirow[c]{2}{*}{2785} & $(\exists,\land)$ & 4 & 2 & 2.218 & 0.098 & - & 0.877 & 0.877 & 0.877 & 0.000 & 1 & 1.562 & 1.562 \\
 &  &  & $(\forall,\lor)$ & 4 & 2 & 0.804 & 0.098 & - & 2.727 & 0.336 & 5.119 & 2.392 & 0 & 91.667 & 91.667 \\
\multirow[c]{2}{*}{\textsf{csn5}} & \multirow[c]{2}{*}{4864} & \multirow[c]{2}{*}{16321} & $(\exists,\land)$ & 5 & 2 & 15.086 & 1.870 & - & 0.504 & 0.504 & 0.504 & 0.000 & 1 & 0.781 & 0.781 \\
 &  &  & $(\forall,\lor)$ & 5 & 2 & 5.019 & 1.924 & - & 7.832 & 7.832 & 7.832 & 0.000 & 1 & 89.309 & 89.309 \\
 \hline
\multirow[c]{2}{*}{\textsf{phil3}} & \multirow[c]{2}{*}{956} & \multirow[c]{2}{*}{2694} & $(\exists,\land)$ & 2 & 3 & 1.943 & 0.047 & - & 89.614 & 3.165 & 176.063 & 86.449 & 1 & 0.941 & 2.197 \\
 &  &  & $(\forall,\lor)$ & 2 & 3 & 0.642 & 0.874 & - & 124.765 & 96.783 & 139.185 & 19.789 & 0 & 100.000 & 100.000 \\
\multirow[c]{2}{*}{\textsf{phil4}} & \multirow[c]{2}{*}{9440} & \multirow[c]{2}{*}{35464} & $(\exists,\land)$ & 2 & 3 & 33.206 & 0.609 & - & \multicolumn{4}{c}{\textbf{TO}} & 3 & 0.911 & 0.911 \\
 &  &  & $(\forall,\lor)$ & 2 & 3 & 6.215 & 1.276 & - & \multicolumn{4}{c}{\textbf{TO}} & 3 & 100.000 & 100.000 \\
 \hline
\multirow[c]{2}{*}{\textsf{sen1}} & \multirow[c]{2}{*}{462} & \multirow[c]{2}{*}{1079} & $(\exists,\land)$ & 3 & 1 & 0.770 & 0.043 & - & \multicolumn{4}{c}{\textbf{TO}} & 1 & 7.576 & 7.576 \\
 &  &  & $(\forall,\lor)$ & 3 & 1 & 0.287 & 0.034 & - & 0.332 & 0.332 & 0.332 & 0.000 & 0 & 97.835 & 97.835 \\
\multirow[c]{2}{*}{\textsf{sen2}} & \multirow[c]{2}{*}{7860} & \multirow[c]{2}{*}{24584} & $(\exists,\land)$ & 3 & 1 & 21.955 & 0.609 & - & \multicolumn{4}{c}{\textbf{TO}} & 1 & 1.081 & 1.081 \\
 &  &  & $(\forall,\lor)$ & 3 & 1 & 5.765 & 0.900 & - & \multicolumn{4}{c}{\textbf{TO}} & 1 & 97.786 & 97.786
\end{tabular}
}
\label{table:summary-results}
\end{table}

\medskip

\noindent\textbf{Results.} The results are summarized in
\Cref{table:summary-results}, where the upper part is concerned with
\ReachInv- and the lower with \MP-queries. We now describe the columns. The column $k$ is the number of
predicates and \# the number of different bounds $\boldLambda$ we considered. \textsf{Build} shows the time for building the LP (for \ReachInv-queries this includes the time for construction of product MDP and MEC quotient), \textsf{Cert} the LP solving time and \textsf{Sched} the
time for computing schedulers from the certificates. The times for
computing witnessing subsystems is shown in the column \textsf{WS}. We provide mean, min, max and standard deviation because the times vary strongly. The number of
timeouts is shown in \textsf{TOs} and $\textbf{TO}$ means that
all computations timed out. In case of a timeout, the best subsystem 
found so far is returned. \textsf{Size} is the number
of states in the subsystem relative to the original MDP (in percentage).

\smallskip

\noindent\emph{Cost of certifying algorithm.} Following an LP-based approach \cite{forejt_quantitative_2011} for verifying a given multi-objective query $\query$, the simple certifying algorithm from \Cref{subsection:farkas-certificates}, needs to solve an LP for both $\query$ and $\neg\query$ in the worst case. 
Thus, the total costs of a certifying algorithm arises from solving two LPs instead of a single one. Our experiments (detailed in \Cref{appendix:experiments}) indicate that the solving time for the LP for $\query$ and $\neg\query$ are comparable. Thus, \textsf{Cert} can be interpreted as the overhead of such certifying algorithm. We observe that solving the LPs is relatively fast and that model building is currently the bottleneck in our prototypical implementation. We plan on providing a more efficient and competitive implementation in future work.  
\textsc{Storm} verified most queries in less than $0.1$ seconds. We refer to \Cref{appendix:experiments} for details and note that the verification algorithm as implemented in \textsc{Storm} \cite{forejt_pareto_2012, quatmann_multi-objective_2021} is based on value-iteration and \emph{non-certifying} (c.f.\ \Cref{remark:pareto-connection}).

\smallskip

\noindent\emph{Witnesses.}
Schedulers can quickly be computed from the certificates. For models where the quotient is smaller than the product MDP, e.g.\ \textsf{csn}, our techniques can be
useful. As for single-objectives \cite{funke_farkas_2020,
jantsch_certificates_2022}, finding small witnessing subsystems is challenging, particularly for $\existsCQ$-queries. For \universalDQ-queries, we often find subsystems in a reasonable amount of time. The number of states in the subsystem heavily depends on the bounds, query type and model. The subsystems for \existsCQ-queries can be substantially smaller than the original MDP, e.g.\ for \textsf{fire9} even $0.61\%$ the size of the original MDP, but can vary strongly for \universalDQ-queries, e.g. for \textsf{fire9} between $2\%$ to $100\%$. We refer to \Cref{appendix:experiments} for plots and details. 

\smallskip

\noindent
Our implementation, experiments and results are made available 
on Zenodo~\cite{baier_certificates_2024}.

\section{Conclusion}
We have given an explicit presentation of certificates for multi-objective queries and their relation to schedulers and witnessing subsystems, thereby generalizing \cite{funke_farkas_2020, jantsch_certificates_2022}. Our prototypical tool implements the presented techniques and has been applied on several case studies. In future work, we want to provide tool support for computing certificates more efficiently and address certificates and witnesses for richer modeling formalisms.

\printbibliography

\newpage
\appendix
\section{Proofs for \Cref{section:farkas-and-witnesses}}
\label{appendix:farkas}

\subsection{Proofs for \Cref{subsection:farkas-certificates}}
\label{subsection:appendix-proof-farkas}
\subsubsection{Farkas Certificates} In this section we provide proofs for results in \Cref{subsection:farkas-certificates} and also certificates for \existsDQ-queries and \universalCQ-queries, which work via simple reduction to the single-objective case \cite{funke_farkas_2020, jantsch_certificates_2022}. The following lemma summarizes the results from Theorem 3.2 in \cite{etessami_multi-objective_2008} in our setting and notation.
\begin{lemma}
\label{lemma:exist-CP-eq}
Let $\mdp = \mdpRF$ be an MDP in reachability form and let $\prop{\gtrsim}{\scheduler}(\boldLambda)$ be a conjunctive property. Then we have:
\begin{enumerate}[label={(\roman*)}, itemindent=1em]
\item $\exists \scheduler \in \schedulers \centerdot \prop{\gtrsim}{\scheduler}(\boldLambda) \iff
\exists \vect{y} \in \realsnn^\SA \centerdot \SM^\top \vect{y} = \initDistr \land \TM^\top \vect{y} \gtrsim \boldLambda$ \label{lemma:exist-CP-eq-lb-eq}
\end{enumerate}
and if $\mdp$ is EC-free we also have:
\begin{enumerate}[label={(\roman*)}, itemindent=1em]
\setcounter{enumi}{1}
\item $\exists \scheduler \in \schedulers \centerdot \prop{\gtrsim}{\scheduler}(\boldLambda) \iff
\exists \vect{y} \in \realsnn^\SA \centerdot \SM^\top \vect{y} = \initDistr \land \TM^\top \vect{y} \lesssim \boldLambda$ \label{lemma:exist-CP-eq-ub-eq}
\end{enumerate}
\end{lemma}
\begin{proof}
\ref{lemma:exist-CP-eq-lb-eq} follows from Theorem 3.2 in \cite{etessami_multi-objective_2008}. For EC-free MDPs, we can directly change the lower bounds in the proof of Theorem 3.2 to upper bounds. The reason is that we know that in an EC-free MDP the absorbing states are reached almost surely. Thus, we get \ref{lemma:exist-CP-eq-ub-eq}.
\qed
\end{proof}

\certificatesExistsCq*
\begin{proof}
To prove \ref{lemma:exist-CP-ub}, we simply apply Lemma 3.17 from \cite{jantsch_certificates_2022} to \Cref{lemma:exist-CP-eq} \ref{lemma:exist-CP-eq-ub-eq}.
Because Lemma 3.17 in \cite{jantsch_certificates_2022} relies on EC-freeness, we cannot do the same for \ref{lemma:exist-CP-lb}. We show that if there exists a $\vect{y} \in \realsnn^\SA$ that satisfies $\SM^\top \vect{y} \leq \initDistr \land \TM^\top \vect{y} \gtrsim \boldLambda$, then there also exists $\vect{y}' \in \realsnn^{\SA}$ such that $\SM^\top \vect{y}' = \initDistr \land \TM^\top \vect{y}' \gtrsim \boldLambda$.

Consider the MDP $\mdp'$ resulting from $\mdp$ by adding a fresh action $\tau$ to each state, that moves to $\exit$ with probability $1$. Clearly, if a \existsCQ-query $\query$ with lower bounds is satisfied in $\mdp$, then also in $\mdp'$ and vice versa. The reason is that the set of paths that reach the targets in $\mdp$ and $\mdp'$ is equivalent. Let $\SM'$ and $\TM'$ be defined as follows:
\[
\SM' = 
\begin{pmatrix}
\SM \\
\vect{I}
\end{pmatrix}
\qquad
\TM' = 
\begin{pmatrix}
\TM \\
\vect{I} \cdot 0
\end{pmatrix}
\]
where $\vect{I}$ is the identity matrix $\vect{I} \in \{0, 1\}^{\states \times \states}$. Now suppose we have $\vect{y} \in \realsnn^\SA$ that satisfies $\SM^\top \vect{y} \leq \initDistr \land \TM^\top \vect{y} \gtrsim \boldLambda$. Then there exists $\vect{z} \in \realsnn^\states$ such that $\SM^\top \vect{y} + \vect{z} = \initDistr$. In particular, we have
\[
(\SM')^\top
\begin{pmatrix}
\vect{y} \\
\vect{z}
\end{pmatrix}
= \initDistr
\qquad
(\TM')^\top
\begin{pmatrix}
\vect{y} \\
\vect{z}
\end{pmatrix}
\gtrsim \boldLambda
\]

From \Cref{lemma:exist-CP-eq} \ref{lemma:exist-CP-eq-lb-eq}, we then know that $(\vect{y}, \vect{z})$ is a certificate for the satisfaction of the \existsCQ-query $\query = \exists \scheduler' \in \schedulers^{\mdp'} \centerdot \prop{\gtrsim}{\scheduler}(\boldLambda)$ in $\mdp'$. As described above, it is directly clear that $\query$ is also satisfied in $\mdp$ and applying \Cref{lemma:exist-CP-eq} \ref{lemma:exist-CP-eq-lb-eq} again, yields the existence of a desired $\vect{y}'$ with $\SM^\top \vect{y}' = \initDistr \land \TM^\top \vect{y}' \gtrsim \boldLambda$. This concludes the proofs.
\qed
\end{proof}

\begin{remark}
Let $\vect{y}(\state)$ for all states $\state \in \states$ be defined as follows:
\[
\vect{y}(\state) = \sum_{\action' \in \actions(\state)} \vect{y}(\state, \action')
\]
From Theorem 3.2 in \cite{etessami_multi-objective_2008} we know that a corresponding memoryless scheduler $\scheduler' \in \mSchedulers^{\mdp'}$ that satisfies $\prop{\gtrsim}{\scheduler'}(\boldLambda)$ can be constructed by setting
\[
\scheduler'(\state)(\action) = \frac{\vect{y}(\state, \action)}{\vect{y}(\state) + \vect{z}(\state)}
\qquad
\scheduler'(\state)(\tau) = \frac{\vect{z}(\state)}{\vect{y}(\state) + \vect{z}(\state)}
\]
for all $(\state, \action) \in \SA$ with $\vect{y}(\state) + \vect{z}(\state) > 0$. For the other states, we define $\scheduler$ to play any available action except $\tau$. Observe that once the $\tau$ action is played, the probability to reach any target is $0$. Thus, we cannot decrease the probability of reaching the target states, if we redistribute the probability of playing the $\tau$ action. Consider the scheduler $\scheduler \in \schedulers^\mdp$ with
\[
\scheduler(\state)(\action) = \scheduler'(\state)(\action) + \frac{\vect{y}(\state, \action)}{\vect{y}(\state)} \cdot \scheduler'(\state)(\tau)
\]
for all $(\state, \action) \in \SA$ with $\vect{y}(\state) > 0$. Due to the observation that not playing the $\tau$ action cannot decrease the probability of reaching the targets, $\scheduler$ also satisfies the query and can also be used as scheduler for $\mdp$ as it does not play the $\tau$ action. Further, we have:
\begin{align*}
\scheduler(\state)(\action) &= \scheduler'(\state)(\action) + \frac{\vect{y}(\state, \action)}{\vect{y}(\state)} \cdot \scheduler'(\state)(\tau) \\
&= \frac{\vect{y}(\state, \action)}{\vect{y}(\state) + \vect{z}(\state)} + \frac{\vect{y}(\state, \action)}{\vect{y}(\state)} \cdot (1 - \frac{\vect{y}(\state)}{\vect{y}(\state) + \vect{z}(\state)})\\
&= \frac{\vect{y}(\state, \action)}{\vect{y}(\state) + \vect{z}(\state)} + \frac{\vect{y}(\state, \action)}{\vect{y}(\state)} - \frac{\vect{y}(\state, \action)}{\vect{y}(\state) + \vect{z}(\state)})\\
&= \frac{\vect{y}(\state, \action)}{\vect{y}(\state)}
\end{align*}
for all states $\state \in \states$ with $\vect{y}(\state) > 0$. This detour shows us that we can directly use a certificate $\vect{y}$ that satisfies $\SM^\top \vect{y} \leq \initDistr \land \TM^\top \vect{y} \gtrsim \boldLambda$ to construct a scheduler $\scheduler$ for $\mdp$.
\end{remark}

Let us now discuss the proof of \Cref{lemma:farkas-universal-DQs}. To this end, we show \Cref{lemma:universal-dq-strict-bounds} and \Cref{lemma:universal-dq-non-strict-bounds} first. \Cref{lemma:farkas-universal-DQs} then follows from these lemmas.

\begin{lemma}[\universalDQ-queries with strict bounds]
\label{lemma:universal-dq-strict-bounds}
Let $\mdp = \mdpRF$ be an MDP in reachability form (possibly with ECs) and $\prop{\bowtie}{\scheduler}(\boldLambda)$ a disjunctive property. Then we have:
\begin{enumerate}[label={(\roman*)}, itemindent=1em]
\item $\forall \scheduler \in \schedulers \centerdot \prop{<}{\scheduler}(\boldLambda) \iff \exists \vect{x} \in \reals^\states \centerdot \exists \vect{z} \in \realsnn^{[k]} \centerdot \SM \vect{x} \geq \TM \vect{z} \land \vect{x}(\init) < \pmb{\lambda}^\top \vect{z}$
\end{enumerate}
and if $\mdp$ is EC-free we also have:
\begin{enumerate}[label={(\roman*)}, itemindent=1em]
\setcounter{enumi}{1}
\item $\forall \scheduler \in \schedulers \centerdot \prop{>}{\scheduler}(\boldLambda) \iff \exists \vect{x} \in \reals^\states \centerdot \exists \vect{z} \in \realsnn^{[k]} \centerdot \SM \vect{x} \leq \TM \vect{z} \land \vect{x}(\init) > \pmb{\lambda}^\top \vect{z}$
\end{enumerate}
\end{lemma}
\begin{proof}
Let us prove the statement for DQs with strict upper bounds.
\begin{align*}
&\phantom{\iff} \forall \scheduler \in \schedulers \centerdot \Disj_{i=1}^{k} \prob_{\mdp, \init}^\scheduler(\eventually  \targetSet_i) < \lambda_i \\ 
& \iff \neg \exists \scheduler \in \schedulers \centerdot \Conj_{i=1}^{k} \prob_{\mdp, \init}^\scheduler(\eventually \targetSet_i) \geq \lambda_i \\
\overset{\text{\Cref{lemma:exist-CP-eq} \ref{lemma:exist-CP-eq-lb-eq}}}&{\iff} \neg \exists \vect{y} \in \realsnn^{\SA} \centerdot \SM^\top \vect{y} = \initDistr \land \TM^\top \vect{y} \geq \pmb{\lambda} \\
& \iff \neg \exists \vect{y} \in \realsnn^{\SA} \centerdot
\begin{pmatrix} \vect{A}^\top \\ -\vect{A}^\top \\ -\vect{T}^\top \end{pmatrix} \vect{y} \leq \begin{pmatrix}\initDistr \\ - \initDistr \\ -\pmb{\lambda} \end{pmatrix} \\
\overset{\text{\Cref{lemma:farkas} \ref{lemma:farkas-1}}}&{\iff}
 \exists \vect{x}_1, \vect{x}_{2} \in \realsnn^S \centerdot \exists \vect{z} \in \realsnn^{[k]} \centerdot \bigl(\vect{A} \ \ -\vect{A} \ \ -\vect{T} \bigl) \begin{pmatrix}  \vect{x}_1 \\ \vect{x}_2 \\ \vect{z} \end{pmatrix} \geq 0 \ \land \\
&\phantom{\iff}  \bigl(\initDistr^\top \ \ -\initDistr^\top \ \ -\pmb{\lambda}^\top \bigl) \begin{pmatrix}  \vect{x}_1 \\ \vect{x}_2 \\ \vect{z} \end{pmatrix} < 0 \\
& \iff \exists \vect{x} \in \reals^S \centerdot \exists \vect{z} \in \realsnn^{[k]} \centerdot \SM \vect{x} \geq \TM \vect{z} \land \initDistr^\top \vect{x} < \pmb{\lambda}^\top \vect{z}
\end{align*}
Observe that $\initDistr^\top \vect{x} = \vect{x}(\init)$. This completes the proof. For strict lower bounds we assume $\mdp$ to be EC-free. Then the proof is analogous and we can apply \Cref{lemma:exist-CP-eq} \ref{lemma:exist-CP-eq-ub-eq}. Note that the statement for lower bounds only holds for EC-free MDPs, because \Cref{lemma:exist-CP-eq} \ref{lemma:exist-CP-eq-lb-eq} relies on EC-freeness.
\qed
\end{proof}

\begin{lemma}[\universalDQ-queries with non-strict bounds]
\label{lemma:universal-dq-non-strict-bounds}
Let $\mdp = \mdpRF$ be an MDP in reachability form (possibly with ECs) and $\prop{\bowtie}{\scheduler}(\boldLambda)$ a disjunctive property. Then we have:
\begin{enumerate}[label={(\roman*)}, itemindent=1em]
\item $\forall \scheduler \in \schedulers \centerdot \prop{\leq}{\scheduler}(\boldLambda) \iff \exists \vect{x} \in \reals^\states \centerdot \exists \vect{z} \in \realsnn^{[k]} \setminus \{\vect{0}\}\centerdot \SM \vect{x} \geq \TM \vect{z} \land \vect{x}(\init) \leq \pmb{\lambda}^\top \vect{z}$
\end{enumerate}
and if $\mdp$ is EC-free we also have:
\begin{enumerate}[label={(\roman*)}, itemindent=1em]
\setcounter{enumi}{1}
\item $\forall \scheduler \in \schedulers \centerdot \prop{\geq}{\scheduler}(\boldLambda) \iff \exists \vect{x} \in \reals^\states \centerdot \exists \vect{z} \in \realsnn^{[k]} \setminus \{\vect{0}\}\centerdot \SM \vect{x} \leq \TM \vect{z} \land \vect{x}(\init) \geq \pmb{\lambda}^\top \vect{z}$
\end{enumerate}
\end{lemma}
\begin{proof}
\label{proof:farkas-universal-DPs-non-strict-bounds}
Again, we only prove the statement for DQs with upper bounds. For lower bounds the proof is analogous.
\begin{align*}
&\phantom{\iff} \forall \scheduler \in \schedulers \centerdot \Disj_{i=1}^{k} \prob_{\mdp, \init}^\scheduler(\eventually \targetSet_i) \leq \lambda_i \\ 
& \iff \neg \exists \scheduler \in \schedulers \centerdot \Conj_{i=1}^{k} \prob_{\mdp, \init}^\scheduler(\eventually  \targetSet_i) > \lambda_i \\
& \iff \neg \exists \varepsilon \in \reals_{>0} \centerdot \exists \scheduler \in \schedulers \centerdot \Conj_{i=1}^{k} \prob_{\mdp, \init}^\scheduler(\eventually  \targetSet_i ) \geq \lambda_i + \varepsilon \\
\overset{\text{\Cref{lemma:exist-CP-eq} \ref{lemma:exist-CP-eq-lb-eq}}}&{\iff}
\neg \exists \varepsilon \in \reals_{>0} \centerdot \exists \vect{y} \in \realsnn^{\SA} \centerdot \SM^\top \vect{y} = \initDistr \land \Conj_{i=1}^{k} \vect{t}_i^\top \vect{y} \geq \lambda_i + \varepsilon \\
& \iff \neg \exists \varepsilon \in \reals_{>0} \centerdot \exists \vect{y} \in \realsnn^{\SA} \centerdot \SM^\top \vect{y} = \initDistr \land \TM^\top \vect{y} \geq \pmb{\lambda} + \vect{1} \cdot \varepsilon
\end{align*}
To apply Farkas' lemma, we need to scale the right-hand side of the equality and inequality with a variable. \sloppy We show that $\SM^\top \vect{y} = \initDistr \land \TM^\top \vect{y} \geq \pmb{\lambda} + \vect{1} \cdot \varepsilon$ has a solution if and only if $\SM^\top \vect{y} = \initDistr \cdot \gamma \land \TM^\top \vect{y} \geq \pmb{\lambda} \cdot \gamma + \vect{1} \cdot \varepsilon$ has a solution, where $\varepsilon \in \reals_{>0}$, $\gamma \in \realsnn$ and $\vect{y} \in \realsnn^{\SA}$. Obviously, the former implies the latter, since we can use the solution of the former and choose $\gamma = 1$ to obtain a solution for the latter. 

For the other direction, suppose we have a solution $\varepsilon$, $\vect{y}$ and $\gamma$ with $\gamma > 0$. Let $\vect{y}' = \frac{\vect{y}}{\gamma}$ and $\varepsilon' = \frac{\varepsilon}{\gamma}$, then we have $\SM^\top \vect{y}' = \SM^\top \frac{\vect{y}}{\gamma} = \frac{\initDistr \cdot \gamma}{\gamma} = \initDistr$ and $\TM^\top \vect{y}' = \TM^\top \frac{\vect{y}}{\gamma} \geq \frac{\pmb{\lambda} \cdot \gamma + \vect{1} \cdot \varepsilon}{\gamma} = \pmb{\lambda} + \vect{1} \cdot \varepsilon'$. Hence, $\vect{y}'$ and $\varepsilon'$ are a solution to the first system. 

Now suppose $\gamma = 0$, so we have $\SM^\top \vect{y} = 0$. Suppose $\vect{y} = 0$, then we get $\TM^\top \vect{y} = 0 \geq \vect{1} \cdot \varepsilon > 0$. Hence $\vect{y} = 0$ cannot hold. Suppose $\vect{y} \neq 0$. The following observation is from the proof of Lemma 3.8 in \cite{jantsch_certificates_2022}. Since we have $\SM^\top \vect{y} = 0$, we also have $\vect{1}^\top\SM^\top \vect{y} = 0$. Observe that $\vect{1}^\top \SM^\top$ corresponds to $1 - \sum_{\state' \in \states} \transMat(\state, \action, \state')$ for every $(\state, \action) \in \SA$. From $\vect{1}^\top\SM^\top \vect{y} = 0$ we have for all $\vect{y}(\state, \action) > 0$ that $\sum_{\state' \in \states} \transMat(\state, \action, \state') = 1$ and $\transMat(\state, \action, \targets) = 0$. This implies $\TM^\top \vect{y} = 0$, again yielding a contradiction and $\gamma \neq 0$ has to hold. Thus we have:
\begin{align*}
&\phantom{\iff} \neg \exists \varepsilon \in \reals_{>0} \centerdot \exists \vect{y} \in \realsnn^{\SA} \centerdot \SM^\top \vect{y} = \initDistr \land \TM^\top \vect{y} \geq \pmb{\lambda} + \vect{1} \cdot \varepsilon \\
& \iff \neg \exists \gamma \in \realsnn \centerdot \exists \varepsilon \in \reals_{>0} \centerdot \exists \vect{y} \in \realsnn^{\SA} \centerdot \SM^\top \vect{y} = \initDistr \cdot \gamma \land \TM^\top \vect{y} \geq \pmb{\lambda} \cdot \gamma + \vect{1} \cdot \varepsilon \\
& \iff \neg \exists \gamma, \varepsilon \in \realsnn \centerdot \exists \vect{y} \in \realsnn^{\SA} \centerdot \begin{pmatrix} \SM^\top & -\initDistr & \vect{0} \\ -\SM^\top & \initDistr & \vect{0} \\ \TM^\top & -\pmb{\lambda} & -\vect{1} \end{pmatrix}
\begin{pmatrix}
\vect{y} \\
\gamma \\
\varepsilon
\end{pmatrix}
\geq 0 \ \land -\varepsilon < 0 \\
\overset{\text{\Cref{lemma:farkas} \ref{lemma:farkas-1}}}&{\iff}
\exists \vect{x}_1, \vect{x}_2 \in \realsnn^{\states} \centerdot \exists \vect{z} \in \realsnn^{[k]} \centerdot \begin{pmatrix} \SM & -\SM & \TM \\ -\initDistr^\top & \initDistr^\top & -\pmb{\lambda}^\top \\ \vect{0}^\top & \vect{0}^\top & -\vect{1}^\top \end{pmatrix} \begin{pmatrix}
\vect{x}_1 \\
\vect{x}_2 \\
\vect{z}
\end{pmatrix} \leq \begin{pmatrix}
\vect{0} \\
0 \\
-1
\end{pmatrix}\\
& \iff \exists \vect{x} \in \reals^{\states} \centerdot \exists \vect{z} \in \realsnn^{[k]} \centerdot \SM \vect{x} \leq - \TM \vect{z} \land - \initDistr^\top \vect{x} \leq \pmb{\lambda}^\top \vect{z} \land \vect{1}^\top\vect{z} \geq 1\\
& \iff \exists \vect{x} \in \reals^{\states} \centerdot \exists \vect{z} \in \realsnn^{[k]} \centerdot \SM \vect{x} \geq \TM \vect{z} \land \initDistr^\top \vect{x} \leq \pmb{\lambda}^\top \vect{z} \land \vect{1}^\top\vect{z} \geq 1
\end{align*}
We claim that $\SM \vect{x} \geq \TM \vect{z} \land \initDistr^\top \vect{x} \leq \pmb{\lambda}^\top \vect{z} \land \vect{1}^\top\vect{z} \geq 1$ has a solution if and only if $\SM \vect{x} \geq \TM \vect{z} \land \initDistr^\top \vect{x} \leq \pmb{\lambda}^\top \vect{z} \land \vect{z} \neq 0$ does. Firstly, the solution to the former is a solution to the latter, since $\vect{1}^\top\vect{z} \geq 1$ implies $\vect{z} \neq 0$. Now let $\vect{x}$ and $\vect{z}$ be a solution of the latter. Since $\vect{z} \neq 0$ and $\vect{z} \geq 0$ there exists $\beta \in \realsnn$ such that $\beta \cdot \vect{1}^\top \vect{z} \geq 1$. Let $\vect{x}' = \vect{x} \cdot \beta$ and $\vect{z}' = \vect{z} \cdot \beta$. Then we get $\SM \vect{x}' = \SM \vect{x} \cdot \beta \geq \TM \vect{z} \cdot \beta = \TM \vect{z}'$ and $\initDistr^\top \vect{x}' = \initDistr^\top \vect{x} \cdot \beta \leq \pmb{\lambda}^\top \vect{z} \cdot \beta = \pmb{\lambda}^\top \vect{z}'$ and by construction $\vect{1}^\top \vect{z} \geq 1$. Hence $\vect{x}'$ and $\vect{z}'$ are a solution to the former system. Altogether this shows the equivalence. Since $\initDistr^\top \vect{x} = \vect{x}(\init)$, this completes the proof.

Again, for lower bounds we assume $\mdp$ to be EC-free. Then the proof is analogous and we apply \Cref{lemma:exist-CP-eq} \ref{lemma:exist-CP-eq-ub-eq} instead of \Cref{lemma:exist-CP-eq} \ref{lemma:exist-CP-eq-lb-eq}.
\qed
\end{proof}
\certificatesUniversalDq*
\begin{proof}
Directly follows from \Cref{lemma:universal-dq-strict-bounds} and \Cref{lemma:universal-dq-non-strict-bounds}.
\qed
\end{proof}
Let us now briefly show how the certificates for \existsDQ-queries and \universalCQ-queries can be derived. To this end, let $\vect{t}_i$ denote the $i$th column of $\vect{T}$ and $\compBowtie \ = \ \leq$ if $\bowtie \ \in \{\geq, > \}$ and $\compBowtie \ = \ \geq$ if $\bowtie \ \in \{\leq, <\}$.
\begin{lemma}[Certificates for \existsDQ-queries]
Let $\mdp = \mdpRF$ be an MDP in reachability form without ECs and let $\targetSet_1, \dots, \targetSet_k$ be target sets. Let $\lambda_i \in [0, 1]$ for all $i \in [k]$ and $\bowtie \ \in \{<, \leq, >, \geq\}$. Then we have:
\[
\exists \scheduler \in \schedulers \centerdot \Disj_{i=1}^{k} \prob_{\mdp}^\scheduler(\eventually  \targetSet_i ) \bowtie \lambda_i
\iff \exists \vect{y} \in \reals^\states \centerdot \Disj_{i=1}^{k} \SM^\top \vect{y} \compBowtie \initDistr \land \vect{t}_i^\top \vect{y} \bowtie \lambda_i
\]
\end{lemma}
\begin{proof}
Observe that we have:
\begin{align*}
\exists \scheduler \in \schedulers \centerdot \Disj_{i=1}^{k} \prob_{\mdp}^\scheduler(\eventually \targetSet_i ) \bowtie \lambda_i \iff \Disj_{i=1}^{k} \exists \scheduler \in \schedulers \centerdot \prob_{\mdp, \init}^\scheduler(\eventually \targetSet_i ) \bowtie \lambda_i
\end{align*}
We can then directly apply the results from the single-objective case \cite{funke_farkas_2020, jantsch_certificates_2022} to each disjunct, thereby yielding the statement.
\qed
\end{proof}
\begin{lemma}[Certificates \universalCQ-queries]
Let $\mdp = \mdpRF$ be an MDP in reachability form without ECs and let $\targetSet_1, \dots, \targetSet_k$ be target sets. Let $\lambda_i \in [0, 1]$ for all $i \in [k]$ and $\bowtie \ \in \{<, \leq, >, \geq\}$. Then we have:
\[
\forall \scheduler \in \schedulers \centerdot \Conj_{i=1}^{k} \prob_{\mdp}^\scheduler(\eventually  \targetSet_i ) \bowtie \lambda_i \iff
\exists \vect{x}_1, \dots, \vect{x}_k \in \reals^\states \centerdot \Conj_{i=1}^{k} \SM \vect{x} \compBowtie \vect{t}_i \land \initDistr^\top \vect{x} \bowtie \lambda_i
\]
\end{lemma}
\begin{proof}
\begin{align*}
\forall \scheduler \in \schedulers \centerdot \Conj_{i=1}^{k} \prob_{\mdp}^\scheduler(\eventually \targetSet_i) \bowtie \lambda_i &\iff \Conj_{i=1}^{k} \forall \scheduler \in \schedulers \centerdot \prob_{\mdp, \init}^\scheduler(\eventually \targetSet_i ) \bowtie \lambda_i
\end{align*}
We then obtain certificate conditions for each conjunct by using results from the single-objective case \cite{funke_farkas_2020, jantsch_certificates_2022}, yielding the statement.
\qed
\end{proof}

\begin{table}[t]
\centering
    \setcellgapes{3pt}
    \makegapedcells
\renewcommand{\tabcolsep}{2mm}
\resizebox{\textwidth}{!}{
\begin{tabular}{|c|c|c|c|c|}
  \cline{3-5}
 \multicolumn{2}{c|}{} & Certificate & Condition & \\ 
 \hline
 \multirow{ 2}{*}{$\exists$} & $\land$ & $\vect{y} \in \realsnn^\SA$ & $\SM^\top \vect{y} \compBowtie \initDistr, \TM^\top \vect{y} \bowtie \pmb{\lambda}$ & \cite{etessami_multi-objective_2008} \\
\cline{2-5}
  & $\lor$  & $\vect{y} \in \realsnn^\SA$ &
$\Disj_{i=1}^k \SM^\top \vect{y} \compBowtie \initDistr \land \vect{t}_i^\top \vect{y} \bowtie \lambda_i$ & \cite{funke_farkas_2020, jantsch_certificates_2022} \\
\hline
  \multirow{ 2}{*}{$\forall$} & $\land$ & $\vect{x}_1, \dots, \vect{x}_k \in \reals^\states$ & $\Conj_{i=1}^k \SM \vect{x} \compBowtie \vect{t}_i \land \initDistr^\top \vect{x} \bowtie \lambda_i$ & \cite{funke_farkas_2020, jantsch_certificates_2022} \\ 
  \cline{2-5}
  & $\lor$  & $\vect{x} \in \reals^\states, \ \begin{aligned}\vect{z} \in \begin{cases}
    \realsnn^{[k]} \setminus \{\vect{0}\}, & \text{if } \bowtie \ \in \{\leq, \geq\}\\
    \realsnn^{[k]}, & \text{else} \\
  \end{cases} \\
\end{aligned}$& $\SM \vect{x} \compBowtie \TM \vect{z} \land \vect{x}(\init) \bowtie \pmb{\lambda}^\top \vect{z}$ & \Cref{lemma:farkas-universal-DQs} \\
\hline
\end{tabular}}
\caption{Farkas certificates and conditions for EC-free MDPs. $\compBowtie =\ \leq$ if $\bowtie \ \in \{\geq, >\}$ and $\compBowtie =\ \geq$ otherwise.}
\label{table:farkas-certificates}
\end{table}
An overview of the certificates and their conditions for EC-free MDPs in reachability form is shown in \Cref{table:farkas-certificates}.

\subsubsection{Farkas certificates and witnessing subsystems} In this section we provide proofs for \Cref{theorem:subsystems-and-lower-bounds} and \Cref{theorem:farkas-support-witnessing-subsystems}. 
\begin{lemma}
\label{lemma:schedulers-subsystems-mdp}
Let $\mathcal{N} = (\states, \actions, \init, \transMat)$ be an MDP. Let $\targets_1, \dots, \targets_k \subseteq \states$ and $\targetSet_1, \dots,\allowbreak \targetSet_\ell \subseteq \states$. Further, let $\vect{r}_1, \dots, \vect{r}_p \in \rationals^\SA$ be reward vectors. Let $\mathcal{N}' = (\states' \union \{\exit\}, \actions, \init, \transMat')$ be a subsystem of $\mathcal{N}$.
\begin{enumerate}[label={(\roman*)}, align=left, leftmargin=*, itemsep=0mm, topsep=0.8mm, parsep=0mm]
\item For every scheduler in $\scheduler \in \schedulers^\mathcal{N}$ there exists a scheduler $\scheduler' \in \schedulers^{\mathcal{N}'}$ such that for all $i \in [k]$ we have $\prob_{\mathcal{N}'}^{\scheduler'}(\eventually \targets_i) \leq \prob_{\mathcal{N}}^\scheduler(\eventually \targets_i)$, for all $j \in [\ell]$ we have $\prob_{\mathcal{N}'}^{\scheduler'}(\globally \targetSet_j) \leq \prob_{\mathcal{N}}^\scheduler(\globally \targetSet_j)$ and for all $h \in [p]$ we have $\expectation[\mathcal{N}']{\scheduler'}{\lrInf{\vect{r'}_h}} \leq \expectation[\mathcal{N}]{\scheduler}{\lrInf{\vect{r}_h}}$ ($\expectation[\mathcal{N}']{\scheduler'}{\lrSup{\vect{r'}_h}} \leq \expectation[\mathcal{N}]{\scheduler}{\lrSup{\vect{r}_h}}$) \label{lemma:subsystem-scheduler-exists}
\item Vice versa, for every scheduler $\scheduler' \in \schedulers^{\mathcal{N}'}$ there exists a scheduler in $\scheduler \in \schedulers^\mathcal{N}$ such that for all $i \in [k]$ we have $\prob_{\mathcal{N}'}^{\scheduler'}(\eventually \targetSet_i) \leq \prob_{\mathcal{N}}^\scheduler(\eventually \targetSet_i)$, for all $j \in [\ell]$ we have $\prob_{\mathcal{N}'}^{\scheduler'}(\globally \targetSet_j) \leq \prob_{\mathcal{N}}^\scheduler(\globally \targetSet_j)$ and for all $h \in [p]$ we have $\expectation[\mathcal{N}']{\scheduler'}{\lrInf{\vect{r'}_h}} \leq \expectation[\mathcal{N}]{\scheduler}{\lrInf{\vect{r}_h}}$ ($ \expectation[\mathcal{N}']{\scheduler'}{\lrSup{\vect{r'}_h}} \leq \expectation[\mathcal{N}]{\scheduler}{\lrSup{\vect{r}_h}}$) \label{lemma:original-scheduler-exists}
\end{enumerate}
\end{lemma}
\begin{proof}
The proof follows the ideas from \cite[Proposition~4.4]{jantsch_certificates_2022}. The set of paths in $\mathcal{N}'$ never visiting $\exit$ are a subset of paths in $\mathcal{N}$, i.e.
\[
\{\Path' \in \paths(\mathcal{N}') \mid \Path' \text{ never visits } \exit \} \subseteq \paths(\mathcal{N})
\]
Further, for all paths $\Path_\exit \in \{\Path' \in \paths(\mathcal{N}') \mid \Path' \text{ visits } \exit \}$, i.e. paths in $\mathcal{N}'$ that visit $\exit$ (and hence stay in $\exit$ forever, as $\exit$ cannot be left), and $\Path \in \paths(\mathcal{N})$ we have $\lrInf{\vect{r}_h'}(\Path_\exit) \leq \lrInf{\vect{r}_h}(\Path)$ ($\lrSup{\vect{r}_h'}(\Path_\exit) \leq \lrSup{\vect{r}_h}(\Path)$) for all $h \in [p]$. Intuitively, by construction of $\vect{r}_h'$, the smallest possible reward is collected in $\exit$ and a path in $\mathcal{N}'$ visiting $\exit$ cannot achieve a higher mean-payoff than any path in $\mathcal{N}$.

Let us prove \ref{lemma:subsystem-scheduler-exists} first. Given a scheduler $\scheduler \in \schedulers^\mathcal{N}$, we choose a scheduler $\scheduler' \in \schedulers^{\mathcal{N}'}$ that behaves like $\scheduler$ in $\states'$ (in $\exit$ the choice does not matter). Recall that the set of actions that are enabled in a state $\state$ in $\mathcal{N}$ and $\mathcal{N}'$ coincide by definition and that once $\exit$ is entered in the subsystem the smallest possible reward is collected. Because the paths in $\mathcal{N}$ under $\scheduler$ and $\mathcal{N}'$ under $\scheduler'$ carry the same probability and the state-action pairs in $\states'$ have the same reward, the statement follows with the observations above.

For \ref{lemma:original-scheduler-exists}, let  $\scheduler' \in \schedulers^{\mathcal{N}'}$ be given. We choose a scheduler $\scheduler$ that behaves like $\scheduler'$ for paths $\Path \in \pathsFin(\mathcal{N})$ that only visit $\states'$. Otherwise, $\scheduler$ is allowed to play any available action. Similarly, the statement then follows.
\qed
\end{proof}

\subsystemsLowerBounds*
\begin{proof}
For \ref{theorem:subsystems-and-lower-bounds-exists}, we can directly apply \Cref{lemma:schedulers-subsystems-mdp} \ref{lemma:original-scheduler-exists}. For \ref{theorem:subsystems-and-lower-bounds-universal} we prove via contraposition, i.e. we show:
\[
\exists \scheduler \in \schedulers^\mdp \centerdot \propAlt{\lesssim}{\scheduler}(\boldLambda) \implies \exists \scheduler' \in \schedulers^{\mdp'} \centerdot \propAlt{\lesssim}{\scheduler'}(\boldLambda)
\]
where $\propAlt{\lesssim}{\scheduler}(\boldLambda)$ is a corresponding conjunctive query if $\prop{\gtrsim}{\scheduler}(\boldLambda)$ is a disjunctive query and $\propAlt{\lesssim}{\scheduler}(\boldLambda)$ is a corresponding disjunctive query if $\prop{\gtrsim}{\scheduler}(\boldLambda)$ is a conjunctive query. Further, we choose $\lesssim \ = \ <$ if $\gtrsim \ = \ \geq$ and $\lesssim \ = \ \leq$ if $\gtrsim \ = \ >$. The statement then follows from \Cref{lemma:schedulers-subsystems-mdp} \ref{lemma:subsystem-scheduler-exists}.
\qed
\end{proof}

\farkasSupportSubsystem*
\begin{proof}~
Let $\SM' = \SM\vert_{\SA'\times\states'} =  \SM_{\mdp_{\states'}}$ and $\TM' = \TM\vert_{\SA'\times [k]} = \TM_{\mdp_{\states'}}$ where $\SA' = \{(\state, \action) \in \SA \mid \state \in \states'\} = \SA_{\mdp_{\states'}}$. Let us prove \ref{theorem:farkas-support-witnessing-subsystems-dq} first. We first note that if there exists $(\vect{x}', \vect{z}') \in \dqPoly{\mdp}{\gtrsim}(\pmb{\lambda})$, then there also exist $(\vect{x}, \vect{z}) \in \dqPoly{\mdp}{\gtrsim}(\pmb{\lambda})$ with $\vect{x} \geq 0$, namely $\vect{x}(\state) = \max\{0, \vect{x}'(\state)\}$ for all $\state \in \states$ and $\vect{z} = \vect{z}'$.
\begin{enumerate}
\item[$\Rightarrow$:] Let $(\vect{x}, \vect{z}) \in \dqPoly{\mdp}{\gtrsim}(\pmb{\lambda})$ with $\vect{x} \geq 0$. Then we have $\SM \vect{x} \leq \TM \vect{z} \land \vect{x}(\init) \gtrsim \pmb{\lambda}^\top \vect{z}$ (and additionally $\vect{z} \neq 0$ if we have non-strict inequalities). Let $\vect{x}' = \vect{x} \vert_{\states'}$ (i.e. $\vect{x}$ restricted to $\states'$). From Lemma 4.22 in \cite{jantsch_certificates_2022} we know that $\SM' \vect{x}' \leq \TM' \vect{z}$ and $\vect{x}'(\init) \gtrsim \pmb{\lambda}^\top \vect{z}$ hold. Intuitively, by omitting columns in $\SA$ (that is columns corresponding to states in $\states\setminus\states'$ and which are thus not in the support of $\vect{x}$) where the corresponding value of $\vect{x}$ is zero does not change the value of the left-hand side. Additionally, omitting rows on both sides also preserves the satisfaction of the inequalities. Consequently, $\vect{x}'$ and $\vect{z}$ are Farkas certificates for the satisfaction of the query in $\mdp_{\states'}$. Using \Cref{lemma:farkas-universal-DQs} we can conclude $\forall \scheduler \in \schedulers \centerdot \Disj_{i=1}^{k} \prob_{\mdp_{\states'}, \init}^\scheduler(\eventually  \targetSet_i ) \gtrsim \lambda_i$.
\item[$\Leftarrow$:] Because $\forall \scheduler \in \schedulers \centerdot \Disj_{i=1}^{k} \prob_{\mdp_{\states'}, \init}^\scheduler(\eventually \targetSet_i) \gtrsim \lambda_i$ holds, we know by \Cref{lemma:farkas-universal-DQs} that there exists $\vect{x}' \in \realsnn^{\states'}$ and $\vect{z} \in \realsnn^{[k]}$ such that $(\vect{x'}, \vect{z}) \in \dqPoly{\mdp_{\states'}}{\gtrsim}(\pmb{\lambda})$, i.e. $\SM' \vect{x}' \leq \TM' \vect{z}$,  $\vect{x}'(\init) \gtrsim \pmb{\lambda}^\top \vect{z}$ and $\vect{1}^\top \vect{z} \leq 1$ (or $\vect{1}^\top \vect{z} = 1$ if we have non-strict inequalities). Let $\vect{x} \in \realsnn^\states$ with $\vect{x}(\state) = \vect{x}'(\state)$ if $\state \in \states'$ and $\vect{x}(\state) = 0$ otherwise. Clearly, we have $\supp{\vect{x}} \subseteq \states'$. Again, applying Lemma 4.22 from \cite{jantsch_certificates_2022} we know that  $\SM \vect{x} \leq \TM \vect{z}$ and  $\vect{x}(\init) \geq \pmb{\lambda}^\top \vect{z}$ hold. Intuitively, adding columns corresponding to states where $\vect{x}$ is zero does not change the left-hand side of the inequalities. Rows corresponding to $(\state, \action) \in \SA$ with $\state \in \states\setminus\states'$ are of the form $- \sum_{\state' \in \states} \transMat(\state, \action, \state') \cdot \vect{x}(\state') \leq \sum_{i=1}^k \sum_{\target \in \targetSet_i} \transMat(\state, \action, \target) \cdot \vect{z}(\targetSet_i)$ because $\vect{x}(\state) = 0$. Since $\vect{x} \geq 0$ and the right-hand side is non-negative, such rows are also satisfied. Thus we have $(\vect{x}, \vect{z}) \in \dqPoly{\mdp}{\gtrsim}(\pmb{\lambda})$.
\end{enumerate}
The proof for \ref{theorem:farkas-support-witnessing-subsystems-cq} is analogous.
\begin{enumerate}
\item[$\Rightarrow$:] Let $\vect{y} \in \cqPoly{\mdp}{\gtrsim}(\boldLambda)$ with $\states' = \stateSupp{\vect{y}}$. For such $\vect{y}$ we have $\SM^\top\vect{y} \leq \initDistr \land \TM^\top \vect{y} \gtrsim \boldLambda$. Now we consider the restriction of $\vect{y}$ to the state action pairs in $\SA'$, i.e. $\vect{y}' = \vect{y} \vert_{\SA'}$. Again, following the reasoning of Lemma 4.22 from \cite{jantsch_certificates_2022} we have that omitting columns of $\SM^\top$ where $\vect{y}$ is zero does not change the value. Similarly, omitting rows preserves the satisfaction of the inequality. Because $\TM^\top \vect{y} = (\TM')^\top \vect{y'}$, we then have $(\SM')^\top\vect{y}' \leq \initDistr \land (\TM')^\top \vect{y}' \gtrsim \boldLambda$. Applying \Cref{lemma:exist-CP} then concludes the proof.
\item[$\Leftarrow$:] Now suppose we have $\states' \subseteq \states$ such that $ \exists \scheduler' \in \schedulers^{\mdp_{\states'}} \centerdot \propAlt{\gtrsim}{\scheduler'}(\boldLambda)$. By \Cref{lemma:exist-CP} we have that there exists $\vect{y}' \in \realsnn^{\SA'}$ such that $(\SM')^\top\vect{y}' \leq \initDistr \land (\TM')^\top \vect{y}' \gtrsim \boldLambda$. Now let $\vect{y} \in \realsnn^\SA$ and we set $\vect{y}(\state, \action) = \vect{y}'(\state, \action)$ if $(\state, \action) \in \SA'$ and $\vect{y}(\state, \action) = 0$ otherwise. Observe that $\stateSupp{\vect{y}} \subseteq \states'$ and for every state $\state \in \states \setminus \states'$ we have $\sum_{\action \in \actions(\state)} \vect{y}(\state, \action) - \sum_{(t, \action)} \transMat(t, \action, \state) \cdot \vect{y}(\state, \action) - \initDistr(\state) \leq 0$, because we have $\sum_{\action \in \actions(\state)} \vect{y}(\state, \action) = 0$. By construction we have $\TM^\top \vect{y} = (\TM')^\top \vect{y'}$. In total, we then have $\SM^\top\vect{y} \leq \initDistr \land \TM^\top \vect{y} \gtrsim \boldLambda$ because adding rows corresponding to $\state \in \states \setminus \states'$ preserves the satisfaction, as well as adding columns where $\vect{y}$ is zero.
\end{enumerate}
\qed
\end{proof}

\subsection{Proofs for \Cref{subsection:transfer-witnesses}}
\label{appendix:transfer}

\subsubsection{Reduction and transfer of subsystems}
Let us discuss the reduction described in \Cref{section:farkas-and-witnesses} and shown in the upper part of \Cref{fig:overview-approach} in detail.
Recall that $\mathcal{N} = (\states_\mathcal{N}, \actions, \bar{\state}, \transMat_\mathcal{N})$ is an arbitrary MDP and $\query_\mathcal{N}$ is a \ReachInv-query containing lower-bounded predicates $\prob_\mathcal{N}^\scheduler(\eventually T_1) \gtrsim \lambda_1, \dots, \prob_\mathcal{N}^\scheduler(\eventually T_k) \gtrsim \lambda_k$ and $\prob_\mathcal{N}^\scheduler(\globally G_1) \gtrsim \xi_1, \dots, \prob_\mathcal{N}^\scheduler(\globally G_\ell) \gtrsim \xi_\ell$.
We follow the construction from \cite{forejt_pareto_2012} for the product MDP $\mdp$. Let $\mdp = (\states, \actions, \init, \transMat)$ where $\states = \states_\mathcal{N} \times 2^{[k]} \times 2^{[\ell]}$, $\init = (\bar{\state}, \emptyset)$ and:
\[
    \transMat((\state, u, v), \action, (\state', u', v')) = 
\begin{cases}
    \transMat_{\mathcal{N}}(\state, \action, \state'),& \text{if } u' = u \union \{i \in [k] \mid \state \in T_i \} \text{ and} \\
    & v' = v \union \{j \in [\ell] \mid \state \in \states\setminus G_j \} \\
    0,              & \text{otherwise.}
\end{cases}
\]
Intuitively, $u$ keeps track of the ``good'' and $v$ the ``bad'' states that have been visited. The predicates can be easily rephrased, i.e.\ $\prob_\mathcal{N}^\scheduler(\eventually T)$ to $\prob_{\mdp}^\scheduler(\eventually (T \times 2^{[k]} \times 2^{[\ell]}))$ and analogously for invariant probabilities. For brevity, we write $\prob_{\mdp}^\scheduler(\eventually T)$ instead.
Because almost all paths eventually stay in a MEC \cite[Theorem~10.120]{baier_principles_2008}, instead of considering $\prob_{\mdp}^\scheduler(\eventually T_i)$, we can consider the probability of eventually staying in a MEC $\mec \in \MECS(\mdp)$ where $T_i$ has already been visited, that is there exists a $(\state, u, v) \in \states(\mec)$ with $i \in u$. Analogously, for $\prob_{\mdp}^\scheduler(\globally \targetSet_j)$ we consider MECs $\mec$ where there exists $(\state, u, v) \in \states(\mec)$ with $j \notin v$. Note that inside MECs, the $u$ and $v$ component of the states are identical. Let $A_i \subseteq \MECS(\mdp)$ denote the set of these MECs for predicates $\prob_{\mdp}^\scheduler(\eventually \targetSet_i)$ and analogously $B_j \subseteq \MECS(\mdp)$ for predicates $\prob_{\mdp}^\scheduler(\globally \targetSet_j)$. We then consider the quotient $\hat{\mdp}$, where reaching $\bot_\mec$ corresponds to staying in MEC $\mec$ forever \cite[Lemma~2.4]{baier_foundations_2022}. Clearly, we can then consider corresponding predicates of the form $\prob_{\hat{\mdp}}^{\hat{\scheduler}}(\eventually \{\bot_\mec \mid \mec \in A_i \})$ and $\prob_{\hat{\mdp}}^{\hat{\scheduler}}(\eventually \{\bot_\mec \mid \mec \in B_j \})$.  
 
Recall that $\iota \colon \states \to \hat{\states}$ maps a state of the product MDP $\mdp$ to the corresponding state in $\hat{\mdp}$.
Given a set of states $\hat{\states}'$ of the MEC quotient, the corresponding set of states in $\mdp$ and $\mathcal{N}$ is given by $\states' = \{(\state, u, v) \in \states \mid \iota((\state, u, v)) \in \hat{\states}' \}$ and $\states_{\mathcal{N}}' = \{\state \in \states_\mathcal{N} \mid \exists u, v \centerdot \iota((\state, u, v)) \in \hat{\states} \}$, respectively.

\transferSubsystem*
\begin{proof}
Let $\hat{\mdp}'$ be the subsystem of $\hat{\mdp}$ induced by a set $\hat{\states}'$ that satisfies $\query_{\hat{\mdp}}$. Let $\mdp'$ be the corresponding subsystem for $\mdp$ induced by $\states'$ and $\mathcal{N}'$ the subsystem of $\mathcal{N}$ induced by $\states_\mathcal{N}'$. Observe that $\hat{\mdp}'$ corresponds to the MEC quotient of $\mdp'$. From \cite[Lemma~2.4]{baier_foundations_2022} we then have that for any scheduler $\hat{\scheduler} \in \schedulers^{\hat{\mdp}'}$, there exists a scheduler $\scheduler \in \schedulers^{\mdp'}$ such that for all $i \in [k]$ and $j \in [\ell]$ we have
\begin{itemize}
\item $\prob_{\hat{\mdp}'}^{\hat{\scheduler}}(\eventually \{\bot_\mec \mid \mec \in A_i \}) = \prob_{\mdp'}^{\scheduler}(\eventually\globally \union_{\mec \in A_i} \states(\mec))$
\item $\prob_{\hat{\mdp}'}^{\hat{\scheduler}}(\eventually \{\bot_\mec \mid \mec \in B_j \}) = \prob_{\mdp'}^\scheduler(\eventually\globally \union_{\mec \in B_j} \states(\mec))$
\end{itemize}
and vice versa. Additionally, for any scheduler $\scheduler \in \schedulers^{\mdp'}$ there exists a scheduler $\scheduler' \in \schedulers^{\mathcal{N}'}$ such that for all $i \in [k]$ and $j \in [\ell]$ we have
\begin{itemize}
\item $\prob_{\mdp'}^{\scheduler}(\eventually\globally \union_{\mec \in A_i} \states(\mec)) \leq \prob_{\mathcal{N}'}^{\scheduler'}(\eventually T_i)$
\item $\prob_{\mdp'}^\scheduler(\eventually\globally \union_{\mec \in B_j} \states(\mec)) \leq \prob_{\mathcal{N}'}^{\scheduler'}(\globally G_j)$
\end{itemize}
and vice versa. This follows from the fact that the set of paths in $\mdp'$ (projected onto states of $\mathcal{N}$) are also present in $\mathcal{N}'$. The statement then follows. \qed

\end{proof}
\subsubsection{Transferring witnessing schedulers}
\dtmcFrequencies*
\begin{proof}
Let $\vect{x} \in \reals^\states$. We consider the linear equation system with $\state \in \states$:
\begin{align*}
\vect{x}(\state) &= \boldDelta(\state) + \sum_{u \in \states} (\vect{x}(u) - \boldMu(u)) \cdot \transMat(u, \state) \\
&= \boldDelta(\state) + \sum_{u \in \states} \vect{x}(u) \cdot \transMat(u, \state) - \sum_{u \in \states} \boldMu(u) \cdot \transMat(u, \state)
\end{align*}
Intuitively, the equations describe the expected frequencies of state $\state$ subtracted by the frequencies that are redirected to the copies of the states. Equivalently, the system can be written in vector-matrix notation as follows:
\begin{equation}
\label{eq:frequencies}
\vect{x} (\vect{I} - \transMat) = \boldDelta - \boldMu \cdot \transMat
\end{equation}
Observe that the steady-state distribution $\boldGamma$ of $\dtmc$ satisfies $\boldGamma (\vect{I} - \transMat) = 0$ and also $\boldGamma > \vect{0}$ since $\dtmc$ is strongly connected. Given a solution $\vect{x}^*$ to \eqref{eq:frequencies}, we know that $\vect{x}^* + r \cdot \boldGamma$ is also a solution to \eqref{eq:frequencies} for all $r \in \reals$. Thus, if there exists a solution, there also exists a solution $\vect{x}^*$ such that $\vect{x}^*(\state) > \boldMu(\state)$ for all states $\state$. Let $\boldLambda(\state) = \frac{\boldMu(\state)}{\vect{x}^*(s)}$, then $\boldLambda(\state) \in [0, 1]$. Setting $\boldMu(\state) = \boldLambda(\state) \cdot \vect{x}^*(\state)$ in \eqref{eq:frequencies} yields for all states $\state$:
\[
\vect{x}^*(\state) = \boldDelta(\state) + \sum_{u \in \states} \vect{x}^*(u) \cdot (1 - \boldLambda(u)) \cdot \transMat(u, \state) = \boldDelta(\state) + \sum_{u \in \states} \vect{x}^*(u) \cdot \transMat_{\dtmc_{\boldLambda}}(u, \state)  
\]
Considering the DTMC $\dtmc_{\boldLambda}$, the expected frequencies $\freq_{\dtmc_{\boldLambda}}(\state)$ are the unique solution of the following system with variables $\vect{z} \in \reals^\states$ and for all states $\state$:
\begin{align*}
\vect{z}(\state) &= \boldDelta(\state) + \sum_{u \in \states} \vect{z}(u) \cdot (1 - \boldLambda(u)) \cdot \transMat(u, \state)\\
\vect{z}(\state') &= \boldLambda(\state) \cdot \vect{z}(\state)
\end{align*}
Thus, $\vect{x}^*(\state) = \freq_{\dtmc_{\boldLambda}}(\state)$ and $\prob_{\dtmc_{\boldLambda}}(\eventually \state') = \freq_{\dtmc_{\boldLambda}}(\state') = \boldLambda(\state) \cdot \vect{x}^*(\state) = \boldMu(\state)$. Hence it remains to be shown that \eqref{eq:frequencies} has a solution. We apply Farkas' lemma (\Cref{lemma:farkas} \ref{lemma:farkas-2}) on \eqref{eq:frequencies} and show that the resulting system (shown below) cannot have a solution.
\begin{equation}
\label{eq:frequencies-farkas}
(\vect{I} - \transMat) \vect{y} = 0 \qquad \text{and} \qquad (\boldDelta - \boldMu \transMat)^\top \vect{y} \neq 0
\end{equation}
Since $\transMat$ is a stochastic matrix (all rows sum up to $1$), we have $(\vect{I} - \transMat) \vect{1} = 0$. Because $\dtmc$ is strongly connected, $\vect{I} - \transMat$ has rank $\card{\states} - 1$ and thus all solutions of $(\vect{I} - \transMat) \vect{y} = 0$ are multiples of $\vect{1}$. Let $\vect{y} = r \cdot \vect{1}$ for some $r \in \reals$. For all distributions $\boldGamma$ we have $\boldGamma^\top \vect{y} = r \cdot \boldGamma^\top \vect{1} = r$. In particular, we have $\boldDelta^\top \vect{y} = r$. Observe that $\boldMu \transMat$ is again a distribution and thus $\transMat^\top \boldMu^\top \vect{y} = r$. We then get $(\boldDelta - \boldMu \transMat)^\top \vect{y} = 0$ contradicting \eqref{eq:frequencies-farkas}. Thus, we can conclude that \eqref{eq:frequencies} has a solution.
\qed
\end{proof}

\section{Proofs for \Cref{section:mean-payoff}}
\label{appendix:mean-payoff}
\certificatesExistsCQ*
\begin{proof}~
\begin{enumerate}
\item[$\Rightarrow$:] Directly follows from \cite[Theorem~4.1]{brazdil_markov_2014} and \cite[Theorem~1]{kretinsky_ltl-constrained_2021}.
\item[$\Leftarrow$:] Let MDP $\mdp' = (\states \union \{\exit\}, \actions \union \{\tau\}, \init, \transMat')$ be the MDP obtained from $\mdp$ by adding a fresh state $\exit$ and transitions to $\exit$ under a fresh action $\tau$ in all states. Let $\states' = \states \union \{\exit\}$. In particular the enabled state action pairs in $\mdp'$ are $\SA' = \SA \union (\states' \times \{\tau\})$. For all $i \in [k]$ we define $\vect{r}_i' \in \rationals_{\geq0}^{\SA'}$ and set $\vect{r}_i'(\state, \action) = \vect{r}_i(\state, \action)$ if $(\state, \action) \in \SA$ and $\vect{r}_i'(\exit, \tau) = \min_{(\state, \action)} \vect{r}_i(\state, \action)$. Because the added transitions under action $\tau$ have lowest possible reward for each reward function, the existence of a strategy $\scheduler'$ for $\mdp'$ that satisfies the mean-payoff constraints implies the existence of satisfying strategy $\scheduler$ for $\mdp$. 

\medskip

\noindent Suppose we have $\vect{x}, \vect{y} \in \realsnn^\SA$ and $\vect{z} \in \realsnn^\states$ that satisfy the constraints. Then for all states $\state \in \states$ let:
We define $\vect{y}', \vect{x}' \in \realsnn^{\SA'}$ for all $(\state, \action) \in \SA'$ as follows:
\[
\vect{y}'(\state, \action) = 
\begin{cases}
	0, &\text{if } \state = \exit\\
    \vect{z}(\state),& \text{if } \state \neq \exit \land \action = \tau \\
    \vect{y}(\state, \action),              & \text{otherwise}
\end{cases}
\]
and
\[
\vect{x}'(\state, \action) = 
\begin{cases}
    \sum_{\state' \in \states} \vect{z}(\state'),& \text{if } \state = \exit \\
    0, & \text{if } \state \neq \exit \land \action = \tau\\
    \vect{x}(\state, \action),              & \text{otherwise}
\end{cases}
\]
We then have for all $\state \in \states$:
\begin{align*}
\initDistr(\state) + \sum_{(\state', \action') \in \SA'} \transMat'(\state', \action', \state) \cdot \vect{y}'(\state', \action') &= \initDistr(\state) + \sum_{(\state', \action') \in \SA} \transMat(\state', \action', \state) \cdot \vect{y}(\state', \action') \\
&= \sum_{\action \in \actions(\state)} \vect{y}(\state, \action) + \vect{x}(\state, \action) + \vect{z}(\state)\\
&= \sum_{\action \in \actions(\state) \union \{\tau\}} \vect{y}'(\state, \action) + \vect{x}'(\state, \action) 
\end{align*}
Further, we have $\initDistr(\exit) + \sum_{(\state', \action') \in \SA'} \transMat'(\state', \action', \exit) \cdot \vect{y}'(\state', \action') = \sum_{\state \in \states} \vect{z}(\state) = \vect{x}'(\exit, \tau)$. Further, we also have for all $\state \in \states$:
\begin{align*}
\sum_{(\state', \action') \in \SA'} \transMat'(\state', \action', \state) \cdot \vect{x}'(\state', \action') &= \sum_{(\state', \action') \in \SA} \transMat(\state', \action', \state) \cdot \vect{x}(\state', \action') \\
&= \sum_{\action \in \actions(\state)} \vect{x}(\state, \action)\\
&= \sum_{\action \in \actions(\state) \union \{\tau\}} \vect{x}'(\state, \action)
\end{align*}
Analogously, we have $\sum_{(\state', \action') \in \SA'} \transMat'(\state', \action', \exit) \cdot \vect{x}'(\state', \action') = \sum_{\state \in \states} \vect{z}(\state) = \vect{x}'(\exit, \tau)$. Lastly, for all $i \in [k]$ we have:
\[
\sum_{(\state, \action) \in \SA'} \vect{x}'(\state, \action) \cdot {\vect{r}_i}'(\state,\action) = \sum_{(\state, \action) \in \SA} \vect{x}(\state, \action) \cdot \vect{r}_i(\state,\action) + \sum_{\state \in \states} \vect{z}(\state) \cdot {\vect{r}_{\min}}(i, \state) \geq \lambda_i
\]
From \cite[Theorem~4.1]{brazdil_markov_2014} and \cite[Theorem~1]{kretinsky_ltl-constrained_2021} we then know that there exists a scheduler $\scheduler' \in \schedulers^{\mdp'}$ such that $\Conj_{i=1}^k \expectation[\mdp', \init]{\scheduler'}{\lrInf{\vect{r}'_i}} \geq \lambda_i$. However, as mentioned above, this also implies the existence of a scheduler $\scheduler \in \schedulers^{\mdp}$ such that $\Conj_{i=1}^k \expectation[\mdp, \init]{\scheduler}{\lrInf{\vect{r}_i}} \geq \lambda_i$.
\end{enumerate}
\qed
\end{proof}
\begin{remark}[Constraints in \cite{brazdil_markov_2014, kretinsky_ltl-constrained_2021}]
The variables $y_s$, constraint 2 and 3 in \cite[Theorem~1]{kretinsky_ltl-constrained_2021} are redundant (as also noted in the work). Let us briefly comment on this redundancy. Consider an MDP where each state has a copy state $\state'$.
Then $y_s$ describes the probability of reaching this copy state $\state'$ \cite{etessami_multi-objective_2008}. The sum $\sum_{s \in S} y_s$ equals $1$ because of the fact that $y_a$ corresponds to the expected frequencies of a scheduler that reaches the absorbing states almost surely \cite[Theorem~3.3.3]{kallenberg_linear_1983} (also see \cite[Remark~3.12]{jantsch_certificates_2022}). From \cite[Lemma~3.8]{jantsch_certificates_2022} it follows that $x_a$ is $0$ for state-action pairs not contained in MECs. Altogether, this makes $y_s$ and constraints 2 and 3 redundant.
\end{remark}
\certificatesForallDQ*
\begin{proof}
We prove the statement via application of Farkas' lemma to the linear system given in \cite[Theorem~4.1]{brazdil_markov_2014}. Observe that we have $\expectation[\mdp, \init]{\scheduler}{\lrSup{\vect{r}_i}} = -\expectation[\mdp, \init]{\scheduler}{\lrInf{-\vect{r}_i}}$ for all schedulers $\scheduler \in \schedulers^\mdp$. The statement can then be shown as follows:
\begin{align*}
\forall \scheduler \in \schedulers \centerdot \Disj_{i=1}^k \expectation{\scheduler}{\lrSup{\vect{r}_i}} \geq \lambda_i
&\iff \neg \exists \scheduler \in \schedulers \centerdot \Conj_{i=1}^k \expectation{\scheduler}{\lrSup{\vect{r}_i}} < \lambda_i\\
&\iff \neg \exists \scheduler \in \schedulers \centerdot \Conj_{i=1}^k -\expectation{\scheduler}{\lrInf{-\vect{r}_i}} < \lambda_i\\
&\iff \neg \exists \scheduler \in \schedulers \centerdot \Conj_{i=1}^k \expectation{\scheduler}{\lrInf{-\vect{r}_i}} > -\lambda_i
\end{align*}
By \cite[Theorem~4.1]{brazdil_markov_2014} and the remark that in \cite{kretinsky_ltl-constrained_2021} that the constraints in \cite[Theorem~4.1]{brazdil_markov_2014} are partly redundant, the existence of a scheduler $\scheduler \in \schedulers$ that satisfies $\Conj_{i=1}^k \expectation{\scheduler}{\lrInf{-\vect{r}_i}} > -\lambda_i$ is equivalent of the satisfiability of the following system of linear equations:
\begin{align*}
\initDistr(\state) + \sum_{(\state', \action') \in \SA} \transMat(\state', \action', \state) \cdot \vect{y}(\state', \action') &= \sum_{\action \in \actions(\state)} \vect{y}(\state, \action) + \vect{x}(\state, \action) &\text{for all } \state \in \states\\
\sum_{(\state', \action') \in \SA} \transMat(\state', \action', \state) \cdot \vect{x}(\state', \action') &= \sum_{\action \in \actions(\state)} \vect{x}(\state, \action) &\text{for all } \state \in \states\\
\sum_{(\state, \action) \in \SA} \vect{x}(\state, \action) \cdot \bigl(- \vect{r}_i(\state,\action) \bigr)&\geq - \lambda_i + \varepsilon &\text{for all } i \in [k]
\end{align*}
where $\vect{x}, \vect{y} \in \realsnn^\SA$ and $\varepsilon > 0$. Equivalently, we can write the system in matrix vector notation as follows:
\begin{equation}
\begin{aligned}
(\vect{D} - \transMat)^\top \vect{y} + \vect{D}^\top \vect{x} &= \initDistr \\
(\vect{D} - \transMat)^\top\vect{x} &= \vect{0}\\
\vect{R}^\top \vect{x} + \vect{1} \cdot \varepsilon &\leq \boldLambda
\end{aligned}
\label{eq:non-scaled-exists}
\end{equation}
Here, $\vect{D} \in \{0, 1\}^{\SA \times \states}$ is defined as $\vect{D}((\state, \action), \state) = 1$ for all $(\state, \action) \in \SA$ and $0$ otherwise. In order to derive certificates and conditions for the universally quantified queries, we instead consider the following system:
\begin{equation}
\begin{aligned}
(\vect{D} - \transMat)^\top \vect{y} + \vect{D}^\top \vect{x} &= \initDistr \cdot \gamma \\
(\vect{D} - \transMat)^\top\vect{x} &= \vect{0}\\
\vect{R}^\top \vect{x} + \vect{1} \cdot \varepsilon &\leq \boldLambda \cdot \gamma\\
\gamma &\geq \varepsilon
\end{aligned}
\label{eq:scaled-exists}
\end{equation}
where again $\vect{x}, \vect{y} \in \realsnn^\SA$ and $\gamma, \varepsilon > 0$. Let us now show the equivalence in terms of satisfiability of those two systems.

\medskip

\noindent \eqref{eq:non-scaled-exists} $\Rightarrow$ \eqref{eq:scaled-exists}: Let $\vect{x}, \vect{y} \in \realsnn^\SA$ and $\varepsilon > 0$ be a solution of \eqref{eq:non-scaled-exists}. We can simply choose $\gamma = 1$ and choose $\varepsilon' = \min\{\gamma, \varepsilon\}$. Then $\vect{x}, \vect{y}$, $\varepsilon'$ and $\gamma$ are a solution to \eqref{eq:scaled-exists}.

\noindent \eqref{eq:non-scaled-exists} $\Leftarrow$ \eqref{eq:scaled-exists}: Let $\vect{x}, \vect{y} \in \realsnn^\SA$ and $\gamma, \varepsilon > 0$ be a solution of \eqref{eq:scaled-exists}. Let $\vect{x}' = \vect{x} \cdot 1 / \gamma$, $\vect{y}' = \vect{y} \cdot 1 / \gamma$ and $\varepsilon' = \varepsilon \cdot 1 / \gamma$. Clearly, we have $(\vect{D} - \transMat)^\top\vect{x}' = \vect{0}$. Further, we have
\[
(\vect{D} - \transMat)^\top \vect{y}' + \vect{D}^\top\vect{x}' = \frac{1}{\gamma} \bigl((\vect{D} - \transMat)^\top \vect{y} + \vect{D}^\top\vect{x} \bigr) = \frac{1}{\gamma} \cdot \initDistr \cdot \gamma = \initDistr,
\]
and
\[
\vect{R}^\top \vect{x'} + \vect{1} \cdot \varepsilon' = \frac{1}{\gamma} \bigl( \vect{R}^\top \vect{x} + \vect{1} \cdot \varepsilon \bigr) \leq \frac{1}{\gamma} \bigl( \boldLambda \cdot \gamma \bigr) \leq \boldLambda
\]
Hence $\vect{x}', \vect{y}'$ and $\varepsilon'$ are a solution to \eqref{eq:non-scaled-exists}.
\medskip

We are concerned with the non-existence of a scheduler and thus equivalently the unsatisfiability of \eqref{eq:scaled-exists}. We can write \eqref{eq:scaled-exists} as follows:
\begin{equation*}
{
\setlength{\arraycolsep}{4pt}
\begin{pmatrix}
(\vect{D} - \transMat)^\top & \vect{D}^\top & -\initDistr & \vect{0}\\
-(\vect{D} - \transMat)^\top & -\vect{D}^\top & \initDistr & \vect{0}\\
\vect{0} & (\vect{D} - \transMat)^\top & \vect{0} & \vect{0}\\
\vect{0} & -(\vect{D} - \transMat)^\top & \vect{0} & \vect{0}\\
\vect{0} & - \vect{R}^\top & \boldLambda & -\vect{1}\\
0 & 0 & 1 & -1
\end{pmatrix}
\begin{pmatrix}
\vect{y}\\
\vect{x}\\
\gamma\\
\varepsilon
\end{pmatrix}
\geq
\begin{pmatrix}
\vect{0}\\
\vect{0}\\
\vect{0}\\
\vect{0}\\
\vect{0}\\
0
\end{pmatrix},
\qquad
\begin{pmatrix}
\vect{0}\\
\vect{0}\\
\vect{0}\\
0\\
-1
\end{pmatrix}^\top
\begin{pmatrix}
\vect{y}\\
\vect{x}\\
\gamma\\
\varepsilon
\end{pmatrix}
<
0
}
\end{equation*}
We then apply Farkas' lemma (\Cref{lemma:farkas} \ref{lemma:farkas-1}), yielding the following system:
\begin{equation*}
{
\setlength{\arraycolsep}{4pt}
\begin{pmatrix}
\vect{D} - \transMat & -(\vect{D} - \transMat) & \vect{0} & \vect{0} & \vect{0} & 0 \\
\vect{D} & -\vect{D} & \vect{D} - \transMat & -(\vect{D} - \transMat) & - \vect{R} & 0 \\
- \initDistr^\top & \initDistr^\top & \vect{0} & \vect{0} & \boldLambda^\top & 1\\
\vect{0} & \vect{0} & \vect{0} & \vect{0} & -\vect{1}^\top & 1
\end{pmatrix}
\begin{pmatrix}
\vect{g}_{+}\\
\vect{g}_{-}\\
\vect{b}_{+}\\
\vect{b}_{-}\\
\vect{z}\\
\beta
\end{pmatrix}
\leq
\begin{pmatrix}
\vect{0}\\
\vect{0}\\
0\\
-1
\end{pmatrix}
}
\end{equation*}
where $\vect{g}_{+}, \vect{g}_{-}, \vect{b}_{+}, \vect{b}_{-} \in \realsnn^\states$, $\vect{z} \in \realsnn^{[k]}$ and $\beta \in \realsnn$. We can further simplify inequalities by defining $\vect{g} \coloneqq \vect{g}_{+} - \vect{g}_{-}$ and $\vect{b} \coloneqq \vect{b}_{+} - \vect{b}_{-}$, yielding:
\begin{equation}
\begin{aligned}
(\vect{D} - \transMat) \vect{g} &\leq \vect{0} \\
\vect{D} \vect{g} + (\vect{D} - \transMat) \vect{b} &\leq \vect{R} \vect{z}\\
\initDistr^\top \vect{g} &\geq \boldLambda^\top \vect{z} + \beta\\
\vect{1}^\top \vect{z} &\geq 1 + \beta
\end{aligned}
\label{eq:farkas-results}
\end{equation}
Observe that any solution of \eqref{eq:farkas-results} where $\beta > 0$ is also a solution to \eqref{eq:farkas-results} when setting $\beta = 0$ because $\boldLambda^\top \vect{z} + \beta \geq \boldLambda^\top \vect{z}$ and $1 + \beta \geq 1$. Hence we can assume $\beta = 0$ and obtain the following conditions:
\begin{align*}
(\vect{D} - \transMat) \vect{g} &\leq \vect{0} \\
\vect{D} \vect{g} + (\vect{D} - \transMat) \vect{b} &\leq \vect{R} \vect{z}\\
\initDistr^\top \vect{g} &\geq \boldLambda^\top \vect{z}\\
\vect{1}^\top \vect{z} &\geq 1
\end{align*}
or equivalently written out explicitly:
\begin{align*}
\vect{g}(\state) &\leq \sum_{\state' \in \states} \transMat(\state, \action, \state') \cdot \vect{g}(\state') &\text{for all } (\state, \action) \in \SA\\
\vect{g}(\state) + \vect{b}(\state) &\leq \sum_{\state' \in \states} \transMat(\state, \action, \state') \cdot \vect{b}(\state') + \sum_{i=1}^k \vect{r}_i(\state, \action) \cdot \vect{z}(i) &\text{for all } (\state, \action) \in \SA\\
\vect{g}(\init) &\geq \sum_{i=1}^k \lambda_i \cdot \vect{z}(i) \qquad \sum_{i=1}^k \vect{z}(i) \geq 1 &
\end{align*}
Observe that if $\sum_{i=1}^k \vect{z}(i) > 1$, then we can simply rescale $\vect{x}$, $\vect{y}$ and $\vect{z}$ by $1 / \sum_{i=1}^k \vect{z}(i)$. Hence, we replace the constraint $\sum_{i=1}^k \vect{z}(i) \geq 1$ with $\sum_{i=1}^k \vect{z}(i) = 1$. Lastly, from \cite[Theorem~4.2.2]{kallenberg_linear_1983} we can conclude that imposing $\vect{g}(\state) \geq \sum_{i=1}^k \vect{z}(i) \cdot \vect{r}_{\min}(i)$ does not change the satisfaction of the system.
\qed
\end{proof}

\begin{lemma}
Let $\mdp =  (\states, \actions, \init, \transMat)$ be an MDP and $\mdp' = (\states' \union \{\exit\}, \actions, \init, \transMat')$ be an induced subsystem of $\mdp$. 
\begin{enumerate}[label={(\roman*)}]
\item If $(\vect{x}', \vect{y}', \vect{z}') \in \mathcal{H}_{\mdp'}^{\mathsf{MP}}(\boldLambda)$, then there exists $(\vect{x}, \vect{y}, \vect{z}) \in \mathcal{H}_{\mdp}^{\mathsf{MP}}(\boldLambda)$ such that $\stateSupp{\vect{x}, \vect{y}} \subseteq \stateSupp{\vect{x}', \vect{y}'}$.\label{theorem:certificate-monotonicity-exists-mp}
\item If $(\vect{g}', \vect{b}', \vect{z}') \in \mathcal{F}_{\mdp'}^{\mathsf{MP}}(\boldLambda)$, then there exists $(\vect{g}, \vect{b}, \vect{z}) \in \mathcal{F}_{\mdp}^{\mathsf{MP}}(\boldLambda)$ such that $\supp{\vect{g} - {\vect{R}_{\min}}\vect{z}} \subseteq \supp{\vect{g}' - {{\vect{R}'}_{\min}}\vect{z}'}$.\label{theorem:certificate-monotonicity-forall-mp}
\end{enumerate}
\label{theorem:certificate-monotonicity}
\end{lemma}
\begin{proof}
\noindent\textbf{Proof of \ref{theorem:certificate-monotonicity-exists-mp}}:
Let $\gamma_i \coloneqq \min_{(\state, \action) \in \SA} \vect{r}_i(\state, \action)$ for all $i \in [k]$. Suppose we have $(\vect{x}', \vect{y}', \vect{z}') \in \mathcal{H}_{\mdp'}^{\mathsf{MP}}$. Then for all $\state \in \states' \union \{\exit\}$ we have:
\begin{align}
\initDistr'(\state) + \sum_{(\state', \action') \in \SA'} \transMat'(\state', \action', \state) \cdot \vect{y}'(\state', \action') &= \sum_{\action \in \actions'(\state)} \vect{y}'(\state, \action) + \vect{x}'(\state, \action) + \vect{z}'(\state) \label{eq:monotonicity-proof-transient}\\
\sum_{(\state', \action') \in \SA'} \transMat'(\state', \action', \state) \cdot \vect{x}'(\state', \action') &= \sum_{\action \in \actions'(\state)} \vect{x}'(\state, \action) \label{eq:monotonicity-proof-recurrent}
\end{align}
and for all $i \in [k]$ we have:
\begin{align}
\sum_{(\state, \action) \in \SA'} \vect{x}'(\state, \action) \cdot \vect{r}'_i(\state,\action) + \sum_{\state \in \states' \union \{\exit\}} \vect{z}'(\state) \cdot \gamma_i &\geq \lambda_i \label{eq:monotonicity-proof-spec}
\end{align}
Let us define $\vect{x}, \vect{y} \in \realsnn^\SA$ as follows:
\[
\vect{x}(\state, \action) = 
\begin{cases}
    \vect{x}'(\state, \action),& \text{if } \state \in \states' \\
    0,              & \text{otherwise}
\end{cases}
\qquad
\vect{y}(\state, \action) = 
\begin{cases}
    \vect{y}'(\state, \action),& \text{if } \state \in \states' \\
    0,              & \text{otherwise}
\end{cases}
\]
Further, we define $\vect{z} \in \reals^\states$ for all $\state \in \states$ as follows:
\[
\vect{z}(\state) = 
\begin{cases}
    \vect{z}'(\state), & \text{if } \state \in \states' \\
    \sum_{\state' \in \states'} \sum_{\action' \in \actions(\state')} \transMat(\state', \action' , \state) \cdot \vect{y}(\state', \action'),              & \text{otherwise}
\end{cases}
\] 

By construction, we have $\stateSupp{\vect{x}} \union \stateSupp{\vect{y}} \subseteq \stateSupp{\vect{x}'} \union \stateSupp{\vect{y}'}$. Now it remains to be shown that $(\vect{x}, \vect{y}, \vect{z}) \in \mathcal{H}_{\mdp}^{\mathsf{MP}}$. To this end, we observe that $\vect{x}'(\state, \action) = 0$ for all states $\state \in \states'$ and $\action \in \actions(\state)$ if $\transMat'(\state, \action, \exit) > 0$ because otherwise \eqref{eq:monotonicity-proof-recurrent} would not be satisfied. Hence $\vect{x}(\state, \action) = 0$ if $\transMat(\state, \action, \state') > \transMat'(\state, \action, \state')$ for some $\state' \in \states$. We then get for all states $\state \in \states \setminus \states'$:
\begin{align*}
\sum_{(\state', \action') \in \SA} \transMat(\state', \action', \state) \cdot \vect{x}(\state', \action') = 0 = \sum_{\action \in \actions(\state)} \vect{x}(\state, \action)
\end{align*}
For all states $\state \in \states'$ we have:
\begin{align*}
\sum_{(\state', \action') \in \SA} \transMat(\state', \action', \state) \cdot \vect{x}(\state', \action') &= \sum_{(\state', \action') \in \SA'} \transMat(\state', \action', \state) \cdot \vect{x}(\state', \action')\\
&= \sum_{(\state', \action') \in \SA'} \transMat'(\state', \action', \state) \cdot \vect{x}(\state', \action')\\
\overset{\eqref{eq:monotonicity-proof-recurrent}}&{=} \sum_{\action \in \actions(\state)} \vect{x}(\state, \action)
\end{align*}
So in total, we have $\sum_{(\state', \action') \in \SA} \transMat(\state', \action', \state) = \sum_{\action \in \actions(\state)} \vect{x}(\state, \action)$ for all $\state \in \states$. Further, for all $\state \in \states'$ we have:
\begin{align*}
&\initDistr(\state) + \sum_{(\state', \action') \in \SA} \transMat(\state', \action', \state) \cdot \vect{y}(\state', \action')\\
={}&\initDistr(\state) + \sum_{(\state', \action') \in \SA'} \transMat'(\state', \action', \state) \cdot \vect{y}'(\state', \action')\\
\overset{\eqref{eq:monotonicity-proof-transient}}{=}{}&\sum_{\action \in \actions(\state)} \vect{y}'(\state, \action) + \vect{x}'(\state, \action) + \vect{z}'(\state)\\
={}&\sum_{\action \in \actions(\state)} \vect{y}(\state, \action) + \vect{x}(\state, \action) + \vect{z}(\state)
\end{align*}
For all $\state \in \states \setminus \states'$ we have:
\begin{align*}
&\initDistr(\state) + \sum_{(\state', \action') \in \SA} \transMat(\state', \action', \state) \cdot \vect{y}(\state', \action')\\
={}&\initDistr(\state) + \sum_{\state' \in \states'} \sum_{\action' \in \actions(\state')} \transMat(\state', \action', \state) \cdot \vect{y}(\state', \action')\\
={}&\vect{z}(\state)\\
={}&\sum_{\action \in \actions(\state)} \vect{y}(\state, \action) + \vect{x}(\state, \action) + \vect{z}(\state)
\end{align*}
Lastly, for all $i \in [k]$ we have:
\begin{align*}
&\sum_{(\state, \action) \in \SA} \vect{x}(\state, \action) \cdot \vect{r}_i(\state,\action) + \sum_{\state \in \states} \vect{z}(\state) \cdot \gamma_i\\
={}&\sum_{\state \in \states'} \sum_{\action \in \actions(\state)} \vect{x}'(\state, \action) \cdot \vect{r}_i'(\state, \action) + \\
&\sum_{\state \in \states \setminus \states'}\sum_{\state' \in \states'} \sum_{\action' \in \actions(\state')} \transMat(\state', \action' , \state) \cdot \vect{y}(\state', \action') \cdot \gamma_i +\\
&\sum_{\state \in \states'} \vect{z}'(\state) \cdot \gamma_i\\
={}&\sum_{\state \in \states'} \sum_{\action \in \actions(\state)} \vect{x}'(\state, \action) \cdot \vect{r}_i'(\state, \action) + \\
&\sum_{\state' \in \states'} \sum_{\action' \in \actions(\state')} \transMat'(\state', \action' , \exit) \cdot \vect{y}(\state', \action') \cdot \gamma_i +\\
&\sum_{\state \in \states'} \vect{z}'(\state) \cdot \gamma_i\\
={}&\sum_{(\state, \action) \in \SA'} \vect{x}'(\state, \action) \cdot \vect{r}_i'(\state, \action) + \sum_{\states' \union \{\exit\}} \vect{z}'(\state) \cdot \gamma_i \geq \lambda_i
\end{align*}
In total, we can therefore conclude that $(\vect{x}, \vect{y}, \vect{z}) \in \mathcal{H}_{\mdp}^{\mathsf{MP}}$.

\bigskip

\noindent\textbf{Proof of \ref{theorem:certificate-monotonicity-forall-mp}}: Let $(\vect{g}', \vect{b}', \vect{z}') \in \mathcal{F}_{\mdp'}^{\mathsf{MP}}(\boldLambda)$. Then the following holds for all $(\state, \action) \in \SA'$:
\begin{align}
\vect{g}'(\state) &\leq \sum_{\state' \in \states' \union \{\exit\}} \transMat'(\state, \action, \state') \cdot \vect{g}'(\state')\\
\vect{g}'(\state) + \vect{b}'(\state) &\leq \sum_{\state' \in \states' \union \{\exit\}} \transMat'(\state, \action, \state') \cdot 
\vect{b}'(\state') + \sum_{i=1}^k \vect{r}'_i(\state, \action) \cdot \vect{z}'(i) \label{eq:proof-monotonicity-forall-bias-gain}
\end{align}
and for all $\state \in \states' \union \{\exit\}$:
\begin{align*}
\vect{g}'(\state) \geq \sum_{i=1}^k \vect{z}'(i) \cdot {\vect{r}_{\min}}(i)
\end{align*}
and for all $i \in [k]$:
\begin{align}
\vect{g}'(\init) &\geq \sum_{i=1}^k \lambda_i \cdot \vect{z}'(i)
\end{align}
Let us define $\vect{z} = \vect{z}'$ and $\vect{g} \in \reals^\states$ as follows:
\[
\vect{g}(\state) = 
\begin{cases}
    \vect{g}'(\state),& \text{if } \state \in \states' \\
    \sum_{i=1}^k \vect{z}'(i) \cdot {\vect{r}_{\min}}(i),              & \text{otherwise}
\end{cases}
\]
By construction we have $\supp{\vect{g} - {\vect{R}_{\min}}\vect{z}} \subseteq \supp{\vect{g}' - {{\vect{R}'}_{\min}}\vect{z}'}$. Analogously, we now need to show that there exists a $\vect{b} \in \reals^\states$ such that $(\vect{g}, \vect{b}, \vect{z}) \in \dqPoly{\mdp}{\geq}(\boldLambda)$. 

\medskip

\noindent For the sake of readability, let us define $\vect{c} \in \realsnn^\SA$ as $\vect{c}(\state, \action) = \sum_{i=1}^k \vect{r}_i(\state, \action) \cdot \vect{z}(i)$ for all $(\state, \action) \in \SA$ and $\vect{c}' \in \realsnn^{\SA'}$ as $\vect{c}'(\state, \action) = \sum_{i=1}^k {\vect{r}_i}'(\state, \action) \cdot \vect{z}(i)$ for all $(\state, \action) \in \SA'$.

We observe that we have $\vect{g}'(\exit) \geq \sum_{i=1}^k \vect{z}'(i) \cdot {\vect{r}_{\min}}'(i) = \sum_{i=1}^k \vect{z}(i) \cdot {\vect{r}_{\min}}(i)$. From \eqref{eq:proof-monotonicity-forall-bias-gain}, we then have for all $\action \in \actions'(\exit)$ that
\[
\vect{g}'(\exit) + \vect{b}'(\exit) \leq \vect{b}'(\exit) + \sum_{i=1}^k \vect{c}'(\exit, \action) = \vect{b}'(\exit) + \sum_{i=1}^k \vect{z}(i) \cdot {\vect{r}_{\min}}(i)
\]
So in total we have $\vect{g}'(\exit) = \sum_{i=1}^k \vect{z}(i) \cdot {\vect{r}_{\min}}(i)$. Then, for all states $\state \in \states'$ and $\action \in \actions(\state)$ we get 
\begin{align*}
\vect{g}(\state) = \vect{g}'(\state) &\leq \sum_{\state' \in \states' \union \{\exit\}} \transMat'(\state, \action, \state') \cdot \vect{g}'(\state')\\
&= \bigl(\sum_{\state' \in \states'} \transMat'(\state, \action, \state') \cdot \vect{g}'(\state') \bigr) + \transMat'(\state, \action, \exit) \cdot \vect{g}'(\exit)\\
&= \bigl(\sum_{\state' \in \states'} \transMat'(\state, \action, \state') \cdot \vect{g}'(\state') \bigr) + \transMat'(\state, \action, \exit) \cdot (\sum_{i=1}^k \vect{z}(i) \cdot {\vect{r}_{\min}}(i))\\
&= \bigl(\sum_{\state' \in \states'} \transMat'(\state, \action, \state') \cdot \vect{g}'(\state') \bigr) + \sum_{\state' \in \states \setminus \states'} \transMat(\state, \action, \state') \cdot \vect{g}(\state')\\
&= \sum_{\state' \in \states} \transMat(\state, \action, \state') \cdot \vect{g}(\state')
\end{align*}
Because $\vect{g}(\state) = \sum_{i=1}^k \vect{z}'(i) \cdot {\vect{r}_{\min}}'(i)$ for $\state \in \states \setminus \states'$ and the $\vect{g}(\state') \geq \sum_{i=1}^k \vect{z}'(i) \cdot {\vect{r}_{\min}}'(i)$ for all $\state' \in \states'$, we can conclude that $\vect{g}(\state) \leq \sum_{\state' \in \states} \transMat(\state, \action, \state') \cdot \vect{g}(\state')$ for all $(\state, \action) \in \SA$. Further, observe that because $\init \in \states'$ we have $\vect{g}(\init) = \vect{g}'(\init) \geq \sum_{i=1}^k \lambda_i \cdot \vect{z}(i)$. 

Now it only remains to be shown that there exists a $\vect{b} \in \reals^\states$ such that for all $(\state, \action) \in \SA$ we have:
\begin{equation}
\vect{g}(\state) + \vect{b}(\state) \leq \sum_{\state' \in \states} \transMat(\state, \action, \state') \cdot \vect{b}(\state') + \vect{c}(\state, \action)
\label{eq:proof-monotonicity-gain-bias-original}
\end{equation}
For the sake of contradiction, suppose this was not the case. Then, by \Cref{lemma:farkas} \ref{lemma:farkas-1} there exists $\vect{x} \in \realsnn^\SA$ such that
\begin{align*}
\sum_{(\state' \action') \in \SA} \transMat(\state', \action', \state) \cdot \vect{x}(\state', \action') &= \sum_{\action \in \actions(\state)} \vect{x}(\state, \action) &\text{for all } \state \in \states\\
\sum_{(\state, \action) \in \SA} \vect{c}(\state, \action) \cdot \vect{x}(\state, \action) &< \sum_{(\state, \action) \in \SA} \vect{x}(\state, \action) \cdot \vect{g}(\state)&
\end{align*}
We note that the first equation describes a recurrent flow. In particular, $\vect{x}(\state, \action) = 0$ if $(\state, \action)$ is not contained in a MEC \cite[Lemma~3.8]{jantsch_certificates_2022}. This allows us to write the second inequality as follows:
\begin{align*}
\sum_{\mec \in \MECS(\mdp)} \sum_{(\state, \action) \in \mec} \vect{c}(\state, \action) \cdot \vect{x}(\state, \action) &= \sum_{(\state, \action) \in \SA} \vect{c}(\state, \action) \cdot \vect{x}(\state, \action)\\
&< \sum_{(\state, \action) \in \SA} \vect{x}(\state, \action) \cdot \vect{g}(\state)\\
&= \sum_{\mec \in \MECS(\mdp)} \sum_{(\state, \action) \in \mec} \vect{x}(\state, \action) \cdot \vect{g}(\state)
\end{align*}
In particular, there exists a MEC $\mec \in \MECS(\mdp)$ such that
\[
\sum_{(\state, \action) \in \mec} \vect{c}(\state, \action) \cdot \vect{x}(\state, \action) < \sum_{(\state, \action) \in \mec} \vect{x}(\state, \action) \cdot \vect{g}(\state)
\]
From \cite[Theorem~4.2.2]{kallenberg_linear_1983}, we know that $\vect{g}(\state) \leq \inf_{\scheduler \in \schedulers^{\mdp'}} \expectation[\mdp',\state]{\scheduler}{\lrSup{\vect{c}'}} \leq \inf_{\scheduler \in \schedulers^\mdp} \expectation[\mdp,\state]{\scheduler}{\lrSup{\vect{c}}}$ for all $\state \in \states'$. 
We write $\scheduler^* \in \schedulers^\mdp$ to denote such optimal scheduler for $\mdp$. Further, observe that all states in a MEC have the same optimal expected mean-payoff. Let us denote this common value by $\nu$, i.e. $\nu = \expectation[\mdp,\state]{\scheduler^*}{\lrSup{\vect{c}}}$ for some $\state \in \states(\mec)$. Then we get:
\[
\sum_{(\state, \action) \in \mec} \vect{c}(\state, \action) \cdot \vect{x}(\state, \action) < \sum_{(\state, \action) \in \mec} \expectation[\mdp,\state]{\scheduler^*}{\lrSup{\vect{c}}} \cdot \vect{x}(\state, \action) = \nu \cdot \sum_{(\state, \action) \in \mec} \vect{x}(\state, \action)
\]
However, this implies that there exists a scheduler $\scheduler \in \mSchedulers^\mdp$ that achieves a strictly lower value inside $\mec$ than $\scheduler^*$. More precisely, inside $\mec$ the scheduler $\scheduler$ ensures that states $\state$ with $\vect{x}(\state, \action) > 0$ for some $(\state, \action) \in \mec$ are reached almost surely and then switches to the strategy $\scheduler_\mec \in \mSchedulers^\mdp$ with $\scheduler_\mec(\state, \action) = \vect{x}(\state, \action) / \sum_{\action \in \actions(\state)} \vect{x}(\state, \action)$ (cf. \cite[Theorem~4.3.1]{kallenberg_linear_1983}). This contradicts the optimality of $\scheduler^*$ and we can conclude that such $\vect{x}$ cannot exist in the first place. Thus there exists $\vect{b} \in \reals^\states$ such that $(\vect{g}, \vect{b}, \vect{z}) \in \mathcal{F}_{\mdp}^{\mathsf{MP}}(\boldLambda)$. \qed
\end{proof}

\certificatesAndSubsystems*
\begin{proof}
The directions from left to right directly follow from \Cref{lemma:certificates-mp-exists}, \Cref{lemma:certificates-mp-universal} and \Cref{theorem:certificate-monotonicity}. Hence, we only need to show that if there exists a certificate for $\mdp$, then the corresponding support induces a subsystem that also satisfies the query. In the following, we write $\mdp' = \mdp_{\states'} = (\states' \union \{\exit\}, \actions, \init, \transMat')$ to denote the subsystem induced by $\states'$. Recall that $\SA' = \SA_{\mdp'} = \{(\state, \action) \in \SA \mid \state \in \states' \} \union \{(\exit, \action) \mid \action \in \actions \}$. 

\medskip

\noindent \textbf{Proof of \ref{theorem:certificate-subsystem-exists}}: Suppose there exists $(\vect{x}, \vect{y}, \vect{z}) \in \mathcal{H}_{\mdp}^{\mathsf{MP}}(\boldLambda)$ such that $\stateSupp{\vect{x}, \vect{y}} \subseteq \states'$. Then for all $\state \in \states' \subseteq \states$ we have:
\begin{align*}
\initDistr(\state) + \sum_{(\state', \action') \in \SA} \transMat(\state', \action', \state) \cdot \vect{y}(\state', \action') &= \sum_{\action \in \actions(\state)} \vect{y}(\state, \action) + \vect{x}(\state, \action) + \vect{z}(\state) \\
\sum_{(\state', \action') \in \SA} \transMat(\state', \action', \state) \cdot \vect{x}(\state', \action') &= \sum_{\action \in \actions(\state)} \vect{x}(\state, \action)
\end{align*}
and for all $i \in [k]$ we have:
\begin{align*}
\sum_{(\state, \action) \in \SA} \vect{x}(\state, \action) \cdot \vect{r}_i(\state,\action) + \sum_{\state \in \states} \vect{z}(\state) \cdot {\vect{r}_{\min}}(i) &\geq \lambda_i
\end{align*}
Let be an arbitrary $a' \in \actions(\exit)$ and let us define $\vect{x}', \vect{y}' \in \realsnn^{\SA'}$ as follows:
\begin{align*}
\vect{y}'(\state, \action) = 
\begin{cases}
    0,& \text{if } \state = \exit \\
    \vect{y}(\state, \action),              & \text{otherwise}
\end{cases}
\\
\vect{x}'(\state, \action) = 
\begin{cases}
    \sum_{(\state', \action') \in \SA'} \transMat'(\state', \action', \exit) \cdot \vect{y}'(\state', \action') , & \text{if } \state = \exit \land a = a' \\
    0 , & \text{if } \state = \exit \land a \neq a' \\
    \vect{x}(\state, \action), & \text{otherwise}
\end{cases}
\end{align*}
Further, let $\vect{z}' \in \realsnn^{\states' \union \{\exit\}}$ with $\vect{z}'(\state) = \vect{z}(\state)$ for $\state \in \states'$ and $\vect{z}'(\exit) = 0$. We now show that the constructed vectors $(\vect{x}', \vect{y}', \vect{z}') \in \mathcal{H}_{\mdp'}^{\mathsf{MP}}(\boldLambda)$. We then get for all $\state \in \states'$:
\begin{align*}
\initDistr'(\state) + \sum_{(\state', \action') \in \SA'} \transMat'(\state', \action', \state) \cdot \vect{y}'(\state', \action') &= \initDistr(\state) + \sum_{(\state', \action') \in \SA'} \transMat(\state, \action, \state) \cdot \vect{y}(\state', \action')\\
&= \sum_{\action \in \actions(\state)} \vect{y}(\state, \action) + \vect{x}(\state, \action)+ \vect{z}(\state) \\
&= \sum_{\action \in \actions(\state)} \vect{y}'(\state, \action) + \vect{x}'(\state, \action) + \vect{z}'(\state)
\end{align*}
Further, we have $\initDistr'(\exit) + \sum_{(\state', \action') \in \SA'} \transMat'(\state', \action', \exit) \cdot \vect{y}'(\state', \action') = \sum_{\action \in \actions(\bot)} \vect{y}'(\exit, \action) + \vect{x}'(\exit, \action)$. Similarly, for all states $\state \in \states'$:
\begin{align*}
\sum_{(\state', \action') \in \SA'} \transMat'(\state', \action', \state) \cdot \vect{x}'(\state', \action') &= \sum_{(\state', \action') \in \SA} \transMat(\state', \action', \state) \cdot \vect{x}(\state', \action')\\
&= \sum_{\action \in \actions(\state)} \vect{x}(\state, \action)\\
&= \sum_{\action \in \actions(\state)} \vect{x}'(\state, \action)
\end{align*}
Again, we have $\sum_{(\state', \action') \in \SA'} \transMat'(\state', \action', \exit) \cdot \vect{x}'(\state', \action') = \sum_{\action \in \actions(\exit)}\vect{x}'(\exit, \action)$. Lastly, we have for all $i \in [k]$:
\begin{align*}
{}&\sum_{(\state, \action) \in \SA'} \vect{x}'(\state, \action) \cdot \vect{r}'_i(\state,\action) + \sum_{\state \in \states' \union \{\exit\}} \vect{z}'(\state) \cdot {\vect{r}_{\min}}'(i) \\
={}&\sum_{(\state, \action) \in \SA} \vect{x}(\state, \action) \cdot \vect{r}_i(\state,\action) + \sum_{\state \in \states} \vect{z}(\state) \cdot {\vect{r}_{\min}}(i)\\
\geq{}& \lambda_i
\end{align*}
Hence $(\vect{x}', \vect{y}') \in \cqPoly{\mdp'}{\geq}(\boldLambda)$ and by \Cref{lemma:certificates-mp-exists} the statement follows.

\medskip

\noindent \textbf{Proof of \ref{theorem:certificate-subsystem-forall}}: Let $(\vect{g}, \vect{b}, \vect{z}) \in \dqPoly{\mdp}{\geq}(\boldLambda)$ such that $\supp{\vect{g} - {\vect{R}_{\min}} \vect{z}} \subseteq \states'$. Then we have:
\begin{align*}
\vect{g}(\state) &\leq \sum_{\state' \in \states} \transMat(\state, \action, \state') \cdot \vect{g}(\state') &\text{for all } (\state, \action) \in \SA\\
\vect{g}(\state) + \vect{b}(\state) &\leq \sum_{\state' \in \states} \transMat(\state, \action, \state') \cdot \vect{b}(\state') + \sum_{i=1}^k \vect{r}_i(\state, \action) \cdot \vect{z}(i) &\text{for all } (\state, \action) \in \SA\\
\vect{g}(\init) &\geq \sum_{i=1}^k \lambda_i \cdot \vect{z}(i) \qquad \sum_{i=1}^k \vect{z}(i) \geq 1 &\\
\vect{g}(\state) &\geq \sum_{i=1}^k \vect{z}(i) \cdot {\vect{r}_{\min}}(i) & \text{for all } \state \in \states
\end{align*}
We now construct corresponding $\vect{g}' \in \realsnn^{\states' \union \{\exit\}}$, $\vect{b}' \in \reals^{\states' \union \{\exit\}}$ and $\vect{z}' \in \realsnn^{[k]}$ and show that $(\vect{g}', \vect{b}', \vect{z}') \in \dqPoly{\mdp'}{\geq}(\boldLambda)$. We set $\vect{z}' = \vect{z}$ and define $\vect{g}'$ and $\vect{b}'$ as follows:
\[
\vect{g}'(\state) = 
\begin{cases}
    \sum_{i=1}^k \vect{z}(i) \cdot {\vect{r}_{\min}}(i), & \text{if } \state = \exit \\
    \vect{g}(\state), & \text{otherwise}
\end{cases}
\qquad
\vect{b}'(\state) = 
\begin{cases}
    \max_{\state' \in \states} \vect{b}(\state'),& \text{if } \state = \exit \\
    \vect{b}(\state),              & \text{otherwise}
\end{cases}
\]
We directly see that
\begin{align*}
\vect{g}'(\state) = \vect{g}(\state) &\leq \sum_{\state' \in \states} \transMat(\state, \action, \state') \cdot \vect{g}(\state')\\
&= \sum_{\state' \in \states'} \transMat(\state, \action, \state') \cdot \vect{g}(\state') + \sum_{\state' \in \states \setminus \states'}\transMat(\state, \action, \state') \cdot \vect{g}(\state')\\
&= \sum_{\state' \in \states'} \transMat(\state, \action, \state') \cdot \vect{g}(\state') + \sum_{\state' \in \states \setminus \states'}\transMat(\state, \action, \state') \cdot (\sum_{i=1}^k \vect{z}(i) \cdot {\vect{r}_{\min}}(i))\\
&= \sum_{\state' \in \states' \union \{\exit\}} \transMat'(\state, \action, \state') \cdot \vect{g}'(\state')
\end{align*}
for all $(\state, \action) \in \SA'$ with $\state \neq \exit$. Let us define $\vect{c}(\state, \action) = \sum_{i=1}^k \vect{r}_i(\state, \action) \cdot \vect{z}(i)$ for all $(\state, \action) \in \SA'$. Then, for all $(\state, \action) \in \SA'$ with $\state \neq \exit$ we have:
\begin{align*}
\vect{g}'(\state) + \vect{b}'(\state) &= \vect{g}(\state) + \vect{b}(\state)\\
&\leq \sum_{\state' \in \states} \transMat(\state, \action, \state') \cdot \vect{b}(\state') + \vect{c}(\state, \action)\\
&= \sum_{\state' \in \states'} \transMat(\state, \action, \state') \cdot \vect{b}(\state') + \sum_{\state' \in \states\setminus\states'} \transMat(\state, \action, \state') \cdot \vect{b}(\state') + \vect{c}(\state, \action)\\
&\leq \sum_{\state' \in \states'} \transMat(\state, \action, \state') \cdot \vect{b}(\state') + \sum_{\state' \in \states\setminus\states'} \transMat(\state, \action, \state') \cdot (\max_{\state'' \in \states} \vect{b}(\state''))+ \vect{c}(\state, \action)\\
&= \sum_{\state' \in \states'} \transMat(\state, \action, \state') \cdot \vect{b}'(\state') + \sum_{\state' \in \states\setminus\states'} \transMat(\state, \action, \state') \cdot \vect{b}'(\exit)+ \vect{c}(\state, \action)\\
&= \sum_{\state' \in \states' \union \{\exit\}} \transMat'(\state, \action, \state') \cdot \vect{b}'(\state') + \vect{c}(\state, \action)
\end{align*}
For $(\exit, \action) \in \SA'$ we have
\begin{align*}
\vect{g}'(\exit) + \vect{b}'(\exit)
&=\sum_{i=1}^k \vect{z}(i) \cdot {\vect{r}_{\min}}(i) + \transMat'(\exit, \action, \exit) \cdot \vect{b}'(\exit)\\
&\leq \sum_{i=1}^k \vect{z}'(i) \cdot \vect{r}'_i(\exit, \action) + \sum_{\state' \in \states' \union \{\exit\}} \transMat'(\exit, \action, \state') \cdot \vect{b}'(\state') 
\end{align*}
Lastly, we have $\vect{g}'(\init) = \vect{g}(\init)$. With that, we can conclude $(\vect{g}', \vect{b}', \vect{z}') \in \dqPoly{\mdp'}{\geq}(\boldLambda)$ and with \Cref{lemma:certificates-mp-universal} the statement follows. \qed
\end{proof}

\section{MILPs for Finding Witnessing Subsystems}
\label{appendix:big-m}
\subsection{MILPs for Reachability}
Recall the MILPs in \Cref{fig:reachability-milps}. We now touch upon the choice of $M$. 
\paragraph{MILPs for \universalDQ-queries.}
Firstly, we note that for the MILP of \universalDQ-queries, we can impose the additional constraint $\sum_{i \in [k]} \vect{z}(i) = 1$ if $\gtrsim \ = \ \geq$ and $\sum_{i \in [k]} \vect{z}(i) \leq 1$ if $\gtrsim \ = \ >$. For $\geq$, observe that given a certificate $(\vect{x}, \vect{z}) \in \dqPoly{\mdp}{\geq}(\pmb{\lambda})$, we can simply rescale with $\gamma = \frac{1}{\sum_{i \in [k]} \vect{z}(i)}$, i.e. $(\gamma \cdot \vect{x}, \gamma \cdot \vect{z}) \in \dqPoly{\mdp}{\gtrsim}(\pmb{\lambda})$, $\sum_{i \in [k]} \gamma \cdot \vect{z}(i) = 1$ and $\supp{\gamma \cdot \vect{x}} = \supp{\vect{x}}$. Analogously, we proceed for $>$. Imposing these additional constraints, ensures that $\vect{x}$ is bounded and an upper bound can be found via LP \cite[Theorem~3.4]{de_alfaro_formal_1997}. As a consequence of \cite[Theorem~3.4]{de_alfaro_formal_1997}, we can also simply choose $k$ as upper bound.
\paragraph{MILPs for \existsCQ-queries.}
Unlike in the single-objective setting \cite{funke_farkas_2020, jantsch_certificates_2022}, the set $\cqPoly{\mdp}{\gtrsim}(\boldLambda)$ is generally unbounded (see \cite[Example~4.3]{jantsch_certificates_2022}). Note that if a minimal witnessing subsystem is given, its certificate can be easily determined and can serve as upper bound. Obviously, it is thus difficult to determine an upper bound $M$ a priori. Here, we resort to indicator constraints, i.e.\ constraints of the form $\gamma(\state) = 0 \implies \vect{y}(\state, \action) = 0$. These constraints are supported by Gurobi \cite{gurobi_optimization_llc_gurobi_2023}.
\subsection{MILPs for Mean-Payoff}
To find minimal witnessing subsystem for mean-payoff queries, we can consider the MILPs shown in \Cref{subfig:mp-milps-universal}. Like for reachability, we use the Big-$M$ encoding. Let us briefly discuss the choice of $M$. Like for reachability, we can impose the additional constraints on $\vect{z}$. Then $\vect{g}$ can again be bounded, e.g. by considering the absolute sum of the smallest and largest rewards (see e.g. \cite{kallenberg_linear_1983}). It is well known that $\vect{x}$ is bounded from above by $\vect{1}$, see e.g. \cite{kretinsky_ltl-constrained_2021}. For the MILP for $\existsCQ$-queries, $\vect{y}$ is again generally unbounded. Here, we also resort to indicator constraints.
\begin{figure}[!t]
\centering
\scriptsize{
\begin{subfigure}[t]{0.475\textwidth}
\centering
$\begin{aligned}
&\text{min } \sum_{\state \in \states} \boldGamma(\state) \text{  subject to:} \\
&\boldGamma \in \{0, 1\}^\states \text{ and } (\vect{x}, \vect{y}, \vect{z}) \in \mathcal{H}_{\mdp}^{\mathsf{MP}}(\boldLambda)\\
&\forall (\state, \action) \in \SA \centerdot \vect{x}(\state, \action) \leq \boldGamma(\state) \cdot M\\
&\forall (\state, \action) \in \SA \centerdot \vect{y}(\state, \action) \leq \boldGamma(\state) \cdot M
\end{aligned}$
\caption{\scriptsize{MILP for $\existsCQ$-mean-payoff queries}}
\label{subfig:mp-milps-exists}
\end{subfigure}%
\hfill
\begin{subfigure}[t]{0.475\textwidth}
\centering
$\begin{aligned}
&\text{min } \sum_{\state \in \state} \boldGamma(\state)\text{  subject to:}\\
&\boldGamma \in \{0, 1\}^\states \text{ and } (\vect{g}, \vect{b}, \vect{z}) \in \mathcal{F}_{\mdp}^{\mathsf{MP}}(\boldLambda)\\
&\forall \state \in \states \centerdot \vect{g}(\state) - \sum_{i=1}^k \vect{z}(i) \cdot \vect{r}_{\min}(i) \leq \boldGamma(\state) \cdot M
\end{aligned}$
\caption{\scriptsize{MILP for $\universalDQ$-mean-payoff queries}}
\label{subfig:mp-milps-universal}
\end{subfigure}%
}
\caption{MILPs for finding minimal witnessing subsystems for mean-payoff queries.}
\end{figure}

\section{Supplementary Material for \Cref{section:experiments}}
\label{appendix:experiments}
Our implementation, experiments and results are made available on Zenodo~\cite{baier_certificates_2024}.

\medskip

\noindent\textbf{Storm results.} The runtimes of \textsc{Storm} in seconds are shown in \Cref{table:storm-results}. We remark that we verified \universalDQ-queries $\query$ by considering the dual \existsCQ-queries $\neg\query$. Note that for some queries, we encountered an error, denoted with \textbf{err}. Note that we were unable to verify the queries of \textsf{zero} with \textsc{Storm}, as the queries were not supported. We refer to the log files in \cite{baier_certificates_2024}. Lastly, we note that the \textsc{Storm} build time is faster than the build time of our implementation, because we have implemented the product construction in Python. The reason is that \textsc{Storm}'s product construction is not available through its Python API.
\begin{table}
\centering
\setlength{\tabcolsep}{4pt}
\def\arraystretch{1}%
\scriptsize
\begin{tabular}{llll|rr}
Model  & Type & $k$ & \# & \textsc{Storm} build time & \textsc{Storm} verification \\
\hline
\multirow[c]{2}{*}{\textsf{coin3}} & $(\exists, \land)$ & 2 & 5 & 0.012 & 0.002 \\
 & $(\forall, \lor)$ & 2 & 5 & 0.012 & 0.002 \\
\multirow[c]{2}{*}{\textsf{coin4}} & $(\exists, \land)$ & 2 & 5 & 0.013 & 0.004 \\
 & $(\forall, \lor)$ & 2 & 5 & 0.012 & 0.003 \\
\multirow[c]{2}{*}{\textsf{coin5}} & $(\exists, \land)$ & 2 & 5 & 0.012 & 0.005 \\
 & $(\forall, \lor)$ & 2 & 5 & 0.014 & 0.005 \\
 \hline
\multirow[c]{2}{*}{\textsf{csn3}} & $(\exists, \land)$ & 3 & 1 & 0.014 & 0.042 \\
 & $(\forall, \lor)$ & 3 & 1 & 0.015 & \textbf{err} \\
\multirow[c]{2}{*}{\textsf{csn4}} & $(\exists, \land)$ & 4 & 1 & 0.029 & 0.063 \\
 & $(\forall, \lor)$ & 4 & 1 & 0.030 & \textbf{err} \\
\multirow[c]{2}{*}{\textsf{csn5}} & $(\exists, \land)$ & 5 & 1 & 0.121 & 0.242 \\
 & $(\forall, \lor)$ & 5 & 1 & 0.115 & \textbf{err} \\
 \hline
\multirow[c]{2}{*}{\textsf{fire3}} & $(\exists, \land)$ & 2 & 5 & 0.051 & 0.022 \\
 & $(\forall, \lor)$ & 2 & 5 & 0.050 & 0.028 \\
\multirow[c]{2}{*}{\textsf{fire6}} & $(\exists, \land)$ & 2 & 5 & 0.094 & 0.042 \\
 & $(\forall, \lor)$ & 2 & 5 & 0.093 & 0.052 \\
\multirow[c]{2}{*}{\textsf{fire9}} & $(\exists, \land)$ & 2 & 5 & 0.154 & 0.073 \\
 & $(\forall, \lor)$ & 2 & 5 & 0.156 & 0.089 \\
\hline\hline
\multirow[c]{2}{*}{\textsf{csn3}} & $(\exists, \land)$ & 3 & 2 & 0.013 & 0.037 \\
 & $(\forall, \lor)$ & 3 & 2 & 0.013 & 0.037 \\
\multirow[c]{2}{*}{\textsf{csn4}} & $(\exists, \land)$ & 4 & 2 & 0.020 & 0.082 \\
 & $(\forall, \lor)$ & 4 & 2 & 0.020 & 0.081 \\
\multirow[c]{2}{*}{\textsf{csn5}} & $(\exists, \land)$ & 5 & 2 & 0.063 & 0.321 \\
 & $(\forall, \lor)$ & 5 & 2 & 0.067 & 0.318 \\
 \hline
\multirow[c]{2}{*}{\textsf{phil3}} & $(\exists, \land)$ & 2 & 3 & 0.016 & 0.036 \\
 & $(\forall, \lor)$ & 2 & 3 & 0.017 & 0.020 \\
\multirow[c]{2}{*}{\textsf{phil4}} & $(\exists, \land)$ & 2 & 3 & 0.069 & 0.773 \\
 & $(\forall, \lor)$ & 2 & 3 & 0.070 & 0.224 \\
 \hline
\multirow[c]{2}{*}{\textsf{sen1}} & $(\exists, \land)$ & 3 & 1 & 0.014 & 0.036 \\
 & $(\forall, \lor)$ & 3 & 1 & 0.014 & 0.035 \\
\multirow[c]{2}{*}{\textsf{sen2}} & $(\exists, \land)$ & 3 & 1 & 0.070 & 0.437 \\
 & $(\forall, \lor)$ & 3 & 1 & 0.068 & 0.368
\end{tabular}
\caption{\textsc{Storm} runtimes}
\label{table:storm-results}
\end{table}

\medskip

\noindent\textbf{Sizes of witnessing subsystems.} Recall that in our experiments we have considered queries with $5$ different bounds for the consensus and firewire models. More specifically, for firewire we consider the labels and queries:
\begin{verbatim}
label "done1" = (s1=8);
label "done2" = (s1=7);
\end{verbatim}
\begin{itemize}
\item $\exists \scheduler \in \schedulers \centerdot \prob^\scheduler(\eventually \texttt{"done1"}) \geq \lambda \land \prob^\scheduler(\eventually \texttt{"done2"}) \geq \lambda$
\item $\forall \scheduler \in \schedulers \centerdot \prob^\scheduler(\eventually \texttt{"done1"}) \geq \lambda \lor \prob^\scheduler(\eventually \texttt{"done2"}) \geq \lambda$
\end{itemize}
where $\lambda \in \{0.01, 0.1325, 0.255, 0.3775, 0.5\}$. For consensus, we consider the labels and queries:
\begin{verbatim}
label "finish1" = pc1=3 & pc2=3 & coin1=1 & coin2=1;
label "finish2" = pc1=3 & pc2=3 & coin1=0 & coin2=0;
\end{verbatim}
\begin{itemize}
\item $\exists \scheduler \in \schedulers \centerdot \prob^\scheduler(\eventually \texttt{"finish1"}) \geq \lambda \land \prob^\scheduler(\eventually \texttt{"finish2"}) \geq \lambda$
\item $\forall \scheduler \in \schedulers \centerdot \prob^\scheduler(\eventually \texttt{"finish1"}) \geq \lambda \lor \prob^\scheduler(\eventually \texttt{"finish2"}) \geq \lambda$
\end{itemize}
where $\lambda \in \{0.05, 0.1125, 0.175, 0.2375, 0.3\}$.
Our implementation computes the witnessing subsystems for theses queries using our MILP approach and returns the best solution that has been found after the time limit. The sizes of the subsystems (relative to the original MDP) are shown in \Cref{fig:witnessing-subsystem-sizes-firewire} and \Cref{fig:witnessing-subsystem-sizes-consensus}. We observe that the subsystems for \universalDQ-queries are significantly larger than for \existsCQ-queries. Additionally, the bound $\lambda$ has a significant influence on the size, particularly for \universalDQ-queries.
\begin{figure}[!h]
\centering
\begin{subfigure}{0.48\textwidth}
\centering
\includegraphics[scale=0.44]{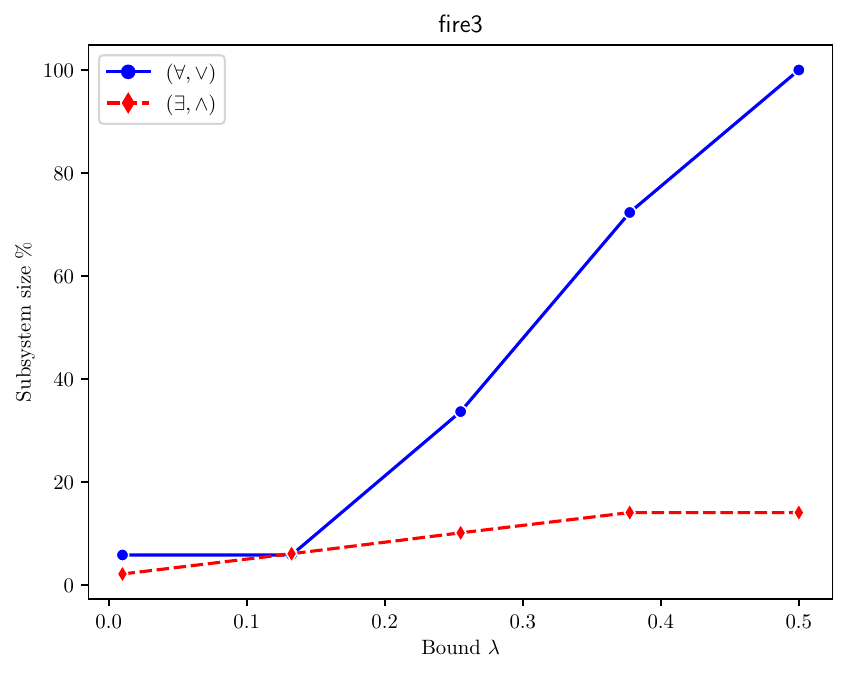}
\caption{\textsf{fire3}}
\end{subfigure}
\hfill
\begin{subfigure}{0.48\textwidth}
\centering
\includegraphics[scale=0.44]{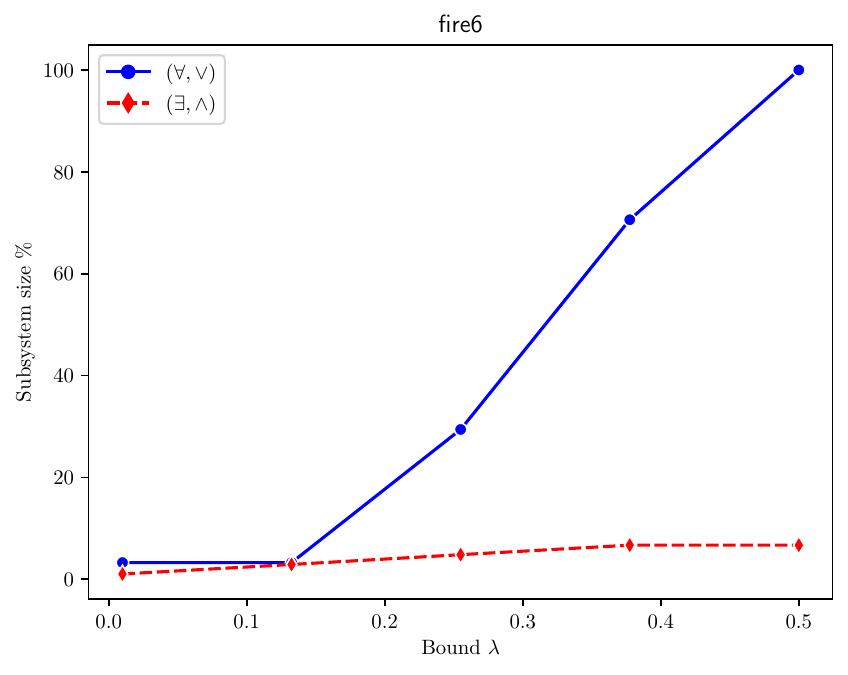}
\caption{\textsf{fire6}}
\end{subfigure}

\begin{subfigure}{1\textwidth}
\centering
\includegraphics[scale=0.44]{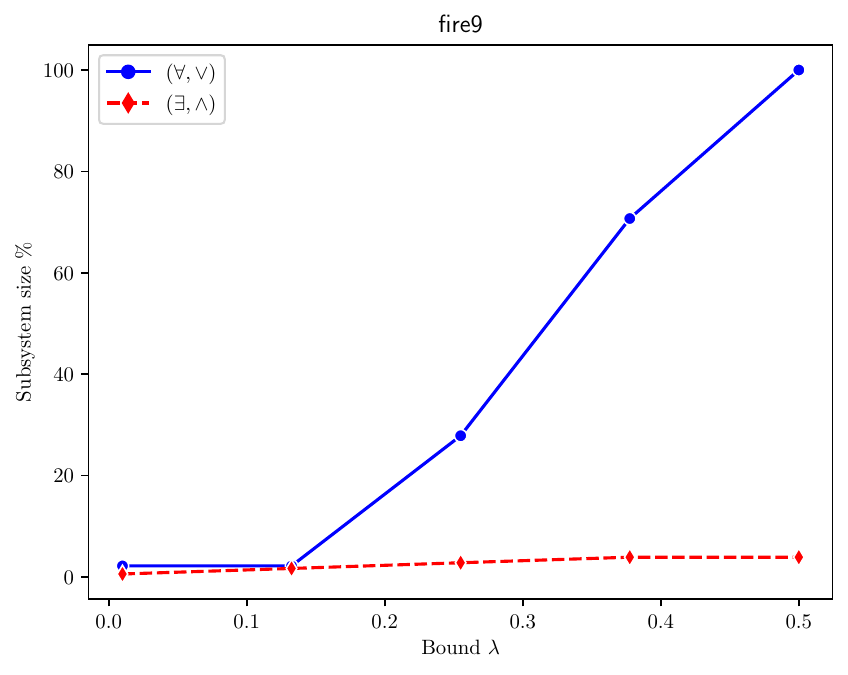}
\caption{\textsf{fire9}}
\end{subfigure}
\caption{Sizes of witnessing subsystems (relative to original MDP) for firewire.}
\label{fig:witnessing-subsystem-sizes-firewire}
\end{figure}
\begin{figure}[!h]
\centering
\begin{subfigure}{0.48\textwidth}
\centering
\includegraphics[scale=0.44]{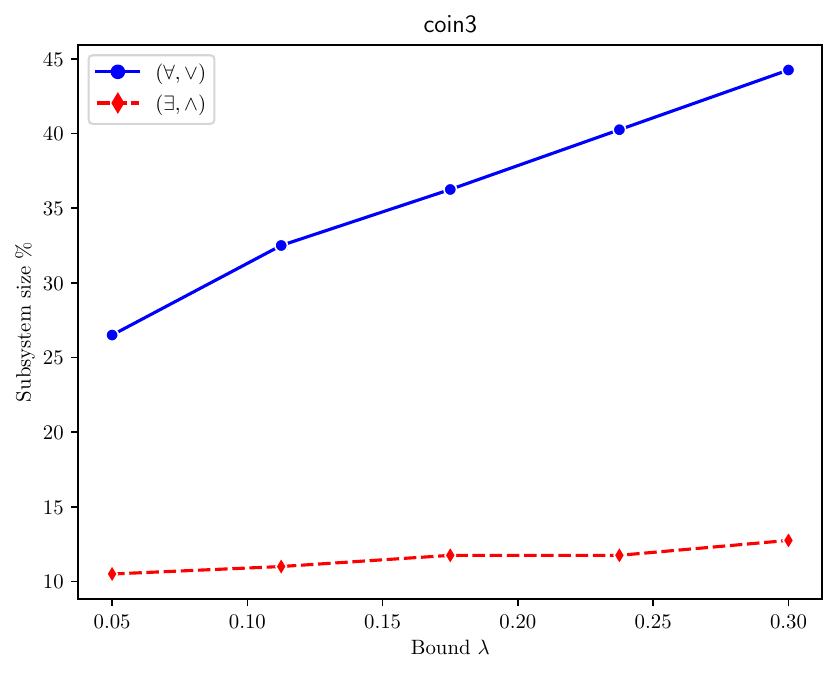}
\caption{\textsf{coin3}}
\end{subfigure}
\hfill
\begin{subfigure}{0.48\textwidth}
\centering
\includegraphics[scale=0.44]{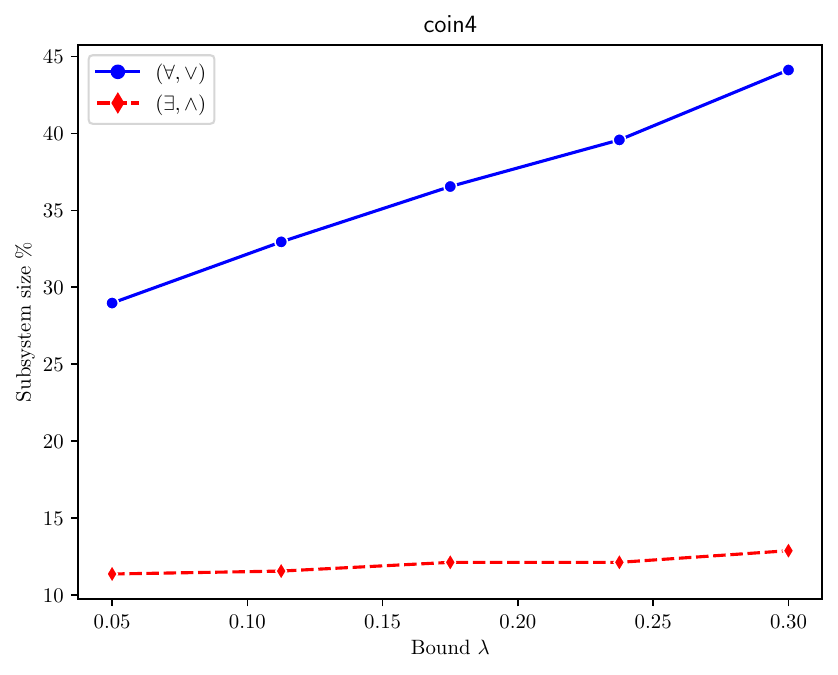}
\caption{\textsf{coin4}}
\end{subfigure}

\begin{subfigure}{1\textwidth}
\centering
\includegraphics[scale=0.44]{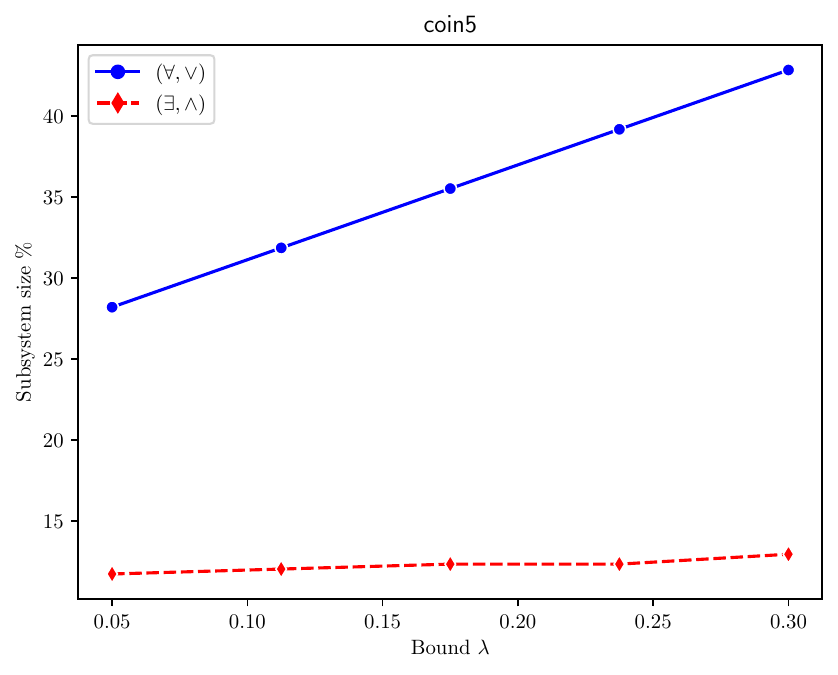}
\caption{\textsf{coin5}}
\end{subfigure}
\caption{Sizes of witnessing subsystems (relative to original MDP) for consensus.}
\label{fig:witnessing-subsystem-sizes-consensus}
\end{figure}

\bigskip

\noindent\textbf{Certification of dual queries.} In our experiments, we consider queries that are satisfied and, thus, for which certificates exist. We also investigate the time it takes for the solver to determine that no certificate exists for the dual query, e.g.\ for a satisfied \existsCQ-query $\query$ we measure the time it takes for the tool to conclude that no certificates exist for the \universalDQ-query $\neg \query$. The results are shown in \Cref{table:runtimes-verification}. The column \textsf{BuildDual} describes the time for building the model for the dual. The column \textsf{CertDual} describes the time for concluding that no certificate exists for the dual query. The column \textsf{CertTotal} describes the total time of \textsf{Cert} and \textsf{CertDual}. We observe that the time for determining that no certificate exists seems to be slightly faster than the time for computing the certificate.

\begin{table}[!h]
\centering
\setlength{\tabcolsep}{4pt}
\def\arraystretch{1}%
\scriptsize
\begin{tabular}{llllll|ccccc}
\toprule
 &  &  &  &  &  & \textsf{Build} & \textsf{BuildDual} & \textsf{Cert} & \textsf{CertDual} & \textsf{CertTotal} \\
Model & $\card{\states}$ & $\card{\SA}$ & Type & $k$ & \# & mean & mean & mean & mean & mean\\
\hline
\multirow[c]{2}{*}{\textsf{coin3}} & \multirow[c]{2}{*}{400} & \multirow[c]{2}{*}{592} & $(\exists,\land)$ & 2 & 5 & 0.250 & 0.086 & 0.012 & 0.008 & 0.021 \\
 &  &  & $(\forall,\lor)$ & 2 & 5 & 0.199 & 0.085 & 0.005 & 0.003 & 0.007 \\
\multirow[c]{2}{*}{\textsf{coin4}} & \multirow[c]{2}{*}{528} & \multirow[c]{2}{*}{784} & $(\exists,\land)$ & 2 & 5 & 0.347 & 0.112 & 0.024 & 0.015 & 0.039 \\
 &  &  & $(\forall,\lor)$ & 2 & 5 & 0.264 & 0.112 & 0.019 & 0.005 & 0.025 \\
\multirow[c]{2}{*}{\textsf{coin5}} & \multirow[c]{2}{*}{656} & \multirow[c]{2}{*}{976} & $(\exists,\land)$ & 2 & 5 & 0.424 & 0.144 & 0.017 & 0.021 & 0.038 \\
 &  &  & $(\forall,\lor)$ & 2 & 5 & 0.326 & 0.140 & 0.012 & 0.006 & 0.018 \\
 \hline
\multirow[c]{2}{*}{\textsf{csn3}} & \multirow[c]{2}{*}{410} & \multirow[c]{2}{*}{913} & $(\exists,\land)$ & 3 & 1 & 0.229 & 0.094 & 0.024 & 0.003 & 0.026 \\
 &  &  & $(\forall,\lor)$ & 3 & 1 & 0.158 & 0.092 & 0.024 & 0.002 & 0.026 \\
\multirow[c]{2}{*}{\textsf{csn4}} & \multirow[c]{2}{*}{2115} & \multirow[c]{2}{*}{5749} & $(\exists,\land)$ & 4 & 1 & 1.529 & 0.714 & 0.038 & 0.007 & 0.045 \\
 &  &  & $(\forall,\lor)$ & 4 & 1 & 0.944 & 0.701 & 0.029 & 0.042 & 0.071 \\
\multirow[c]{2}{*}{\textsf{csn5}} & \multirow[c]{2}{*}{10610} & \multirow[c]{2}{*}{33493} & $(\exists,\land)$ & 5 & 1 & 13.544 & 7.866 & 0.063 & 0.057 & 0.120 \\
 &  &  & $(\forall,\lor)$ & 5 & 1 & 8.859 & 7.793 & 0.058 & 0.078 & 0.137 \\
 \hline
\multirow[c]{2}{*}{\textsf{fire3}} & \multirow[c]{2}{*}{4093} & \multirow[c]{2}{*}{5519} & $(\exists,\land)$ & 2 & 5 & 2.454 & 0.782 & 0.033 & 0.034 & 0.068 \\
 &  &  & $(\forall,\lor)$ & 2 & 5 & 1.996 & 0.792 & 0.078 & 0.026 & 0.105 \\
\multirow[c]{2}{*}{\textsf{fire6}} & \multirow[c]{2}{*}{8618} & \multirow[c]{2}{*}{12948} & $(\exists,\land)$ & 2 & 5 & 6.045 & 1.724 & 0.044 & 0.061 & 0.105 \\
 &  &  & $(\forall,\lor)$ & 2 & 5 & 4.463 & 1.723 & 0.218 & 0.068 & 0.287 \\
\multirow[c]{2}{*}{\textsf{fire9}} & \multirow[c]{2}{*}{14727} & \multirow[c]{2}{*}{24229} & $(\exists,\land)$ & 2 & 5 & 12.353 & 3.262 & 0.073 & 0.105 & 0.178 \\
 &  &  & $(\forall,\lor)$ & 2 & 5 & 8.317 & 3.277 & 0.506 & 0.093 & 0.599 \\
\end{tabular}

\caption{Runtimes for concluding non-existence of dual certificates.}
\label{table:runtimes-verification}
\end{table}

\end{document}